\documentclass[a4paper,UKenglish,cleveref,autoref,thm-restate,colorlinks]{lipics-v2021}

\pdfoutput=1 %
\hideLIPIcs  %

\title{Perturbation equivalence in labelled Markov chains} %

\titlerunning{Perturbation equivalence in labelled Markov chains} %

\author{{Syyeda Zainab} Fatmi}{University of Oxford, United Kingdom}{}{https://orcid.org/0000-0001-7899-8665}{Clarendon Fund}%

\author{Stefan Kiefer}{University of Oxford, United Kingdom}{}{https://orcid.org/0000-0003-4173-6877}{}

\author{{James C.~A.} Main}{University of Oxford, United Kingdom}{}{https://orcid.org/0009-0000-8471-4833}{}

\author{David Parker}{University of Oxford, United Kingdom}{}{https://orcid.org/0000-0003-4137-8862}{}

\authorrunning{{S.~Z.} Fatmi, S.~Kiefer, {J.~C.~A.} Main, D. Parker} %

\Copyright{{Syyeda Zainab} Fatmi, Stefan Kiefer, {James C.~A.} Main, David Parker} %



\category{} %

\relatedversion{} %

\supplement{}%
\supplementdetails[linktext={UL/UB implementation}]{Source code}{https://github.com/zainabfatmi/prism/blob/bisimulation/prism/src/explicit/bisim/UniversalBisimulation.java}

\funding{\emph{Stefan Kiefer} and \emph{James C.~A.~Main}: EPSRC grant EP/Z536179/1}

\acknowledgements{}%

\nolinenumbers

\EventEditors{John Q. Open and Joan R. Access}
\EventNoEds{2}
\EventLongTitle{42nd Conference on Very Important Topics (CVIT 2016)}
\EventShortTitle{CVIT 2016}
\EventAcronym{CVIT}
\EventYear{2016}
\EventDate{December 24--27, 2016}
\EventLocation{Little Whinging, United Kingdom}
\EventLogo{}
\SeriesVolume{42}
\ArticleNo{23}

\usepackage[utf8]{inputenc}
\usepackage[T1]{fontenc}
\usepackage{amsmath, amssymb, amsthm, verbatim, mathtools, thm-restate}
\usepackage{IEEEtrantools}
\usepackage{mathrsfs, dsfont}
\usepackage[ruled, vlined]{algorithm2e}
\SetKw{KwGuess}{guess}
\SetKwComment{Comment}{$\rhd$ }{}
\SetCommentSty{}
\usepackage{breakcites}
\usepackage{multirow}

\usepackage{todonotes}

\usepackage{tikz}
\usetikzlibrary{automata,positioning,backgrounds}

\numberwithin{equation}{section}
\numberwithin{figure}{section}

\usepackage{pdflscape}
\usepackage{longtable}
\usepackage{booktabs}
\newcolumntype{R}[1]{>{\raggedleft\let\newline\\\arraybackslash\hspace{0pt}}m{#1}}
\newcolumntype{M}[1]{>{\centering\let\newline\\\arraybackslash\hspace{0pt}}m{#1}}
\newcolumntype{Y}{>{\centering\arraybackslash}X}
\usepackage[table]{xcolor}

\definecolor{custom-main}{rgb}{0.0, 0.0, 0.0}
\definecolor{custom-sub}{rgb}{0.0, 0.0, 0.0}
\definecolor{custom-purple}{rgb}{0.8, 0.1, 0.8} %
\definecolor{custom-amethyst}{rgb}{0.6, 0.4, 0.8}
\definecolor{custom-red}{rgb}{0.9, 0.1, 0.1}
\definecolor{custom-green}{rgb}{0.2, 0.7, 0.4}
\definecolor{custom-turquoise}{rgb}{0.3, 0.8, 0.8}

\tikzset{
  >=stealth,
  initial text=,
  left sided/.style={
    draw=none,
    append after command={
      [shorten <= -0.5\pgflinewidth]
      (\tikzlastnode.north west) edge[dashed](\tikzlastnode.south west)
    }
  },
  two sided/.style={
    draw=none,
    append after command={
      [shorten <= -0.5\pgflinewidth]
      (\tikzlastnode.north west) edge[dashed](\tikzlastnode.south west)
      (\tikzlastnode.north east) edge[dashed](\tikzlastnode.south east)
    }
  },
  right sided/.style={
    draw=none,
    append after command={
      [shorten <= -0.5\pgflinewidth]
      (\tikzlastnode.north east) edge[dashed](\tikzlastnode.south east)
    }
  },
  every state/.style={circle, minimum size=1cm},
  every path/.style={thick},
  initial text=,
  node distance=1cm,
    dotnode/.style={
    state,
    white,
    text=black,
    inner sep=0,
    align=center
  },
  emptynode/.style={
    minimum size=0,
    inner sep=0,
    outer sep=0
  },
    diagonal fill/.style 2 args={
  fill=#2, path picture={
    \fill[#1, sharp corners] (path picture bounding box.south west) -| %
    (path picture bounding box.north east) -- cycle;
    }
  },
  reversed diagonal fill/.style 2 args={
    fill=#2, path picture={
      \fill[#1, sharp corners] (path picture bounding box.north west) |- 
      (path picture bounding box.south east) -- cycle;
    }
  }
}

\tikzstyle{stochasticc} = [fill, circle, minimum size=0.1cm, inner sep=0.05cm, outer sep=0cm]
\tikzstyle{stochastics} = [fill, rectangle, minimum size=0.1cm, inner sep=0.05cm, outer sep=0cm]

\DeclareMathOperator*{\argmin}{argmin}

\newcommand{\init}{\mathsf{init}}

\newcommand{\IR}{\mathbb{R}}

\newcommand{\IN}{\mathbb{N}}
\newcommand{\IQ}{\mathbb{Q}}

\newcommand{\INpos}{\IN_{>0}}

\newcommand{\ccInt}[2]{\left[#1, #2\right]}
\newcommand{\ocInt}[2]{\left]#1, #2\right]}

\newcommand{\leLex}{\mathrel{\leq_{\mathsf{lex}}}}

\newcommand{\nc}{\textsf{NC}}
\newcommand{\logspace}{\textsf{L}}
\newcommand{\nlogspace}{\textsf{NL}}
\newcommand{\conlogspace}{\textsf{co-NL}}
\newcommand{\ptime}{\textsf{P}}

\newcommand{\np}{\textsf{NP}}
\newcommand{\coNP}{\textsf{co-NP}}

\newcommand{\ceql}{\ensuremath{\mathsf{C}_{=}\mathsf{L}}}
\newcommand{\tczero}{\ensuremath{\mathsf{TC}^{0}}}

\newcommand{\subsets}[1]{2^{#1}}

\newcommand{\dist}[1]{\mathcal{D}(#1)}
\newcommand{\measure}{\mu} %
\newcommand{\measureB}{\nu} %
\newcommand{\cyl}[1]{\mathsf{Cyl}\left(#1\right)}

\newcommand{\supp}[1]{\mathsf{supp}(#1)}

\newcommand{\event}{E}

\newcommand{\iPos}{j}
\newcommand{\iLast}{r}
\newcommand{\iSeq}{n} %

\newcommand{\word}{w}

\newcommand{\objective}{\event}

\newcommand{\spanVect}[1]{\mathsf{span}(#1)}

\newcommand{\oneVect}{\mathbf{1}}

\newcommand{\vectSpace}{V}

\newcommand{\mchain}{\mathcal{M}}
\newcommand{\trans}{\tau}
\newcommand{\transB}{\sigma}
\newcommand{\states}{S}
\newcommand{\state}{s}
\newcommand{\stateB}{t}
\newcommand{\stateC}{q}
\newcommand{\mchainTuple}{(\states, \trans)}

\newcommand{\lbl}{\ell}
\newcommand{\lblSet}{L}
\newcommand{\lmc}{(\states, \trans, \lbl, \lblSet)}

\newcommand{\diag}{\states^2_\Delta}

\newcommand{\probL}[1]{\mathbb{P}_{#1}}
\newcommand{\probPV}[2]{\probL{#1}^{#2}}
\newcommand{\probLV}[2]{\probL{#1}^{#2}}

\newcommand{\cons}[1]{\mathsf{cons}(#1)} %

\newcommand{\stable}{\mathsf{Stable}}

\newcommand{\langEquiv}[1]{\equiv_{#1}}
\newcommand{\bisim}[1]{\sim_{#1}}
\newcommand{\eClass}{C}
\newcommand{\eClassB}{D}
\newcommand{\eClassC}{E}
\newcommand{\eClassTwo}[3]{C_{#1, #2}}

\newcommand{\plays}[1]{\mathsf{Paths}(#1)}
\newcommand{\play}{\rho}
\newcommand{\hist}{h}

\newcommand{\mdp}{\mathcal{P}}
\newcommand{\actions}{A}
\newcommand{\action}{\mathsf{m}}

\newcommand{\strat}{\pi}

\newcommand{\ltrans}[2]{M_{#1}^{#2}}

\newcommand{\literal}{\lambda}
\newcommand{\lit}[1]{\literal_{#1}}
\newcommand{\litset}{\Lambda}
\newcommand{\posset}{P}
\newcommand{\negset}{N}
\newcommand{\neglit}[1]{\bar{#1}}

\newcommand\vartextvisiblespace[1][.5em]{%
  \makebox[#1]{%
    \kern.07em
    \vrule height.3ex
    \hrulefill
    \vrule height.3ex
    \kern.07em
  }%
}

\makeatletter
\def\squarecorner#1{
    \pgf@x=\the\wd\pgfnodeparttextbox%
    \pgfmathsetlength\pgf@xc{\pgfkeysvalueof{/pgf/inner xsep}}%
    \advance\pgf@x by 2\pgf@xc%
    \pgfmathsetlength\pgf@xb{\pgfkeysvalueof{/pgf/minimum width}}%
    \ifdim\pgf@x<\pgf@xb%
        \pgf@x=\pgf@xb%
    \fi%
    \pgf@y=\ht\pgfnodeparttextbox%
    \advance\pgf@y by\dp\pgfnodeparttextbox%
    \pgfmathsetlength\pgf@yc{\pgfkeysvalueof{/pgf/inner ysep}}%
    \advance\pgf@y by 2\pgf@yc%
    \pgfmathsetlength\pgf@yb{\pgfkeysvalueof{/pgf/minimum height}}%
    \ifdim\pgf@y<\pgf@yb%
        \pgf@y=\pgf@yb%
    \fi%
    \ifdim\pgf@x<\pgf@y%
        \pgf@x=\pgf@y%
    \else
        \pgf@y=\pgf@x%
    \fi
    \pgf@x=#1.5\pgf@x%
    \advance\pgf@x by.5\wd\pgfnodeparttextbox%
    \pgfmathsetlength\pgf@xa{\pgfkeysvalueof{/pgf/outer xsep}}%
    \advance\pgf@x by#1\pgf@xa%
    \pgf@y=#1.5\pgf@y%
    \advance\pgf@y by-.5\dp\pgfnodeparttextbox%
    \advance\pgf@y by.5\ht\pgfnodeparttextbox%
    \pgfmathsetlength\pgf@ya{\pgfkeysvalueof{/pgf/outer ysep}}%
    \advance\pgf@y by#1\pgf@ya%
}
\makeatother

\pgfdeclareshape{square}{
    \savedanchor\northeast{\squarecorner{}}
    \savedanchor\southwest{\squarecorner{-}}

    \foreach \x in {east,west} \foreach \y in {north,mid,base,south} {
        \inheritanchor[from=rectangle]{\y\space\x}
    }
    \foreach \x in {east,west,north,mid,base,south,center,text} {
        \inheritanchor[from=rectangle]{\x}
    }
    \inheritanchorborder[from=rectangle]
    \inheritbackgroundpath[from=rectangle]
}

\begin{document}

\maketitle

\begin{abstract}
Behavioural equivalences such as language equivalence and probabilistic bisimilarity are fundamental techniques for reducing the size of probabilistic models prior to verification.
However, these equivalences are sensitive to the precise values of transition probabilities, making them unsuitable in applications where probabilities are obtained from measurements or subject to approximation.
Motivated by settings in which the support graph of a labelled Markov chain is known but the transition probabilities are uncertain, we study robust variants of these equivalences.
We introduce universal perturbation equivalence, which captures a variant of equivalence that is resilient to all perturbations of transition probabilities: two states or distributions are universally (perturbation) equivalent if they remain equivalent under every assignment of transition probabilities consistent with the support graph.
We also consider the dual notion of existential (perturbation) equivalence, which holds whenever there exists an assignment of transition probabilities that yields equivalence.

We establish that, for states, universal language equivalence coincides with universal probabilistic bisimilarity and develop a characterisation that yields a polynomial-time partition refinement algorithm.
We implement the algorithm and demonstrate experimentally that it is effective as a technique for robust model reduction.
We further show that universal language equivalence for distributions is closely related to the state case, and prove $\nlogspace$-completeness of deciding universal equivalence for both states and distributions.
Turning to existential equivalence, we prove that, for states, existential language equivalence coincides with existential probabilistic bisimilarity and deterministic witness transition functions always suffice, leading to an $\np$-completeness result.
In contrast, we show that existential language equivalence for distributions is $\exists\mathbb{R}$-complete.
Together, these results provide a complete complexity landscape for perturbation behavioural equivalence in labelled Markov chains.
\end{abstract}

\newpage

\section{Introduction}
\subparagraph*{Probabilistic models.}
\emph{Labelled Markov chains} (LMCs) are a fundamental model for probabilistic systems in verification, with states labelled with atomic propositions, which capture known properties or facts about the state, and transitions governed by probabilities that reflect the stochastic dynamics of the system.
\emph{Labelled Markov decision processes} (MDPs) generalise labelled Markov chains with the addition of non-determinism.

A major challenge in the verification of such systems is the state-space explosion problem, which often renders model checking computationally expensive or even infeasible.
One of the techniques used to mitigate this problem is to identify and merge system states that are behaviourally indistinguishable \cite{BK08}, thereby reducing the time required to verify properties of the system \cite{DBLP:conf/tacas/KatoenKZJ07,AKJB25}.
We study two prominent notions of behavioural equivalence.

\subparagraph*{Behavioural equivalence.}
\emph{Language equivalence}, also known as \emph{trace equivalence}, for labelled Markov chains compares two states or distributions over states by asking whether they induce the same probability for every measurable set of traces over the label alphabet. Intuitively, two states (or distributions) are considered language-equivalent if their observable executions are indistinguishable.
The origins of this language-based notion of equivalence can be traced back to work in automata theory by Sch{\"{u}}tzenberger \cite{DBLP:journals/iandc/Schutzenberger61b} and Paz \cite{Paz71}.
In practice and theory, language equivalence is attractive because it can be decided in polynomial time \cite{DBLP:journals/iandc/Schutzenberger61b, DBLP:journals/siamcomp/Tzeng92, DBLP:journals/ijfcs/DoyenHR08} and in the lower complexity classes \nc~\cite{DBLP:journals/ipl/Tzeng96} and \ceql~\cite{DBLP:conf/stacs/CernyS26}.

\emph{Probabilistic bisimilarity}, or \emph{bisimilarity} for short, is one of the most widely studied notions of behavioural equivalence for probabilistic models, introduced by Larsen and Skou \cite{DBLP:conf/popl/LarsenS89}, who adapted the notion of \emph{lumpability} \cite{KS60} to the setting of labelled Markov chains. Compared to language equivalence, bisimilarity is a finer equivalence relation that captures not only a system's observable behaviour, but also its internal probabilistic branching structure. Intuitively, two states are considered bisimilar if they have the same labels and they transition into each equivalence class with the same probabilities.
Bisimilarity can also be computed in
polynomial time \cite{DBLP:conf/cav/Baier96, DBLP:conf/fossacs/ChenBW12} and is implemented in state-of-the-art probabilistic verification tools such as PRISM \cite{DBLP:conf/cav/KwiatkowskaNP11} and Storm \cite{DBLP:journals/sttt/HenselJKQV22}.

\subparagraph*{Robustness.}
Behavioural equivalences are notoriously sensitive to the exact values of the transition probabilities. In fact, in \cite{DBLP:conf/qest/JaegerMLM14,DBLP:conf/fsttcs/Kiefer021,DBLP:conf/cav/FatmiKPB25}, it was observed that small changes to these probabilities can drastically change the behaviour of states. This motivates the need for more robust notions of equivalences when probabilities arise from measurements or approximations.

In practice, transition probabilities are often uncertain even when the underlying graph is known. \emph{Interval} and \emph{parametric Markov chains} model such uncertainty explicitly using probability intervals \cite{DBLP:conf/lics/JonssonL91} and expressions over parameters \cite{DBLP:conf/ictac/Daws04}, respectively. In this paper, we take a deliberately qualitative viewpoint, assuming that the graph of possible transitions is known while transition probabilities are unknown or subject to change. This corresponds to an \emph{incomplete Markov chain} \cite{DBLP:conf/tacas/BenediktLW13}, a special case of an interval Markov chain where allowed transitions range over $[0,1]$ and disallowed edges have the interval $[0,0]$.

\subparagraph*{Perturbation equivalence.}
We study the influence of perturbations on behavioural equivalences through concepts of \emph{perturbation equivalence}.
A robust form of equivalence can be obtained by requiring that states remain equivalent under \emph{all perturbations} of the transition function: we say that two states are \emph{universally perturbation language equivalent} (UL) if they are language-equivalent for all probability assignments for (existing) transitions of the Markov chain.
For $\varepsilon>0$, we also consider $\varepsilon$-UL, which requires pairs of states to be language-equivalent under every perturbed transition function in the $\varepsilon$-neighbourhood of the original transition function. We show that, for every $\varepsilon>0$, UL and $\varepsilon$-UL coincide. Thus, perhaps surprisingly, robustness to arbitrarily small perturbations implies robustness to all perturbations.
We also investigate the dual problem of the existence of a perturbation that makes states language-equivalent: we say that two states are \emph{existentially perturbation language equivalent} (EL) if there exists a valuation of the transition probabilities under which the states are equivalent.
We further consider perturbation language equivalence for distributions over states.
Analogous variants are defined for bisimilarity, yielding \emph{universal perturbation bisimilarity} (UB), $\varepsilon$-UB and \emph{existential perturbation bisimilarity} (EB).

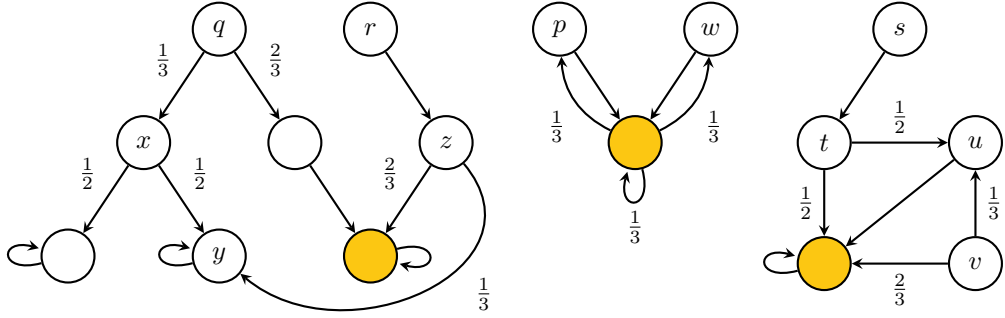
\begin{figure}[ht]
  \centering
  \begin{tikzpicture}[every state/.style={minimum size=0.7cm}]
    \node[state] at (5.5,3) (lml) {$p$};
    \node[state] at (7.5,3) (lmr) {$w$};
    \node[state,fill=lipicsYellow] at (6.5,1.5) (lb) {};
    \path[-stealth] (lml) edge (lb);
    \path[-stealth] (lmr) edge (lb);
    \path[-stealth] (lb) edge[bend right] node[below right] {$\frac{1}{3}$} (lmr);
    \path[-stealth] (lb) edge[bend left] node[below left] {$\frac{1}{3}$} (lml);
    \path[-stealth] (lb) edge[loop below] node[below] {$\frac{1}{3}$} (lb);
    \node[state] at (1,3) (q) {$q$};
    \node[state] at (3,3) (r) {$r$};
    \node[state] at (0,1.5) (x) {$x$};
    \node[state] at (2,1.5) (a) {};
    \node[state] at (4,1.5) (b) {$z$};
    \node[state] at (1,0) (y) {$y$};
    \node[state] at (-1,0) (d) {};
    \node[state, fill=lipicsYellow] at (3,0) (c) {};
    \path[-stealth] (q) edge node[above left] {$\frac{1}{3}$} (x);
    \path[-stealth] (q) edge node[above right] {$\frac{2}{3}$} (a);
    \path[-stealth] (r) edge (b);
    \path[-stealth] (x) edge node[above right] {$\frac{1}{2}$} (y);
    \path[-stealth] (x) edge node[above left] {$\frac{1}{2}$} (d);
    \path[-stealth] (d) edge[loop left] (d);
    \path[-stealth] (a) edge (c);
    \path[-stealth] (b) edge node[above left] {$\frac{2}{3}$} (c);
    \path[-stealth] (y) edge[loop left] (y);
    \path[-stealth] (c) edge[loop right] (c);
    \path[-stealth,overlay] (b) edge[out=310,in=320,looseness=1.4] (y);
    \node at (4.5,-0.5) (z) {$\frac{1}{3}$};
    \node[state] at (10,3) (s) {$s$};
    \node[state] at (9,1.5) (t) {$t$};
    \node[state] at (11,1.5) (u) {$u$};
    \node[state] at (11,-0.1) (v) {$v$};
    \node[state, fill=lipicsYellow] at (9,-0.1) (w) {};
    \path[-stealth] (s) edge (t);
    \path[-stealth] (t) edge node[above] {$\frac{1}{2}$} (u);
    \path[-stealth] (t) edge node[left] {$\frac{1}{2}$} (w);
    \path[-stealth] (u) edge (w);
    \path[-stealth] (w) edge[loop left] (w);
    \path[-stealth] (v) edge node[right] {$\frac{1}{3}$} (u);
    \path[-stealth] (v) edge node[below] {$\frac{2}{3}$} (w);
  \end{tikzpicture}
  \caption{A labelled Markov chain to illustrate the concepts of perturbation equivalence. All omitted transition probabilities are 1. Colours indicate the labelling of states.}
  \label{figure:intro-example}
\end{figure}

Consider the labelled Markov chain in \cref{figure:intro-example}.
The states $p$ and $w$ illustrate universal equivalence. They are both UL and UB, meaning that they remain language-equivalent and bisimilar for all perturbations of the transition function, provided that no new transitions are added. %
Similarly, the states $x$ and $y$ are both UL and UB.
While the states $q$ and $r$ are language-equivalent, hence EL, they are not UL.
This can be witnessed by any perturbation of the outgoing probabilities of $q$.
Although $q$ and $r$ are not bisimilar, they are EB as they can be made bisimilar, e.g., by assigning probability~$0$ to the transitions from $q$ to $x$ and from $z$ to $y$ (and assigning probability~$1$ to the other outgoing transitions of $q$ and $z$).
The states $t$ and $u$ are neither language-equivalent nor bisimilar under the given transition probabilities. Nevertheless, they are both EL and EB, since assigning probability~$0$ to the transition from $t$ to $u$ makes them equivalent in both senses.
Similarly, suitable perturbations of the transition probabilities make the states pairs $(s, v)$ and $(t, v)$ both language-equivalent and bisimilar, so they are also EL and EB.
In contrast, no variation of the transition probabilities can make $s$ and $t$ equivalent and, thus, they are neither EL nor EB.
Therefore, neither EL nor EB are equivalence relations.
This observation has important algorithmic consequences.  Efficient practical approaches, such as partition refinement, rely on the target relation being an equivalence relation. As a result, these techniques cannot be applied directly for EL and EB.

\subparagraph*{Contributions}
Our contributions can be summarised as follows.
First, we prove that, for states, UL coincides with UB (\cref{corollary:ule:upb}) and we develop a structural characterisation (\cref{theorem:ule-states}) that underpins an efficient decision procedure, establishing that the problem is $\nlogspace$-complete.
Second, we show that, for any $\varepsilon > 0$, UL, $\varepsilon$-UL and $\varepsilon$-UB are equivalent for states (\cref{proposition:ule:epsilon}), hence, robustness to arbitrarily small perturbations guarantees robustness to all perturbations.
Third, we implement a polynomial-time partition refinement algorithm for computing these universal equivalence relations and report experimental evidence demonstrating that it is efficient and useful in practice for reducing the size of the state space before verification under uncertain transition probabilities.
Fourth, we show that two distributions are UL if and only if they assign the same probability to every equivalence class induced by UL for states (\cref{theorem:ule-distributions}) and, consequently, UL for distributions is also $\nlogspace$-complete.
Fifth, we establish that EL for states coincides with EB (\cref{thm:ele-epb}) and show that witness transition functions for these existential equivalence notions can always be chosen to be deterministic (\cref{thm:ele-det}), which allows us to prove that these decision problems are $\np$-complete.
Finally, we analyse the EL problem for distributions and prove that it is $\exists\IR$-complete,\footnote{The hardness of EL for distributions is established by a reduction from the non-negative matrix factorisation (NMF) problem (cf.\ Section~\ref{appendix:ele:nmf}). The NMF problem is \np-hard~\cite{DBLP:journals/siamjo/Vavasis09} and shown to be complete for the existential theory of the reals ($\exists\IR$-complete) in an unpublished paper~\cite{shitov18NMF,shitov21NMF}. We assume $\exists\IR$-completeness of the NMF problem in the sequel.} revealing a sharp increase in computational complexity compared to the universal case.
\cref{table:complexity} reiterates our complexity results along with theorem references.

\begin{table}[bt]
  \caption{Summary of complexity results.
    We use the subscripts $S$ and $\mathcal{D}$ to denote the variants of the language equivalence problems that study the equivalence of states and distributions respectively.
  }\label{table:complexity}
  \centering
    \begin{tabular}{ M{0.35\textwidth} M{0.3\textwidth} M{0.26\textwidth} }
      \toprule
      UB, UL$_S$, and UL$_\mathcal{D}$ & EB and EL$_S$ & EL$_\mathcal{D}$ \\
            \nlogspace-c. (Thms.\ \ref{theorem:ule:nl-states} and \ref{theorem:ule:d:complexity})
            & \np-c. (Thm.\ \ref{theorem:complexity:ele-states})
            & $\exists\IR$-c. (Thm.\ \ref{theorem:nmf-ele})
      \\
      \bottomrule
  \end{tabular}
\end{table}

\subparagraph*{Related work.}
Kiefer and Tang \cite{DBLP:conf/fsttcs/Kiefer020} study whether there exist memoryless strategies under which two labelled MDPs are equivalent, or inequivalent, with respect to language equivalence or bisimilarity.
Varying the transition probabilities of an LMC while preserving its support can be viewed as resolving the choices of a labelled deterministic MDP. While this gives a reduction of our perturbation equivalence problems to their setting, it does not yield our finer complexity bounds or structural results for LMCs (cf.\ Appendix~\ref{appendix:mdps}). %

Our work is also related to the study of behavioural equivalences for probabilistic models with uncertain transition probabilities. Hashemi et al.\ investigate bisimulation and minimisation for interval MDPs \cite{DBLP:journals/corr/HashemiHK14,DBLP:conf/lata/HashemiH0STW16}.
They show that minimisation for interval MDPs is $\coNP$-complete in general, but becomes solvable in polynomial time when the branching degree is bounded by a constant. Unlike our setting, they retain explicit probability intervals.

A different notion of bisimulation for LMCs, called \emph{robust (probabilistic) bisimilarity}, was introduced in \cite{DBLP:conf/cav/FatmiKPB25}.
It characterises continuity of the \emph{probabilistic bisimilarity distance} (see \cite{DBLP:conf/concur/DesharnaisGJP99} for a definition) for bisimilar state pairs \cite{concur}.
We briefly discuss the relationship between UB and robust bisimilarity.
While both equivalence relations are a refinement of bisimilarity, they capture different notions of robustness and are, therefore, not directly comparable.
Intuitively, UB requires two states to remain bisimilar under every structure-preserving perturbation of the transition probabilities. In contrast, loosely speaking, robust bisimilarity requires two states to remain behaviourally similar under all perturbations of the transition probabilities, including those that add transitions. We illustrate this distinction in Appendix~\ref{appendix:robust}.

\subparagraph*{Outline.}
The remainder of the paper is organised as follows. \cref{section:preliminaries} recalls basic definitions. \cref{section:problems} formalises universal and existential perturbation language equivalence and probabilistic bisimilarity. \cref{section:ue} studies universal equivalence, developing a fixed-point characterisation of UL for states, its coincidence with UB, complexity and algorithmic results, along with an experimental evaluation. \cref{section:ee} studies existential equivalence, establishing the coincidence of EB and EL for states, via the existence of deterministic witness transition functions, and complexity results. \cref{section:conclusion} concludes and outlines directions for future work. Omitted proofs appear in the appendix.

\section{Preliminaries}
\label{section:preliminaries}

\subparagraph*{Probability.} Let $A$ be a countable set.
We write $\dist{A}$ for the set of distributions over $A$, i.e., the set of functions $\measure\colon A\to\ccInt{0}{1}$ such that $\sum_{a\in A}\measure(a) = 1$.
The support of a distribution $\measure\in\dist{A}$ is $\supp{\measure} = \{a\in A\mid \measure(a)> 0\}$.

\subparagraph*{Markov chains.}
A (finite) \textit{labelled Markov chain} is a tuple $\mchain = \lmc$ where $\states$ is a finite set of states, $\trans\colon\states\to\dist{\states}$ is a probabilistic transition function, $\lblSet$ is a finite set of labels and $\lbl\colon\states\to\lblSet$ is a labelling function.
For all $\state,\stateB\in\states$, we write $\trans(\state, \stateB)$ for $\trans(\state)(\stateB)$.
For all $\state\in\states$, the elements in $\supp{\trans(\state)}$ are the \textit{successors} of $\state$.
The transition function $\trans$ is \textit{deterministic} if, for all $\state\in\states$, $\trans(\state)$ is a Dirac distribution, i.e., each state has one successor.

A \textit{path} of $\mchain$ is a sequence $\state_0\state_1\state_2\ldots\in\states^\omega$ such that for all $\iPos\in\IN$, $\trans(\state_{\iPos}, \state_{\iPos+1}) > 0$, and a \textit{history} of $\mchain$ is a finite prefix of a path.
We let $\plays{\mchain}$ denote the set of paths of $\mchain$.
We extend the labelling function from states to histories and paths in the usual way, i.e., for all paths and histories $\state_0\state_1\ldots$, we let $\lbl(\state_0\state_1\ldots) = \lbl(\state_0)\lbl(\state_1)\ldots$.

Given an initial distribution $\measure_\init\in\dist{\states}$, $\mchain$ induces a distribution $\probPV{\measure_\init}{\trans}$ over $\plays{\mchain}$ (with its usual sigma-algebra) in the usual way.
We include the transition function $\trans$ in the notation $\probPV{\measure_\init}{\trans}$ as, in the sequel, we will vary the transition function of the Markov chain while leaving the state space and labelling untouched.
For any measurable $\event\subseteq\lblSet^\omega$, we abuse notation and write $\probPV{\measure_\init}{\trans}(\event)$ for 
\(
  \probPV{\measure_\init}{\trans}\left(\{\play\in\plays{\mchain}\mid\lbl(\play)\in\event\}\right).
\)
For any finite word $\word\in\lblSet^*$, we write $\probLV{\measure_\init}{\trans}(\word)$ for the probability $\probLV{\measure_\init}{\trans}(\word\lblSet^\omega)$ of the infinite continuations of $\word$.

\subparagraph*{Language-equivalence and probabilistic bisimilarity.}
Let $\measure$ and $\measureB\in\dist{\states}$ be initial distributions of $\mchain$.
We say that $\measure$ and $\measureB$ are \textit{$\trans$-language-equivalent}, denoted $\measure\langEquiv{\trans}\measureB$, if for all measurable $\objective\subseteq\lblSet^\omega$, $\probLV{\measure}{\trans}(\objective) = \probLV{\measureB}{\trans}(\objective)$.
We note that $\measure\langEquiv{\trans}\measureB$ if and only for all $\word\in\lblSet^*$, $\probLV{\measure}{\trans}(\word) = \probLV{\measureB}{\trans}(\word)$.
Two states $\state, \stateB\in\states$ are $\trans$-language-equivalent, denoted $\state\langEquiv{\trans}\stateB$, if their associated Dirac distributions are language-equivalent.

An equivalence relation $R\subseteq\states\times\states$ is a \emph{(probabilistic) bisimulation} with respect to $\trans$, if for all $(\state, \stateB)\in R$, $\lbl(\state) = \lbl(\stateB)$ and, for all equivalence classes $\eClass\in\states/R$, $\trans(\state, \eClass) = \trans(\stateB, \eClass)$.
Two states $\state, \stateB\in\states$ are $\trans$-bisimilar, denoted by $\state\bisim{\trans}\stateB$, if there exists a bisimulation $R$ such that $(\state, \stateB)\in R$.
The relation $\bisim{\trans}$ is the coarsest bisimulation on $\mchain$ and is called \textit{$\trans$-bisimilarity}.

\section{Perturbation equivalence problems}
\label{section:problems}
Let $\mchain = \lmc$ be an LMC.
We consider two types of problems for language-equivalence and bisimilarity: the \textit{universal perturbation equivalence problem} and the \textit{existential perturbation equivalence problem}.
Intuitively, the universal perturbation equivalence problems for these relations ask if two states are equivalent under all transition functions and the existential perturbation equivalence problems ask if there exists some transition function that makes the states equivalent.

If we quantify over all transition functions $\transB\colon\states\to\dist{\states}$ with no conditions, both of these problems are (almost) trivial.
For the universal problems, if two labels appear in $\mchain$, then any two distinct states $\state, \stateB$ can be made inequivalent by directing all the probability from $\state$ to itself and all probability from $\stateB$ to some state $\stateC\in\states$ with $\lbl(\stateC)\neq\lbl(\state)$.
On the other hand, for the existential problems, two states can be made equivalent if and only if they share the same label: this is witnessed by any transition function that gives both states a self-loop.

For these reasons, we limit our attention to transitions functions that are consistent with the transition function $\trans$ of $\mchain$.
A transition function is consistent with $\trans$ if it can be obtained from $\trans$ by changing transition probabilities without introducing new transitions.
Formally, a transition function $\transB\colon\states\to\dist{\states}$ is \textit{consistent with $\trans$} if for all $\state\in\states$,
\(\supp{\transB(\state)}\subseteq\supp{\trans(\state)}.\)
We let $\cons{\trans}$ denote the set of transition functions that are consistent with $\trans$.

We now formalise the four main notions of perturbation equivalence we consider.
The first two are based on language-equivalence.
\begin{definition}
  Let $\measure, \measureB\in\dist{\states}$.
  We say that $\measure$ and $\measureB$ are \emph{universally perturbation language-equivalent} (UL) if, for all $\transB\in\cons{\trans}$, we have $\measure\langEquiv{\transB}\measureB$. We say that $\measure$ and $\measureB$ are \emph{existentially perturbation language-equivalent} (EL) if there exists $\transB\in\cons{\trans}$ such that $\measure\langEquiv{\transB}\measureB$.
  Two states are UL (resp.~EL) whenever their associated Dirac distributions are UL (resp.~EL).
\end{definition}

We now provide the analogous definitions for bisimilarity.
\begin{definition}
  Let $\state, \stateB\in\states$.
  We say that $\state$ and $\stateB$ are \emph{universally perturbation bisimilar} (UB) if, for all $\transB\in\cons{\trans}$, we have $\state\bisim{\transB}\stateB$. We say that $\state$ and $\stateB$ are \emph{existentially perturbation bisimilar} (EB) if there exists $\transB\in\cons{\trans}$ such that $\state\bisim{\transB}\stateB$.
\end{definition}

In the following, we refer to the problem of deciding universal perturbation language equivalence as the UL problem.
We define the EL, UB and EB problems similarly.
We therefore use the acronyms UL, EL, UB, and EB to denote both the corresponding relations and decision problems, with the intended interpretation determined by the context.

\section{Universal equivalence}
\label{section:ue}
This section is concerned with the universal perturbation equivalence relations.
In \cref{section:ule:states}, we focus on the UL and UB relations for states, develop efficient algorithms to compute these relations and demonstrate their efficiency in practice for state-space reduction.
We study the UL relation for distributions in \cref{section:ule:dist}.
Technical details are deferred to Appendix~\ref{appendix:ue}.

We fix an LMC $\mchain = \lmc$ for the remainder of the section.

\subsection{Universal equivalence for states}\label{section:ule:states}
We first provide a fixed-point characterisation of the UL relation over states and use it to show that the UL relation over states is exactly the UB relation.
We then provide a polynomial-time partition refinement algorithm to compute the UL relation and show that deciding UL for states is $\nlogspace$-complete.
Finally, we implement our partition refinement algorithm to compute the UL and UB equivalence relations and report experimental evidence demonstrating that it is efficient and useful in practice for state-space reduction prior to probabilistic model checking, when the transition probabilities may be uncertain, obtained from measurements, or subject to approximation.

\subparagraph*{Fixed point characterisation of universal language equivalence.} %
We show that the restriction of the UL relation to states can be computed as the greatest fixed point of some operator $\subsets{\states\times\states}\to\subsets{\states\times\states}$.
We let $U = \{(\state, \stateB)\in\states^2\mid \state \text{ and } \stateB\text{ are UL}\}$ denote this restriction.

We first introduce some notation to define the fixed point operator.
We let $\diag = \{\, (s, s) \mid s \in S \,\}$ and $\mathord{\sim_0} = \{(s, t) \in S \times S \mid \ell(s) = \ell(t)\}$ respectively denote the equality and label-equality relations on $\states$, and let $\ccInt{\diag}{\mathord{\sim_0}} = \{R \subseteq S \times S \mid \diag\subseteq R \subseteq\mathord{\sim_0}\}$.
We consider the operator $\stable\colon \ccInt{\diag}{\mathord{\sim_0}}\to \ccInt{\diag}{\mathord{\sim_0}}$ defined, for all $R\in\ccInt{\diag}{\mathord{\sim_0}}$, by
\[
  \stable(R) = \diag \cup \{\, (s, t) \in R \mid \supp{\trans(s)} \times \supp{\trans(t)} \subseteq R\}.
\]

The operator $\stable$ is monotone under set inclusion and, thus, has a greatest fixed point by the Knaster-Tarski fixed point theorem \cite[Theorem~2.35]{DBLP:books/daglib/0023601}.
Our goal is to show that this greatest fixed point is $U$.
We establish this by reasoning on universal perturbation language equivalence over words of fixed length, which is formally defined, for all $n\in\INpos$, as the relation
  $U_n = \left\{(\state, \stateB)\in\states\times\states\mid
  \forall\word\in\lblSet^n,\,\forall\transB\in\cons{\trans},\,
  \probLV{\state}{\transB}(\word) = \probLV{\stateB}{\transB}(\word)
  \right\}$.
By definition, we have $U_1 = \mathord{\sim_0}$ and $U = \bigcap_{n\in\INpos}U_n$.
To show that $U$ is the greatest fixed point of $\stable$, it suffices to prove that for all $n\in\INpos$, $U_{n} = \stable^{n-1}(U_1)$.
We sketch a proof by induction.

\begin{restatable}{proposition}{theoremULStates}
  \label{theorem:ule-states}
  For all $n\in\INpos$, we have $U_{n} = \stable^{n-1}(U_1)$
  Thus, $U= \bigcap_{n\in\INpos}U_n$ is the greatest fixed point of $\stable$.
\end{restatable}
\begin{proof}[Proof sketch]
  The proof is by induction on $n$.
The base case is direct.
We then assume by induction that $U_{n} = \stable^{n-1}(U_1)$ and must show that $U_{n+1} = \stable(U_n)$.
Let $\state, \stateB\in\states$ such that $\state\neq\stateB$ (pairs of the form $(\state, \state)$ are trivially in $U_n$ and $\stable(U_n)$ for all $n\geq 1$).

Suppose that $(\state, \stateB)\in\stable(U_n)$. By the definition of $\stable$, the labels of $\state$ and $\stateB$ coincide and all of their successors agree over words of length $n$ under all consistent transition functions.
It follows that $\state$ and $\stateB$ agree over words of length $n+1$ under all consistent transition functions.

For the other case, suppose that $(\state, \stateB)\notin\stable(U_{n})$.
If $(\state, \stateB)\notin U_n$, we directly obtain $(\state, \stateB)\notin U_{n+1}$.
We therefore assume that $(\state, \stateB)\in U_n$.
We thus have $\lbl(\state) = \lbl(\stateB) = a$ (because $U_n\subseteq U_1$ and $U_1$ is label equality).
By the definition of $\stable$, there exist $\state'\in\supp{\trans(\state)}$ and $\stateB'\in\supp{\trans(\stateB')}$ such that $(\state', \stateB')\notin U_n$.
We consider a transition function $\transB\in\cons{\trans}$ such that $\state'$ and $\stateB'$ do not agree over some $\word\in\lblSet^{n}$ under $\transB$.
If $\state$ and $\stateB$ do not agree on $a\word$ under $\transB$, this shows that $(\state, \stateB)\notin U_{n+1}$.
Otherwise, $\state$ has a successor $\state''$ such that $\probLV{\state''}{\transB}(\word) \neq \probLV{\state'}{\transB}(\word)$ or $\stateB$ has a successor $\stateB''$ such that $\probLV{\stateB''}{\transB}(\word) \neq \probLV{\stateB'}{\transB}(\word)$.
Assuming that $\state$ has such a successor $\state''$, we consider a perturbation of $\transB$ obtained by shifting some outgoing probability in $\transB$ from $\state$ between $\state''$ and $\state'$.
We show that this changes the probability of $a\word$ only from $\state$, due to $\state$ and $\stateB$ agreeing over words of length $n$ for all consistent transition functions (see Lemma~\ref{lemma:all-shorter-agree} in Appendix~\ref{appendix:ue} for a formal justification).
This shows that $(\state, \stateB)\notin U_{n+1}$.
\end{proof}

\cref{theorem:ule-states} suggests a polynomial-time partition refinement algorithm (cf.~\cref{algorithm:stable}) to compute $U$.
\cref{theorem:ule-states} also implies the equivalence of UL and UB for states, in contrast to standard language equivalence and probabilistic bisimilarity \cite[Figure~1]{DBLP:journals/ijfcs/DoyenHR08}.
\begin{theorem}\label{corollary:ule:upb}
  Two states are UL if and only if they are UB.
\end{theorem}
\begin{proof}
  Bisimilarity implies language equivalence, hence, UB states are UL.
  Conversely, the definition of $\stable$ and \cref{theorem:ule-states} imply that, for all $\transB\in\cons{\trans}$, $U$ is a bisimulation with respect to $\transB$.
  Thus, UL states are also UB.
\end{proof}

We now consider the robustness of language equivalence and bisimilarity to small perturbations.
Given $\varepsilon > 0$, we say that two states are $\varepsilon$-UL (resp.~$\varepsilon$-UB) if they remain language-equivalent (resp. bisimilar) under any $\transB\in\cons{\trans}$ such that any probability of $\transB$ differs from its counterpart in $\trans$ by at most $\varepsilon$.
In the proof of \cref{theorem:ule-states}, to witness that two states are not UL, we use a transition function that is consistent with $\trans$. However, its transition probabilities may differ substantially from those of $\trans$.
We can strengthen the proof of \cref{theorem:ule-states} to show that arbitrarily small perturbations of $\trans$ (i.e., consistent transition functions that are very similar to $\trans$) suffice to witness that two states are not UL.
The key observation is that, in the proof of \cref{theorem:ule-states}, we perturb (at most) one transition of a consistent transition function obtained from the inductive hypothesis. %
Whenever such a perturbation is necessary, any non-zero perturbation is sufficient for the argument.
This observation can be used to establish the following additional characterisations of UL and UB.

\begin{restatable}{proposition}{propositionEpsilonU}\label{proposition:ule:epsilon}
  For all $\varepsilon > 0$, two states are UL if and only if they are $\varepsilon$-UL (resp.~$\varepsilon$-UB).
\end{restatable}

\begin{remark}
  \cref{proposition:ule:epsilon} implies that to witness that two states are not UL, it is sufficient to change the transition probabilities of $\trans$ without disabling any of its transitions.
  In particular, two states $\state$ and $\stateB$ are UL if and only if $\state\langEquiv{\transB}\stateB$ holds for all $\transB\in\cons{\trans}$ such that, for all $\stateC\in\states$, we have $\supp{\transB(\stateC)} = \supp{\trans(\stateC)}$.
  \hfill$\lhd$
\end{remark}

By adapting the reasoning used to establish \cref{proposition:ule:epsilon}, we can also show that deterministic transition functions suffice to witness that states are not UL.
\begin{restatable}{proposition}{propositionDeterministicU}\label{proposition:ule:deterministic}
  If two states $\state$ and $\stateB$ are not UL, then there exist a deterministic $\transB\in\cons{\trans}$ and a word $\word\in\lblSet^*$ such that $\probLV{\state}{\transB}(\word) = 1$ and $\probLV{\stateB}{\transB}(\word) = 0$. %
\end{restatable}

\subparagraph*{Algorithms and complexity.}%

The characterisation from \cref{theorem:ule-states} suggests a polynomial-time algorithm to compute the UL relation based on computing the greatest fixed point of the $\stable$ operator.
We outline this approach in \cref{algorithm:stable}.
We report experimental results for this approach below. %
In the remainder of the section, we provide tight complexity bounds on the problem of deciding if two states are UL: we show that the problem is $\nlogspace$-complete.

\begin{algorithm}[t]
\caption{Algorithm to compute UL, the greatest fixed point of $\stable$.}
\label{algorithm:stable}
\KwData{A labelled Markov chain $\mchain = \lmc$, the equality relation on $\states$, i.e.~$\diag$, and the label-equality relation on $\states$, i.e.~$\mathord{\sim_0}$.}
$R \gets \mathord{\sim_0}$\;
\Repeat{$R = R_{\mathrm{old}}$}{
  $R_{\mathrm{old}} \gets R$\;
  $R \gets S^2_\Delta \cup \{\, (s, t) \in R_{\mathrm{old}} \mid \supp{\tau(s)} \times \supp{\tau(t)} \subseteq R_{\mathrm{old}} \,\}$\;
}
\Return $R$.
\end{algorithm}

Our non-deterministic algorithm is based on the two following basic behaviours of UL states: \emph{``synchronisation''} and \emph{``trace-uniqueness''}.
These behaviours arise from the recursive structure of $\stable$. %
In the case of synchronisation, the states generate the same finite prefix until all of their respective paths reach a common state.
Once the paths meet, the subsequent behaviour is identical thus indistinguishable based on the starting state.
Synchronisation is illustrated, e.g., by states $p$ and $w$ of \cref{figure:intro-example}, which synchronise after one step.
If synchronisation does not occur, then all paths from both UL states generate the same unique infinite word, as states $x$ and $y$ of \cref{figure:intro-example}. %
It follows that, intuitively, two states are not UL if they can reach different labels without synchronising. Our algorithm builds on the fact that such a discrepancy can be witnessed by histories of bounded length.

\begin{restatable}{lemma}{lemmaULStateWitnesses}\label{lemma:ule:state-witnesses}
  Let $\state, \stateB\in\states$.
  Then $\state$ and $\stateB$ are not UL if and only if there exist histories $\hist = \state_1\state_2\ldots\state_\iLast$ and $\hist'=\stateB_1\stateB_2\ldots\stateB_\iLast$ with $\state_1 = \state$, $\stateB_1=\stateB$ and $\iLast\leq|\states|$ such that $\lbl(\hist) \neq \lbl(\hist')$ and for all $1\leq\iPos\leq\iLast$, $\state_\iPos\neq\stateB_\iPos$.
\end{restatable}

Lemma~\ref{lemma:ule:state-witnesses} implies that checking if two states $\state$ and $\stateB$ are \textit{not UL} can be done by non-deterministically constructing two histories from $\state$ and $\stateB$ in parallel only retaining the last states and the length of the current histories.
At each step of the construction, we reject if the last states of both histories are equal or the history has length $|\states|$, we accept if the two states have different labels, and we continue to the next step otherwise.
This procedure (summarised in \cref{algorithm:states-nonule} in \cref{appendix:ue:nl-algorithm}) uses only logarithmic space. It follows that deciding UL for states is in $\conlogspace=\nlogspace$~\cite{DBLP:journals/siamcomp/Immerman88,DBLP:journals/acta/Szelepcsenyi88}.
We can also show $\nlogspace$-hardness by a straightforward reduction from the complement of directed graph reachability (cf.~\cref{lemma:hardness:ule}, Appendix~\ref{appendix:ue:hardness}).
The following theorem summarises our results.
\begin{theorem}\label{theorem:ule:nl-states}
  Deciding if two states are UL (or, equivalently, UB) is $\nlogspace$-complete.
\end{theorem}

Lemma~\ref{lemma:ule:state-witnesses} also reveals that two UL states remain language-equivalent even in a setting where the (consistent) transition function is allowed to change at every time step.
Thus, UL captures a robust form of language-equivalence that is invariant not only under every $\transB\in\cons{\trans}$, but also under arbitrary temporal variations of the transition probabilities.

\subparagraph*{Experiments.}

Our theoretical results show that UB (or, equivalently, UL for states) can be decided in polynomial time via partition refinement. We now investigate whether this efficiency also translates to practical performance and assess the usefulness of UB as a minimisation technique for probabilistic model checking. Since UB identifies states whose equivalence is independent of the precise transition probabilities, it can be viewed as a robust analogue of bisimilarity. Rather than minimising a model for a fixed assignment of transition probabilities, UB produces a minimised model that remains valid under every perturbation of the transition probabilities, provided that the underlying graph structure remains consistent.

To this end, we implemented \cref{algorithm:stable} in the widely used probabilistic model checker PRISM~\cite{DBLP:conf/cav/KwiatkowskaNP11}, an open-source tool providing quantitative verification and analysis of several types of probabilistic models.
We evaluated our algorithm by applying it to all (discrete-time) labelled Markov chains from the Quantitative Verification Benchmark Set (QVBS) \cite{DBLP:conf/tacas/HartmannsKPQR19}, 
which comprises $151$ model instances, with state spaces ranging up to $14,123,252$ states.
The details of the experimental setup and the benchmarking results can be found in Appendix~\ref{appendix:experiments}.
The main conclusions drawn from our experiments are that
\begin{enumerate}
\item UB mostly achieves reductions comparable to bisimilarity, even producing the same minimised model on $50\%$ of the evaluated instances, and
\item UB is generally faster than, or comparable to, bisimilarity. %
\end{enumerate}

Overall, our experiments demonstrate that UB is an effective and scalable state space reduction technique. Although it is coarser than bisimilarity and therefore generally produces larger quotient models, the reduction in state space size is comparable for many benchmark families, while computation is typically faster. Moreover, UB is especially well suited to the analysis of labelled Markov chains with uncertain, imprecise, or perturbed transition probabilities. Since universal equivalence captures robust behavioural equivalence, states identified by UB remain indistinguishable under every instantiation of the transition probabilities, ensuring that the computed quotient is also valid for the actual system.

\subsection{Universal equivalence for distributions}\label{section:ule:dist}

We now consider the UL relation for distributions.
We first show that two distributions are UL if and only if they attribute the same probability to all state-UL equivalence classes. %
We build on this characterisation to show that the complexity of deciding UL of distributions is identical to that of deciding UL of states.

\subparagraph*{Characterisation of language equivalence.} %

Let $U$ denote the UL relation restricted to states (as in the previous section).
We characterise universal language equivalence for distributions in terms of the equivalence classes of $U$.

\begin{restatable}{theorem}{theoremULDistributions}\label{theorem:ule-distributions}
For all $\measure, \measureB \in \dist{\states}$, $\measure$ and $\measureB$ are UL if and only if $\measure(\eClass) = \measureB(\eClass)$ holds for all $\eClass\in\states/U$.
\end{restatable}
\begin{proof}[Proof sketch]
We prove, by induction on $n \in\INpos$, that two distributions assign the same probability to every word of length $n$ under every $\transB\in\cons{\trans}$ if and only if they agree on every $U_n$-class. See Appendix~\ref{appendix:ul-distributions} for an extended proof sketch.

Recall that language equivalence is determined by words of length at most $|\states|$ (see, e.g.,~\cite{DBLP:journals/siamcomp/Tzeng92}).
Taking $n=|\states|$ and $U_{|\states|}=U$ yields the desired result.
\end{proof}

\subparagraph*{Algorithm and complexity.} %

We show that we can decide whether two distributions are UL in $\nlogspace$ using the characterisation of \cref{theorem:ule-distributions}.
Instead of computing the state-UL relation, we iterate over each state $\state$ and show that the distributions agree on the class of $\state$, by querying an $\nlogspace$ oracle for UL to identify which states are UL to $\state$. Non-deterministic logarithmic space suffices to implement this algorithm. The use of an $\nlogspace$ oracle does not increase the complexity, as $\nlogspace^\nlogspace = \nlogspace$~\cite{DBLP:journals/siamcomp/Immerman88}. To check the equality of the distributions on each equivalence class in logarithmic space, we use the fact that the equality of sums of rational numbers can be decided in logarithmic space~\cite{DBLP:journals/tcs/Jerabek12}.

The matching lower bound follows from Lemma~\ref{lemma:hardness:ule}, since states are a special case of distributions, yielding the following result.
\begin{theorem}\label{theorem:ule:d:complexity}
  Deciding if two distributions are UL is $\nlogspace$-complete.
\end{theorem}

\section{Existential equivalence}
\label{section:ee}
We study the existential perturbation equivalence relations and their corresponding decision problems.
We prove in Section~\ref{section:ele:deterministic} that EL for states coincides with EB by establishing that witness transition functions for EL can be chosen deterministic, and are therefore also witnesses of EB.
We discuss the complexity of the EL and EB problems in Section~\ref{section:ele:complexity}.
Technical details are deferred to Appendix~\ref{appendix:ee}.

\subsection{Deterministic witnesses for the comparison of states}\label{section:ele:deterministic}
Fix a labelled Markov chain $\mchain = \lmc$. %
A key structural insight is that probabilistic branching is unnecessary for establishing existential perturbation language equivalence: two states are EL if and only if they are language-equivalent under some deterministic consistent transition function. This characterisation is crucial for the $\np$ upper bound established in the following section.
We prove a stronger implication: if two states generate the same words (not necessarily with the same probabilities) under some consistent transition function, then there is a deterministic transition function witnessing that they are EL.

\begin{restatable}{theorem}{theoremELnew}\label{thm:ele-det}
  For all $\state, \stateB\in\states$, if there exists $\transB\in\cons{\trans}$ such that $\state$ and $\stateB$ generate the same words under $\transB$, then there exists a deterministic $\transB^{\mathsf{det}}\in\cons{\trans}$ such that $\state\langEquiv{\transB^{\mathsf{det}}}\stateB$.
\end{restatable}
\begin{proof}[Proof sketch]
  Let $\state, \stateB\in\states$ and $\transB\in\cons{\trans}$ such that $\state$ and $\stateB$ generate the same words under $\transB$.
  It suffices to show that there exists a deterministic $\transB^{\mathsf{det}}\in\cons{\trans}$ such that $\state$ and $\stateB$ generate the same word of length $n$ where $n=|\states|$ (see, e.g.,~\cite{DBLP:journals/siamcomp/Tzeng92}).
Fix a total order $\leq$ on $\lblSet$ and let $\leLex$ be the induced \emph{lexicographic order} on $\lblSet^*$.
For each $\stateC\in\states$, let $\word_{\stateC}\in\lblSet^n$ be the $\leLex$-minimal word that can occur from $\stateC$ under $\transB$.
By the assumption on $\transB$, we have $\word_\state = \word_\stateB$.

We construct $\transB^{\mathsf{det}}$ from $\transB$ by retaining, from each state $\stateC$, one successor $\stateC'$ with minimal $\word_{\stateC'}$. By definition of the lexicographic order, the unique length-$n$ word induced from a state $\stateC$ under $\transB^{\mathsf{det}}$ is $\word_{\stateC}$. Thus, $\state$ and $\stateB$ generate the same unique word under $\transB^{\mathsf{det}}$, i.e., $\state\langEquiv{\transB^{\mathsf{det}}}\stateB$.
\end{proof}

The idea of using a lexicographic ordering to filter out choices used above also appears, e.g., in~\cite{DBLP:conf/birthday/ChoffrutG99} to derive a rational function from a rational relation.

As language equivalence and bisimilarity are equivalent in deterministic Markov chains, \cref{thm:ele-det} implies the following. %

\begin{restatable}{theorem}{theoremELepb}\label{thm:ele-epb}
  Two states are EL if and only if they are EB.
\end{restatable}

We briefly comment on variants of EL and EB requiring small perturbations of $\trans$.
For $\varepsilon>0$, we say that two states are $\varepsilon$-EL (resp.~$\varepsilon$-EB) if there exists $\transB\in\cons{\trans}$ in the $\varepsilon$-neighbourhood of $\trans$ such that they are language-equivalent (resp. bisimilar) under $\transB$.
We observe the following implications.
\begin{restatable}{proposition}{propEpsilonELImplications}\label{proposition:ele:epsilon-implications}
  For any $\varepsilon > 0$ and two states, $\varepsilon$-EB implies $\varepsilon$-EL, and $\varepsilon$-EL implies EL and EB.
  The converse implications do not hold in general.
\end{restatable}

\subsection{The complexity of EL and EB} \label{section:ele:complexity}

We discuss the complexity of the existential problems: we show that the EL and EB problem for states are $\np$-complete and that the EL problem for distributions is $\exists\IR$-complete.

\subparagraph*{Equivalence of states.}
For deciding if states are EL or EB, a direct $\np$ upper bound follows from \cref{thm:ele-det}: it suffices to guess a deterministic transition function that is consistent with the input transition function and check in polynomial time whether the two considered states are equivalent.
We show that these problems are $\np$-hard by a reduction from the $3$-SAT problem in Appendix~\ref{appendix:ele:np-hardness}.
We obtain the following complexity result.
\begin{restatable}{theorem}{theoremComplexityELstates}\label{theorem:complexity:ele-states}
Deciding if two states are EL (or, equivalently, EB) is $\np$-complete.
\end{restatable}

\subparagraph*{Equivalence of distributions.}
The EL problem for distributions is harder than the EL problem for states, since randomness is necessary in general (cf.~\cref{figure:eled}).
Membership of the EL problem for distributions in $\exists\IR$ follows from the $\exists\IR$-membership of a generalisation of the EL problem to MDPs~\cite{DBLP:conf/fsttcs/Kiefer020} (cf.~Appendix~\ref{section:mdps:problems} for a reduction).
For the lower bound, we present a reduction from the $\exists\IR$-complete non-negative matrix factorisation problem to the EL problem for distributions in Appendix~\ref{appendix:ele:nmf}.
We obtain the following result.

\begin{restatable}{theorem}{theoremNMFele}\label{theorem:nmf-ele}
Deciding if two distributions are EL is $\exists\IR$-complete.
\end{restatable}
\begin{figure}[htb]
  \centering
  \begin{tikzpicture}[every state/.style={minimum size=0.7cm}]
    \node at (7,1.25) () {$\nu$};
    \node at (2,1.25) () {$\mu$};
    \node[state,fill=black,scale=0.4,minimum size=1pt] at (7,1) (nu) {};
    \node[state,fill=black,scale=0.4,minimum size=1pt] at (2,1) (mu) {};
    \node[state] at (8,0) (s) {};
    \node[state] at (0,0) (w) {};
    \node[state] at (2,0) (t) {$s$};
    \node[state, fill=lipicsYellow] at (4,0) (v) {};
    \node[state] at (6,0) (u) {};
    \path[-stealth] (nu) edge node[above right] {$\frac{1}{2}$} (s);
    \path[-stealth] (nu) edge node[above left] {$\frac{1}{2}$} (u);
    \path[-stealth] (mu) edge node[left] {$1$} (t);
    \path[-stealth] (t) edge (w);
    \path[-stealth] (t) edge (v);
    \path[-stealth] (u) edge (v);
    \path[-stealth] (s) edge[loop right] (s);
    \path[-stealth] (v) edge[loop above] (v);
    \path[-stealth] (w) edge[loop left] (w);
  \end{tikzpicture}
  \caption{An LMC in which deterministic transition functions do not suffice to witness that the distributions $\mu$ and $\nu$ are EL, since we must assign probability $\frac{1}{2}$ to each outgoing transition from $s$.}
  \label{figure:eled}
\end{figure}
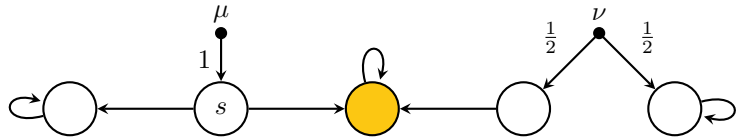

\section{Conclusion}
\label{section:conclusion}
We have introduced perturbation variants of language equivalence and probabilistic bisimilarity for labelled Markov chains, motivated by settings in which the underlying graph is known while the transition probabilities may be uncertain. We have established a complete complexity landscape for these notions.  We have shown that universal perturbation equivalence is a stable notion of equivalence that admits several characterisations, namely \cref{corollary:ule:upb} and \cref{proposition:ule:epsilon} show that UL, UB, $\varepsilon$-UL and $\varepsilon$-UB all coincide.
Our experimental results have demonstrated that universal equivalence provides an efficient and robust minimisation technique for probabilistic model checking.

Several directions for future work emerge. A natural next step is to extend perturbation behavioural equivalence beyond labelled Markov chains to richer probabilistic models, such as probabilistic automata or stochastic games, where non-determinism and probability interact. Another promising direction is to study relaxations of universal equivalence, which may offer a larger reduction of the model while retaining a useful degree of robustness. Finally, it would be interesting to investigate the effectiveness of our algorithm on interval Markov chains.

\bibliography{references}

@string{lmcs = {Logical Methods in Computer Science}}

@string{acta = {Acta Informatica}}

@string{siamjc = {SIAM Journal on Computing}}

@string{siamjo = {SIAM Journal on Optimization}}

@string{tcs = {Theoretical Computer Science}}

@string{ipl = {Information Processing Letters}}

@string{infcontrol = {Information and Control}}

@string{jcss = {Journal of Computer and System Sciences}}

@string{sttt = {International Journal on Software Tools for Technology Transfer}}

@string{ijfcs = {International Journal of Foundations of Computer Science}}

@string{concur99 = {Proceedings of the 10th International Conference on Concurrency Theory, {CONCUR} 1999, Eindhoven, The Netherlands, August 24--27, 1999}}

@string{concur26 = {Proceedings of the 37th International Conference on Concurrency Theory, {CONCUR} 2026, Liverpool, United Kingdom, September 1--5, 2026}}

@string{fossacs12 = {Proceedings of the 15th International Conference on Foundations of Software Science and Computational Structures, {FoSSaCS} 2012, Held as Part of {ETAPS} 2012, Tallinn, Estonia, March 24 -- April 1, 2012}}

@string{fsttcs20 = {Proceedings of the 40th {IARCS} Annual Conference on Foundations of Software Technology and Theoretical Computer Science, {FSTTCS} 2020, {BITS} Pilani, {K}{K} Birla Goa Campus, Goa, India (Virtual Conference), December 14--18, 2020}}

@string{fsttcs21 = {Proceedings of the 41st {IARCS} Annual Conference on Foundations of Software Technology and Theoretical Computer Science, {FSTTCS} 2021, Virtual Conference, December 15--17, 2021}}

@string{lics91 = {Proceedings of the 6th Annual {IEEE} Symposium on Logic in Computer Science, {LICS} 1991, Amsterdam, The Netherlands, July 15--18, 1991}}

@string{stacs26 = {Proceedings of the 43rd International Symposium on Theoretical Aspects of Computer Science, {STACS} 2026, Grenoble, France, March 9--13, 2026}}

@string{tacas07 = {Proceedings of the 13th International Conference on Tools and Algorithms for the Construction and Analysis of Systems, {TACAS} 2007, Held as Part of {ETAPS} 2007, Braga, Portugal, March 24 -- April 1, 2007}}

@string{tacas13 = {Proceedings of the 19th International Conference on Tools and Algorithms for the Construction and Analysis of Systems, {TACAS} 2013, Held as Part of {ETAPS} 2013, Rome, Italy, March 16--24, 2013}}

@string{tacas19I = {Proceedings (Part {I}) of the 25th International Conference on Tools and Algorithms for the Construction and Analysis of Systems, {TACAS} 2019, Held as Part of {ETAPS} 2019, Prague, Czech Republic, April 6--11, 2019}}

@string{cav96 = {Proceedings of the 8th International Conference on Computer Aided Verification, {CAV} 1996, New Brunswick, NJ, USA, July 31 -- August 3, 1996}}

@string{cav11 = {Proceedings of the 23rd International Conference on Computer Aided Verification, {CAV} 2011, Snowbird, UT, USA, July 14--20, 2011}}

@string{cav25II = {Proceedings (Part {II}) of the 37th International Conference on Computer Aided Verification, {CAV} 2025, Zagreb, Croatia, July 23--25, 2025}}

@string{qest14 = {Proceedings of the 11th International Conference on Quantitative Evaluation of Systems, {QEST} 2014, Florence, Italy, September 8--10, 2014}}

@string{stoc71 ={Proceedings of the 3rd Annual {ACM} Symposium on Theory of Computing, {STOC} 1971, Shaker Heights, Ohio, {USA}, May 3--5, 1971}}

@string{syncop14 = {Proceedings of the 1st International Workshop on Synthesis of Continuous Parameters, SynCoP 2014, Grenoble, France, April 6, 2014}}

@string{lata16 = {Proceedings of the 10th International Conference on Language and Automata Theory and Applications, {LATA} 2016, Prague, Czech Republic, March 14--18, 2016}}

@string{ictac04 = {Proceedings of the 1st International Colloquium on Theoretical Aspects of Computing, {ICTAC} 2004, Guiyang, China, September 20--24, 2004}}

@string{popl89 = {Proceedings of the 16th Annual {ACM} Symposium on Principles of Programming Languages, POPL 1989, Austin, {TX}, {USA}, January 11--13, 1989}}

@book{BK08,
  author       = {Christel Baier and
                  Joost{-}Pieter Katoen},
  title        = {Principles of model checking},
  publisher    = {{MIT} Press},
  year         = {2008},
  isbn         = {978-0-262-02649-9},
  bibsource    = {dblp computer science bibliography, https://dblp.org}
}

@inproceedings{Co71,
  author       = {Stephen A. Cook},
  editor       = {Michael A. Harrison and
                  Ranan B. Banerji and
                  Jeffrey D. Ullman},
  title        = {The Complexity of Theorem-Proving Procedures},
  booktitle    = stoc71,
  pages        = {151--158},
  publisher    = {{ACM}},
  year         = {1971},
  url          = {https://doi.org/10.1145/800157.805047},
  doi          = {10.1145/800157.805047},
  bibsource    = {dblp computer science bibliography, https://dblp.org}
}

@article{Jon75,
  author       = {Neil D. Jones},
  title        = {Space-Bounded Reducibility among Combinatorial Problems},
  journal      = jcss,
  volume       = {11},
  number       = {1},
  pages        = {68--85},
  year         = {1975},
  url          = {https://doi.org/10.1016/S0022-0000(75)80050-X},
  doi          = {10.1016/S0022-0000(75)80050-X},
  bibsource    = {dblp computer science bibliography, https://dblp.org}
}

@article{DBLP:journals/sttt/HenselJKQV22,
  author       = {Christian Hensel and
                  Sebastian Junges and
                  Joost{-}Pieter Katoen and
                  Tim Quatmann and
                  Matthias Volk},
  title        = {The probabilistic model checker {Storm}},
  journal      = sttt,
  volume       = {24},
  number       = {4},
  pages        = {589--610},
  year         = {2022},
  url          = {https://doi.org/10.1007/s10009-021-00633-z},
  doi          = {10.1007/S10009-021-00633-Z},
  bibsource    = {dblp computer science bibliography, https://dblp.org}
}

@inproceedings{DBLP:conf/cav/KwiatkowskaNP11,
  author    = {Marta Z. Kwiatkowska and
               Gethin Norman and
               David Parker},
  editor    = {Ganesh Gopalakrishnan and
               Shaz Qadeer},
  title     = {{PRISM} 4.0: Verification of Probabilistic Real-Time Systems},
  booktitle = cav11,
  series    = {Lecture Notes in Computer Science},
  volume    = {6806},
  pages     = {585--591},
  publisher = {Springer},
  year      = {2011},
  url       = {https://doi.org/10.1007/978-3-642-22110-1\_47},
  doi       = {10.1007/978-3-642-22110-1\_47},
  bibsource = {dblp computer science bibliography, https://dblp.org}
}

@inproceedings{DBLP:conf/fsttcs/Kiefer020,
  author       = {Stefan Kiefer and
                  Qiyi Tang},
  editor       = {Nitin Saxena and
                  Sunil Simon},
  title        = {Comparing Labelled {Markov} Decision Processes},
  booktitle    = fsttcs20,
  series       = {LIPIcs},
  pages        = {49:1--49:16},
  publisher    = {Schloss Dagstuhl - Leibniz-Zentrum f{\"{u}}r Informatik},
  year         = {2020},
  url          = {https://doi.org/10.4230/LIPIcs.FSTTCS.2020.49},
  doi          = {10.4230/LIPICS.FSTTCS.2020.49},
  bibsource    = {dblp computer science bibliography, https://dblp.org}
}

@inproceedings{DBLP:conf/fossacs/ChenBW12,
  author       = {Di Chen and
                  Franck {van Breugel} and
                  James Worrell},
  editor       = {Lars Birkedal},
  title        = {On the Complexity of Computing Probabilistic Bisimilarity},
  booktitle    = fossacs12,
  series       = {Lecture Notes in Computer Science},
  pages        = {437--451},
  publisher    = {Springer},
  year         = {2012},
  url          = {https://doi.org/10.1007/978-3-642-28729-9\_29},
  doi          = {10.1007/978-3-642-28729-9\_29},
  bibsource    = {dblp computer science bibliography, https://dblp.org}
}

@article{DBLP:journals/ipl/Tzeng96,
  author       = {Wen{-}Guey Tzeng},
  title        = {On Path Equivalence of Nondeterministic Finite Automata},
  journal      = ipl,
  volume       = {58},
  number       = {1},
  pages        = {43--46},
  year         = {1996},
  url          = {https://doi.org/10.1016/0020-0190(96)00039-7},
  doi          = {10.1016/0020-0190(96)00039-7},
  bibsource    = {dblp computer science bibliography, https://dblp.org}
}

@article{DBLP:journals/lmcs/FijalkowKS20,
  author       = {Nathana{\"{e}}l Fijalkow and
                  Stefan Kiefer and
                  Mahsa Shirmohammadi},
  title        = {Trace Refinement in Labelled {Markov} Decision Processes},
  journal      = lmcs,
  volume       = {16},
  number       = {2},
  year         = {2020},
  url          = {https://doi.org/10.23638/LMCS-16(2:10)2020},
  doi          = {10.23638/LMCS-16(2:10)2020},
  bibsource    = {dblp computer science bibliography, https://dblp.org}
}

@article{DBLP:journals/siamcomp/Immerman88,
  author       = {Neil Immerman},
  title        = {Nondeterministic Space is Closed Under Complementation},
  journal      = siamjc,
  volume       = {17},
  number       = {5},
  pages        = {935--938},
  year         = {1988},
  url          = {https://doi.org/10.1137/0217058},
  doi          = {10.1137/0217058},
  bibsource    = {dblp computer science bibliography, https://dblp.org}
}

@article{DBLP:journals/acta/Szelepcsenyi88,
  author       = {R{\'{o}}bert Szelepcs{\'{e}}nyi},
  title        = {The Method of Forced Enumeration for Nondeterministic Automata},
  journal      = acta,
  volume       = {26},
  number       = {3},
  pages        = {279--284},
  year         = {1988},
  url          = {https://doi.org/10.1007/BF00299636},
  doi          = {10.1007/BF00299636},
  bibsource    = {dblp computer science bibliography, https://dblp.org}
}

@article{DBLP:journals/tcs/Jerabek12,
  author       = {Emil Jer{\'{a}}bek},
  title        = {Root finding with threshold circuits},
  journal      = tcs,
  volume       = {462},
  pages        = {59--69},
  year         = {2012},
  url          = {https://doi.org/10.1016/j.tcs.2012.09.001},
  doi          = {10.1016/J.TCS.2012.09.001},
  bibsource    = {dblp computer science bibliography, https://dblp.org}
}

@article{DBLP:journals/siamcomp/ChandraSV84,
  author       = {Ashok K. Chandra and
                  Larry J. Stockmeyer and
                  Uzi Vishkin},
  title        = {Constant Depth Reducibility},
  journal      = siamjc,
  volume       = {13},
  number       = {2},
  pages        = {423--439},
  year         = {1984},
  url          = {https://doi.org/10.1137/0213028},
  doi          = {10.1137/0213028},
  bibsource    = {dblp computer science bibliography, https://dblp.org}
}

@article{DBLP:journals/jcss/HesseAB02,
  author       = {William Hesse and
                  Eric Allender and
                  David A. Mix Barrington},
  title        = {Uniform constant-depth threshold circuits for division and iterated
                  multiplication},
  journal      = jcss,
  volume       = {65},
  number       = {4},
  pages        = {695--716},
  year         = {2002},
  url          = {https://doi.org/10.1016/S0022-0000(02)00025-9},
  doi          = {10.1016/S0022-0000(02)00025-9},
  bibsource    = {dblp computer science bibliography, https://dblp.org}
}

@article{shitov18NMF,
  author    = {Yaroslav Shitov},
  title     = {A universality theorem for nonnegative matrix factorizations}, 
  journal   = {CoRR},
  volume    = {abs/1606.09068},
  year      = {2018},
  url       = {https://arxiv.org/abs/1606.09068}, 
}

@misc{shitov21NMF,
  author    = {Yaroslav Shitov},
  title     = {A universality theorem for nonnegative matrix factorizations}, 
  url       = {https://vixra.org/abs/2101.0167},
  year      = {2021},
  note      = {viXra e-Print archive}
}

@inproceedings{DBLP:conf/birthday/ChoffrutG99,
  author       = {Christian Choffrut and
                  Serge Grigorieff},
  editor       = {Juhani Karhum{\"{a}}ki and
                  Hermann A. Maurer and
                  Gheorghe Paun and
                  Grzegorz Rozenberg},
  title        = {Uniformization of Rational Relations},
  booktitle    = {Jewels are Forever, Contributions on Theoretical Computer Science
                  in Honor of Arto Salomaa},
  pages        = {59--71},
  publisher    = {Springer},
  year         = {1999},
  bibsource    = {dblp computer science bibliography, https://dblp.org}
}

@inproceedings{DBLP:conf/stacs/CernyS26,
  author       = {Marek Cern{\'{y}} and
                  Tim Seppelt},
  editor       = {Meena Mahajan and
                  Florin Manea and
                  Annabelle McIver and
                  Kim Thang Nguyen},
  title        = {Homomorphism Indistinguishability, Multiplicity Automata Equivalence,
                  and Polynomial Identity Testing},
  booktitle    = stacs26,
  series       = {LIPIcs},
  pages        = {25:1--25:20},
  publisher    = {Schloss Dagstuhl - Leibniz-Zentrum f{\"{u}}r Informatik},
  year         = {2026},
  url          = {https://doi.org/10.4230/LIPIcs.STACS.2026.25},
  doi          = {10.4230/LIPICS.STACS.2026.25},
  bibsource    = {dblp computer science bibliography, https://dblp.org}
}

@inproceedings{DBLP:conf/tacas/HartmannsKPQR19,
  author       = {Arnd Hartmanns and
                  Michaela Klauck and
                  David Parker and
                  Tim Quatmann and
                  Enno Ruijters},
  editor       = {Tom{\'{a}}s Vojnar and
                  Lijun Zhang},
  title        = {The Quantitative Verification Benchmark Set},
  booktitle    = tacas19I,
  series       = {Lecture Notes in Computer Science},
  volume       = {11427},
  pages        = {344--350},
  publisher    = {Springer},
  year         = {2019},
  url          = {https://doi.org/10.1007/978-3-030-17462-0\_20},
  doi          = {10.1007/978-3-030-17462-0\_20},
  bibsource    = {dblp computer science bibliography, https://dblp.org}
}

@article{CR93,
  title={Nonnegative ranks, decompositions, and factorizations of nonnegative matrices},
  author={Cohen, Joel E and Rothblum, Uriel G},
  journal={Linear Algebra and its Applications},
  volume={190},
  pages={149--168},
  year={1993},
  publisher={Elsevier}
}

@article{DBLP:journals/siamjo/Vavasis09,
  author       = {Stephen A. Vavasis},
  title        = {On the Complexity of Nonnegative Matrix Factorization},
  journal      = siamjo,
  volume       = {20},
  number       = {3},
  pages        = {1364--1377},
  year         = {2009},
  url          = {https://doi.org/10.1137/070709967},
  doi          = {10.1137/070709967},
  bibsource    = {dblp computer science bibliography, https://dblp.org}
}

@inproceedings{DBLP:conf/cav/FatmiKPB25,
  author       = {Syyeda Zainab Fatmi and
                  Stefan Kiefer and
                  David Parker and
                  Franck van Breugel},
  editor       = {Ruzica Piskac and
                  Zvonimir Rakamaric},
  title        = {Robust Probabilistic Bisimilarity for Labelled {Markov} Chains},
  booktitle    = cav25II,
  series       = {Lecture Notes in Computer Science},
  volume       = {15932},
  pages        = {254--275},
  publisher    = {Springer},
  year         = {2025},
  url          = {https://doi.org/10.1007/978-3-031-98679-6\_12},
  doi          = {10.1007/978-3-031-98679-6\_12},
  bibsource    = {dblp computer science bibliography, https://dblp.org}
}

@inproceedings{DBLP:conf/lics/JonssonL91,
  author       = {Bengt Jonsson and
                  Kim Guldstrand Larsen},
  title        = {Specification and Refinement of Probabilistic Processes},
  booktitle    = lics91,
  pages        = {266--277},
  publisher    = {{IEEE} Computer Society},
  year         = {1991},
  url          = {https://doi.org/10.1109/LICS.1991.151651},
  doi          = {10.1109/LICS.1991.151651},
  bibsource    = {dblp computer science bibliography, https://dblp.org}
}

@inproceedings{DBLP:conf/ictac/Daws04,
  author       = {Conrado Daws},
  editor       = {Zhiming Liu and
                  Keijiro Araki},
  title        = {Symbolic and Parametric Model Checking of Discrete-Time {Markov} Chains},
  booktitle    = ictac04,
  series       = {Lecture Notes in Computer Science},
  volume       = {3407},
  pages        = {280--294},
  publisher    = {Springer},
  year         = {2004},
  url          = {https://doi.org/10.1007/978-3-540-31862-0\_21},
  doi          = {10.1007/978-3-540-31862-0\_21},
  bibsource    = {dblp computer science bibliography, https://dblp.org}
}

@inproceedings{DBLP:conf/tacas/BenediktLW13,
  author       = {Michael Benedikt and
                  Rastislav Lenhardt and
                  James Worrell},
  editor       = {Nir Piterman and
                  Scott A. Smolka},
  title        = {{LTL} Model Checking of Interval {Markov} Chains},
  booktitle    = tacas13,
  series       = {Lecture Notes in Computer Science},
  volume       = {7795},
  pages        = {32--46},
  publisher    = {Springer},
  year         = {2013},
  url          = {https://doi.org/10.1007/978-3-642-36742-7\_3},
  doi          = {10.1007/978-3-642-36742-7\_3},
  bibsource    = {dblp computer science bibliography, https://dblp.org}
}

@article{DBLP:journals/corr/abs-2009-01217,
  author       = {Stefan Kiefer},
  title        = {Notes on Equivalence and Minimization of Weighted Automata},
  journal      = {CoRR},
  volume       = {abs/2009.01217},
  year         = {2020},
  url          = {https://arxiv.org/abs/2009.01217},
  eprinttype   = {arXiv},
  eprint       = {2009.01217},
  bibsource    = {dblp computer science bibliography, https://dblp.org}
}

@inproceedings{DBLP:conf/tacas/KatoenKZJ07,
  author       = {Joost{-}Pieter Katoen and
                  Tim Kemna and
                  Ivan S. Zapreev and
                  David N. Jansen},
  editor       = {Orna Grumberg and
                  Michael Huth},
  title        = {Bisimulation Minimisation Mostly Speeds Up Probabilistic Model Checking},
  booktitle    = tacas07,
  series       = {Lecture Notes in Computer Science},
  volume       = {4424},
  pages        = {87--101},
  publisher    = {Springer},
  year         = {2007},
  url          = {https://doi.org/10.1007/978-3-540-71209-1\_9},
  doi          = {10.1007/978-3-540-71209-1\_9},
  bibsource    = {dblp computer science bibliography, https://dblp.org}
}

@article{DBLP:journals/iandc/Schutzenberger61b,
  author       = {Marcel Paul Sch{\"{u}}tzenberger},
  title        = {On the Definition of a Family of Automata},
  journal      = infcontrol,
  volume       = {4},
  number       = {2-3},
  pages        = {245--270},
  year         = {1961},
  url          = {https://doi.org/10.1016/S0019-9958(61)80020-X},
  doi          = {10.1016/S0019-9958(61)80020-X},
  bibsource    = {dblp computer science bibliography, https://dblp.org}
}

@book{Paz71,
  author       = {Azaria Paz},
  title        = {Introduction to Probabilistic Automata},
  publisher    = {Academic Press},
  year         = {1971},
  isbn         = {978-0-12-547650-8},
  doi          = {10.1016/C2013-0-11297-4}
}

@article{DBLP:journals/ijfcs/DoyenHR08,
  author       = {Laurent Doyen and
                  Thomas A. Henzinger and
                  Jean{-}Fran{\c{c}}ois Raskin},
  title        = {Equivalence of Labeled {Markov} Chains},
  journal      = ijfcs,
  volume       = {19},
  number       = {3},
  pages        = {549--563},
  year         = {2008},
  url          = {https://doi.org/10.1142/S0129054108005814},
  doi          = {10.1142/S0129054108005814},
  bibsource    = {dblp computer science bibliography, https://dblp.org}
}

@article{DBLP:journals/siamcomp/Tzeng92,
  author       = {Wen{-}Guey Tzeng},
  title        = {A Polynomial-Time Algorithm for the Equivalence of Probabilistic Automata},
  journal      = siamjc,
  volume       = {21},
  number       = {2},
  pages        = {216--227},
  year         = {1992},
  url          = {https://doi.org/10.1137/0221017},
  doi          = {10.1137/0221017},
  bibsource    = {dblp computer science bibliography, https://dblp.org}
}

@inproceedings{DBLP:conf/popl/LarsenS89,
  author       = {Kim Guldstrand Larsen and
                  Arne Skou},
  title        = {Bisimulation Through Probabilistic Testing},
  booktitle    = popl89,
  pages        = {344--352},
  publisher    = {{ACM} Press},
  year         = {1989},
  url          = {https://doi.org/10.1145/75277.75307},
  doi          = {10.1145/75277.75307},
  bibsource    = {dblp computer science bibliography, https://dblp.org}
}

@book{KS60,
    author    = "John G. Kemeny and J. Laurie Snell",
    title     = "Finite {M}arkov chains",
    publisher = "Springer-Verlag",
    year      = "1960",
    address   = "Heidelberg, Germany"
}

@inproceedings{DBLP:conf/cav/Baier96,
  author       = {Christel Baier},
  editor       = {Rajeev Alur and
                  Thomas A. Henzinger},
  title        = {Polynomial Time Algorithms for Testing Probabilistic Bisimulation
                  and Simulation},
  booktitle    = cav96,
  series       = {Lecture Notes in Computer Science},
  volume       = {1102},
  pages        = {50--61},
  publisher    = {Springer},
  year         = {1996},
  url          = {https://doi.org/10.1007/3-540-61474-5\_57},
  doi          = {10.1007/3-540-61474-5\_57},
  bibsource    = {dblp computer science bibliography, https://dblp.org}
}

@inproceedings{DBLP:conf/qest/JaegerMLM14,
  author       = {Manfred Jaeger and
                  Hua Mao and
                  Kim Guldstrand Larsen and
                  Radu Mardare},
  editor       = {Gethin Norman and
                  William H. Sanders},
  title        = {Continuity Properties of Distances for {Markov} Processes},
  booktitle    = qest14,
  series       = {Lecture Notes in Computer Science},
  volume       = {8657},
  pages        = {297--312},
  publisher    = {Springer},
  year         = {2014},
  url          = {https://doi.org/10.1007/978-3-319-10696-0\_24},
  doi          = {10.1007/978-3-319-10696-0\_24},
  bibsource    = {dblp computer science bibliography, https://dblp.org}
}

@inproceedings{DBLP:conf/fsttcs/Kiefer021,
  author       = {Stefan Kiefer and
                  Qiyi Tang},
  editor       = {Mikolaj Bojanczyk and
                  Chandra Chekuri},
  title        = {Approximate Bisimulation Minimisation},
  booktitle    = fsttcs21,
  series       = {LIPIcs},
  volume       = {213},
  pages        = {48:1--48:16},
  publisher    = {Schloss Dagstuhl - Leibniz-Zentrum f{\"{u}}r Informatik},
  year         = {2021},
  url          = {https://doi.org/10.4230/LIPIcs.FSTTCS.2021.48},
  doi          = {10.4230/LIPICS.FSTTCS.2021.48},
  bibsource    = {dblp computer science bibliography, https://dblp.org}
}

@inproceedings{DBLP:conf/concur/DesharnaisGJP99,
  author       = {Jos{\'{e}}e Desharnais and
                  Vineet Gupta and
                  Radha Jagadeesan and
                  Prakash Panangaden},
  editor       = {Jos C. M. Baeten and
                  Sjouke Mauw},
  title        = {Metrics for Labeled {Markov} Systems},
  booktitle    = concur99,
  series       = {Lecture Notes in Computer Science},
  volume       = {1664},
  pages        = {258--273},
  publisher    = {Springer},
  year         = {1999},
  url          = {https://doi.org/10.1007/3-540-48320-9\_19},
  doi          = {10.1007/3-540-48320-9\_19},
  bibsource    = {dblp computer science bibliography, https://dblp.org}
}

@book{DBLP:books/daglib/0023601,
  author       = {Brian A. Davey and
                  Hilary A. Priestley},
  title        = {Introduction to Lattices and Order, Second Edition},
  publisher    = {Cambridge University Press},
  year         = {2002},
  url          = {https://doi.org/10.1017/CBO9780511809088},
  doi          = {10.1017/CBO9780511809088},
  isbn         = {978-0-521-78451-1},
  bibsource    = {dblp computer science bibliography, https://dblp.org}
}

@inproceedings{AKJB25,
    author    = {Adnan Ahmed and Hiva Karami and Anto Nanah Ji and Franck van Breugel},
    title     = {Bisimulation Minimisation Still Mostly Speeds Up Probabilistic Model Checking},
    booktitle = {Proceedings of the 4th Workshop on Reproducibility and Replication of Research Results},
    year      = {2025},
    address   = {Hamilton, ON, Canada},
    month     = may
}

@inproceedings{DBLP:conf/lata/HashemiH0STW16,
  author       = {Vahid Hashemi and
                  Holger Hermanns and
                  Lei Song and
                  K. Subramani and
                  Andrea Turrini and
                  Piotr Wojciechowski},
  editor       = {Adrian{-}Horia Dediu and
                  Jan Janousek and
                  Carlos Mart{\'{\i}}n{-}Vide and
                  Bianca Truthe},
  title        = {Compositional Bisimulation Minimization for Interval {Markov} Decision
                  Processes},
  booktitle    = lata16,
  series       = {Lecture Notes in Computer Science},
  volume       = {9618},
  pages        = {114--126},
  publisher    = {Springer},
  year         = {2016},
  url          = {https://doi.org/10.1007/978-3-319-30000-9\_9},
  doi          = {10.1007/978-3-319-30000-9\_9},
  bibsource    = {dblp computer science bibliography, https://dblp.org}
}

@inproceedings{DBLP:journals/corr/HashemiHK14,
  author       = {Vahid Hashemi and
                  Hassan Hatefi and
                  Jan Krc{\'{a}}l},
  editor       = {{\'{E}}tienne Andr{\'{e}} and
                  Goran Frehse},
  title        = {Probabilistic Bisimulations for {PCTL} Model Checking of Interval
                  {MDPs} (extended version)},
  booktitle    = syncop14,
  series       = {{EPTCS}},
  volume       = {145},
  pages        = {19--33},
  year         = {2014},
  url          = {https://doi.org/10.4204/EPTCS.145.4},
  doi          = {10.4204/EPTCS.145.4},
  bibsource    = {dblp computer science bibliography, https://dblp.org}
}

@inproceedings{concur,
    author       = {Syyeda Zainab Fatmi and Stefan Kiefer and David Parker and Franck van Breugel},
    editor       = {Ana Sokolova and Patrick Totzke},
    title        = {On the continuity of the probabilistic bisimilarity distance},
    booktitle    = concur26,
    series       = {LIPIcs},
    volume       = {391},
    pages        = {34:1--34:18},
    publisher    = {Schloss Dagstuhl - Leibniz-Zentrum f{\"{u}}r Informatik},
    year         = {2026},
    doi          = {10.4230/LIPIcs.CONCUR.2026.34}
}

\appendix
\newpage

\section{Details of Section~\ref{section:ue} -- \nameref{section:ue}}
\label{appendix:ue}
We fix a labelled Markov chain $\mchain = \lmc$ for the whole section.

\subsection{Comparison to robust probabilistic bisimilarity}
\label{appendix:robust}

The labelled Markov chain in \cref{figure:upb-types} illustrates the distinction between robust probabilistic bisimilarity and UB that was outlined in the introduction.
The states $s$ and $x$ are robustly probabilistically bisimilar, however, they are not UB, because any perturbation of the outgoing transition probabilities of $s$ destroys probabilistic bisimilarity.
Conversely, $u$ and $w$ are UB, but they are not robustly probabilistically bisimilar, since non-consistent perturbations can significantly change the behaviour of these states.
The states $v$ and $y$ satisfy both notions, while $s$ and $t$ are probabilistically bisimilar but satisfy neither.

\begin{figure}[ht]
  \centering
  \begin{tikzpicture}[every state/.style={minimum size=0.7cm}]
    \node[state] at (1,3) (s) {$s$};
    \node[state] at (3,3) (t) {$t$};
    \node[state, fill=lipicsYellow] at (0,1.5) (u) {$u$};
    \node[state] at (2,1.5) (v) {$v$};
    \node[state, fill=lipicsYellow] at (4,1.5) (w) {$w$};
    \node[state] at (1,0) (x) {$x$};
    \node[state] at (3,0) (y) {$y$};
    \path[-stealth] (s) edge (u);
    \path[-stealth] (s) edge (v);
    \path[-stealth] (t) edge (v);
    \path[-stealth] (t) edge (w);
    \path[-stealth] (u) edge[loop left] (u);
    \path[-stealth] (v) edge (w);
    \path[-stealth] (w) edge[loop right] (w);
    \path[-stealth] (x) edge (u);
    \path[-stealth] (x) edge (v);
    \path[-stealth] (y) edge (w);
  \end{tikzpicture}
  \caption{A labelled Markov chain witnessing the types of probabilistically bisimilar states. The transition probabilities are omitted, assume a uniform distribution over successors for each state.}
  \label{figure:upb-types}
\end{figure}
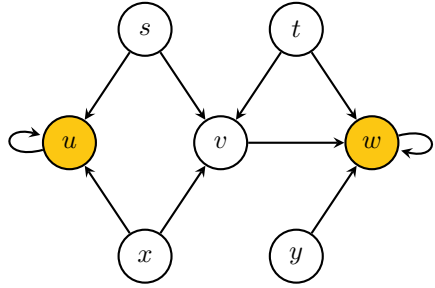

\subsection{Proof of \cref{theorem:ule-states}}

The goal of this section is to prove \cref{theorem:ule-states}.
We recall some notation from the main text.
We let $U$ denote the restriction of the UL relation to states, and for all $n\in\INpos$, we let $U_n$ denote the set of states that are universally perturbation language equivalent for words of length $n$.
We consider the operator $\stable\colon \ccInt{\diag}{\mathord{\sim_0}}\to \ccInt{\diag}{\mathord{\sim_0}}$ defined, for all $R\in\ccInt{\diag}{\mathord{\sim_0}}$, by
\[
  \stable(R) = \diag \cup \{\, (s, t) \in R \mid \supp{\trans(s)} \times \supp{\trans(t)} \subseteq R\}.
\]

We first prove the following technical lemma, that is useful in the last point of the proof sketched in the main text.
The lemma states that, for all $n\geq 2$, if all successors of two states $\state$ and $\stateB$ are UL for words of length $n-1$, then changing the outgoing probabilities of $\state$ will not affect the probability of words of length $n$ from any state, and will not affect the probability of words of length $n+1$ from states other than $\state$.
\begin{lemma}\label{lemma:all-shorter-agree}
  Let $n\geq 2$ and $\state, \stateB\in\states$. %
  Assume that for all $\state'\in\supp{\trans(\state)}$ and all $\stateB'\in\supp{\trans(\stateB)}$, $(\state', \stateB')\in U_{n-1}$.
  Let $\transB, \transB'\in\cons{\trans}$ such that $\transB$ and $\transB'$ agree over $\states\setminus\{\state\}$.
  Then:
  \begin{itemize}
  \item for all $\word\in\lblSet^n$ and all $\stateC\in\states$, $\probLV{\stateC}{\transB}(\word)=\probLV{\stateC}{\transB'}(\word)$;
  \item for all $\word\in\lblSet^{n+1}$ and all $\stateC\in\states$ such that $\stateC\neq\state$, $\probLV{\stateC}{\transB}(\word)=\probLV{\stateC}{\transB'}(\word)$.
  \end{itemize}
\end{lemma}
\begin{proof}
  For the first item, we prove the following property by induction on $1\leq m\leq n$: for all $\stateC\in\states$ and all $\word\in\lblSet^m$, $\probLV{\stateC}{\transB}(\word)=\probLV{\stateC}{\transB'}(\word)$.
  For $m=1$, i.e., for words of length $1$, we have $\probLV{\stateC}{\transB'}(a) = \probLV{\stateC}{\transB}(a)$ (which is $1$ if $a = \lbl(\stateC)$ and $0$ otherwise) for all $\stateC \in \states$ and all $a \in \lblSet$.
  
  Now let $1 \leq m < n$ and assume that the claim holds for all words of length $m$ by induction.
  Let $\stateC \in \states$, $a \in\lblSet$, and $\word \in\lblSet^{m}$.
  We must show that $\probLV{\stateC}{\transB}(a\word) = \probLV{\stateC}{\transB'}(a\word)$.
  If $a \neq\lbl(\stateC)$, we have $\probLV{\stateC}{\transB}(a\word) = \probLV{\stateC}{\transB'}(a\word) = 0$.
  We now assume that $a = \lbl(\stateC)$.
  If $\stateC \neq \state$, it follows from $\transB(\stateC) = \transB'(\stateC)$ and the induction hypothesis that
  \begin{equation}\label{equation:all-shorter:eq}
    \probLV{\stateC}{\transB'}(a\word)
    = \sum_{\stateC' \in S}\transB'(\stateC, \stateC')\, \probLV{\stateC'}{\transB'}(\word)
    = \sum_{\stateC' \in \states}\transB(\stateC, \stateC')\, \probLV{\stateC'}{\transB}(\word)
    = \probLV{\stateC}{\transB}(a\word).
  \end{equation}
  We now assume that $\stateC = \state$.
  Fix $\stateB'\in\supp{\trans(\stateB)}$.
  It follows from our assumption on $\state$ and $\stateB$ that for all $\state'\in\supp{\trans(\state)}$, we have $\probLV{\state'}{\transB}(\word) = \probLV{\stateB'}{\transB}(\word)$ (because $\word\in\lblSet^m$ and $m\leq n-1$).
  We obtain that
  \begin{equation*}
    \probLV{\state}{\transB}(a\word)
    =
    \sum_{\state' \in\supp{\trans(\state)}}\transB(s, \state')\, \probLV{\state'}{\transB}(\word) = \probLV{\stateB'}{\transB}(\word) \text{ and }
    \probLV{\state}{\transB'}(a\word)
    = \probLV{\stateB'}{\transB'}(\word).
  \end{equation*}
  It follows from the induction hypothesis that $\probLV{\stateB'}{\transB}(\word) = \probLV{\stateB'}{\transB'}(\word)$, i.e., $\probLV{\state}{\transB}(a\word) = \probLV{\state}{\transB'}(a\word)$.
  This completes the proof of the first item.

  The second item can be shown in the same way as the first inductive case above, by using the first item of the lemma instead of the induction hypothesis and Equation~\eqref{equation:all-shorter:eq}.
\end{proof}

We now formally prove \cref{theorem:ule-states}.

\theoremULStates*
\begin{proof}
  Let $R_1 = U_1$ and, for all $n\geq 1$, $R_{n+1} = \stable^{n}(U_1)$.
  We prove that, for all $n\geq 1$, $R_n = U_n$ by induction.
  For base case $n = 1$, we have we have $R_1 = U_1 = \mathord{\sim_0}$ by definition.  

  In the inductive case, we assume that $R_n = U_n$ and that for all $(\state, \stateB)\in\states^2\setminus U_n$, there exists $\word\in\lblSet^n$ and a $\transB\in\cons{\trans}$ such that $\probLV{\state}{\transB}(\word) \neq \probLV{\stateB}{\transB}(\word)$.
  We prove that  $R_{n+1} = U_{n+1}$ and that for all $(\state, \stateB)\in\states^2\setminus U_{n+1}$, there exists $\word\in\lblSet^{n+1}$ and a $\transB\in\cons{\trans}$ such that $\probLV{\state}{\transB}(\word) \neq \probLV{\stateB}{\transB}(\word)$.
  Let $\state, \stateB \in\states$.
  We prove the inclusions $R_{n+1}\subseteq U_{n+1}$ and $\states^2\setminus R_{n+1} \subseteq \states^2\setminus U_{n+1}$.
  We fix $(\state, \stateB)\in\states^2$ for the remainder of the proof.

  Suppose $(\state, \stateB) \in R_{n+1}$.
  If $(\state, \stateB) \in \diag$, i.e., $\state=\stateB$, then $(\state, \stateB) \in U_{n+1}$ by definition.
  We now assume that $\state \neq \stateB$.
  Since $R_{n+1} \subseteq R_1$, we have $\lbl(\state) = \lbl(\stateB)$.
  Let $\transB \in \cons{\trans}$. Then
  \begin{equation}
    \label{equation:support-equiv}
    \supp{\trans(s)} \times \supp{\trans(t)} \subseteq R_n = U_n
  \end{equation}
  Let $a\in\lblSet$ and $\word \in \lblSet^{n}$.
  If $a\neq\lbl(\state)=\lbl(\stateB)$, then $\probLV{\state}{\transB}(a\word) = \probLV{\stateB}{\transB}(a\word) = 0$.
  Otherwise, if $a = \lbl(\state) = \lbl(\stateB)$, then
  \[
    \probLV{\state}{\transB}(a\word) = \sum_{\stateC \in \supp{\transB(s)}} \transB(s, \stateC)\, \probLV{\stateC}{\transB}(\word)
  \]
  and
  \[
    \probLV{\stateB}{\transB}(a\word) = \sum_{\stateC \in \supp{\transB(t)}} \transB(t, \stateC)\, \probLV{\stateC}{\transB}(\word)
  \]
  By \eqref{equation:support-equiv}, $\probLV{\stateC}{\transB}(\word)$ is identical for all $\stateC\in\supp{\transB(\state)}\cup\supp{\transB(\stateB)}$, thus $(\state, \stateB) \in U_{n+1}$.

  We now let $(\state, \stateB)\notin R_{n+1}$ and show that there exists $\word'\in\lblSet^{n+1}$ and a transition function $\transB\in\cons{\trans}$ such that $\probLV{\state}{\transB}(\word')\neq\probLV{\stateB}{\transB}(\word')$, thereby proving that $(\state, \stateB)\notin U_{n+1}$.
  If $(\state, \stateB)\notin R_n = U_n$, then the result follows directly from the induction hypothesis.
  We therefore assume that $(\state, \stateB) \in U_n$.
  In particular, we have $\lbl(\state) = \lbl(\stateB)$; we let $a = \lbl(\state)$ denote this common label.

  By definition of $\stable$, $\state \neq \stateB$ and there exists $(\state', \stateB') \in \supp{\trans(\state)} \times \supp{\trans(\stateB)}$ such that $(\state', \stateB')\notin R_n = U_n$.
  Since $(\state',\stateB') \not\in U_n$, there exist $\word\in\lblSet^{n}$ and $\transB \in \cons{\trans}$ such that $\probLV{\state'}{\transB}(\word) \neq \probLV{\stateB'}{\transB}(\word)$.
  We distinguish two cases.
  If $\probLV{\state}{\transB}(a\word)\neq \probLV{\stateB}{\transB}(a\word)$, then the proof is finished.

  We now assume that $\probLV{\state}{\transB}(a\word) = \probLV{\stateB}{\transB}(a\word)$.
  Then there exists $\state''\in \supp{\transB(\state)}$ such that $\probLV{\state''}{\transB}(\word) \neq \probLV{\state'}{\transB}(\word)$ or there exists $\stateB'' \in\supp{\transB(\stateB)}$ such that $\probLV{\stateB''}{\transB}(\word) \neq \probLV{\stateB'}{\transB}(\word)$ (otherwise, we would have $\probLV{\state}{\transB}(a\word) = \probLV{\state'}{\transB}(\word) \neq \probLV{\stateB'}{\transB}(\word) = \probLV{\stateB}{\transB}(a\word)$, which contradicts $\probLV{\state}{\transB}(a\word) = \probLV{\stateB}{\transB}(a\word)$).
  We assume without loss of generality that there exists $\state'' \in \supp{\transB(s)}$ such that 
\begin{equation}
\label{equation:distinguish}
\probLV{\state'}{\transB}(\word) \neq \probLV{\state''}{\transB}(\word).
\end{equation}

For all $\delta\in\ccInt{0}{\transB(\state, \state'')}$, we define $\transB_\delta \in\cons{\trans}$ by, for all $\stateC, \stateC'\in\states$,
\[
\transB_\delta(\stateC, \stateC')=
\begin{cases}
\transB(\stateC, \stateC') + \delta \qquad & \text{if } (\stateC, \stateC') = (\state, \state')\\
\transB(\stateC, \stateC') - \delta & \text{if } (\stateC, \stateC') = (\state, \state'')\\
\transB(\stateC, \stateC') & \text{otherwise}.
\end{cases}
\]
Intuitively, $\transB_\delta$ redirects some $\delta$ of the outgoing probability of $\state$ from $\state''$ to $\state'$.
We show that $\probLV{\state}{\transB_\delta}(a\word) \neq \probLV{\stateB}{\transB_\delta}(a\word)$ if $\delta>0$, which establishes $(\state, \stateB) \notin U_{n+1}$.

We first prove that $\probLV{\stateB}{\transB}(a\word) = \probLV{\stateB}{\transB_\delta}(a\word)$ (i.e., perturbing $\transB$ as above does not change the probability of $a\word$ from $\stateB$) and that $\probLV{\stateC}{\transB}(\word) = \probLV{\stateC}{\transB_\delta}(\word)$ for all $\stateC\in\states$ and all $\delta\in\ccInt{0}{\transB(\state, \state'')}$.
We distinguish two cases, depending on whether $n=1$ or $n\geq 2$.
If $n=1$, then $\word\in\lblSet$, and the claim follows immediately from the facts that the probability of words of length $1$ does not depend on the transition function from any state, and the fact that $\transB$ and the functions $\transB_\delta$ agree on $\stateB\neq\state$.
Assume now that $n\geq 2$.
We apply Lemma~\ref{lemma:all-shorter-agree} to compare the probability of words under $\transB$ and $\transB_\delta$, as $(\state, \stateB)\in R_n$ implies that $\supp{\trans(\state)}\times\supp{\trans(\stateB)}\subseteq R_{n-1} = U_{n-1}$ (by the inductive hypothesis):
Lemma~\ref{lemma:all-shorter-agree} implies that for all $\delta\in\ccInt{0}{\transB(\state, \state'')}$ and all $\stateC\in\states$, $\probLV{\stateB}{\transB_\delta}(a\word) = \probLV{\stateB}{\transB}(a\word)$ and $\probLV{\stateC}{\transB_\delta}(\word) = \probLV{\stateC}{\transB}(\word)$.

  We obtain, by definition of $\transB_\delta$ and the claim shown above (applied in sequence), that
  \begin{align*}
    \probLV{\state}{\transB_\delta}(a\word)
    & =
      \sum_{\stateC\in\states}\transB(\state, \stateC)\cdot \probLV{\stateC}{\transB_\delta}(\word)
      + \delta\cdot \probLV{\state'}{\transB_\delta}(\word)
      - \delta\cdot \probLV{\state''}{\transB_\delta}(\word) \\
    & = \sum_{\stateC\in\states}\transB(\state, \stateC)\cdot \probLV{\stateC}{\transB}(\word)
      + \delta\cdot \left(\probLV{\state'}{\transB}(\word)
      - \probLV{\state''}{\transB}(\word)\right) \\
    & = \probLV{\state}{\transB}(a\word)
      + \delta\cdot \left(\probLV{\state'}{\transB}(\word)
    - \probLV{\state''}{\transB}(\word)\right).
  \end{align*}

  It follows from the above, $\probLV{\state}{\transB}(a\word) = \probLV{\stateB}{\transB}(a\word)$ and $\probLV{\state'}{\transB}(\word)\neq \probLV{\state''}{\transB}(\word)$ that, for all $0 <\delta \leq \transB(\state, \state'')$, we have
  \[
    \probLV{\state}{\transB_\delta}(a\word) - \probLV{\stateB}{\transB_\delta}(a\word) = \delta\cdot \left(\probLV{\state'}{\transB}(\word)
      - \probLV{\state''}{\transB}(\word)\right) \neq 0.
  \]
  In particular, we have $\probLV{\state}{\transB_\delta}(a\word) \neq \probLV{\stateB}{\transB_\delta}(a\word)$ for $\delta=\transB(\state, \state'')> 0$, which shows that $(\state, \stateB)\notin U_{n+1}$.
  This completes the proof.
\end{proof}

\subsection{Proof of \cref{proposition:ule:epsilon} and \cref{proposition:ule:deterministic}}
This section presents proofs of \cref{proposition:ule:epsilon} and \cref{proposition:ule:deterministic}.
We prove both theorems using an adaptation of the proof of \cref{theorem:ule-states}.
As explained in the main text, we exploit the fact that whenever we perturb transition functions in the proof \cref{theorem:ule-states} to obtain transition functions that witness non-UL of two states, any of the (non-zero) perturbations we introduce is a suitable witness.
In the upcoming proof of \cref{proposition:ule:epsilon}, we rely on small perturbations.
In contrast, in the subsequent proof of \cref{proposition:ule:deterministic}, we exploit maximal perturbations of deterministic transition functions to obtain deterministic transition functions.

We first prove \cref{proposition:ule:epsilon}.
\propositionEpsilonU*

\begin{proof}
  If two states are UL, then they are $\varepsilon$-UL for all $\varepsilon > 0$: $\varepsilon$-UL requires language equivalence under fewer transition functions than UL in general.

  We now establish that if two states $\state$ and $\stateB$ are not UL, i.e., if $(\state, \stateB)\notin U_n$ for some $n\in\INpos$, then, for all $\varepsilon > 0$, they are not $\varepsilon$-UL.
  We present an argument by induction.
  We show that for all $n\in\INpos$, all $(\state, \stateB)\in\states^2\setminus U_n$ and all $\varepsilon > 0$, there exist a transition function $\transB^\varepsilon\in\cons{\trans}$ in the $\varepsilon$-neighbourhood of $\trans$ (i.e., such that for all $\stateC,\stateC'\in\states$, $|\trans(\stateC, \stateC')-\transB^\varepsilon(\stateC, \stateC')|\leq\varepsilon$) and a word $\word\in\lblSet^n$ such that $\probLV{\state}{\transB^\varepsilon}(\word)\neq\probLV{\stateB}{\transB^\varepsilon}(\word)$.

  For the base case, let $n=1$ and let $(\state, \stateB)\notin U_1$.
  We recall that $U_1$ is label equality.
  Hence it suffices to choose $\word = \lbl(\state)$ and, for all $\varepsilon > 0$, let $\transB^\varepsilon = \trans$ for this case.

  We now assume that the claim holds for elements of $\states^2\setminus U_n$ and show that it holds for elements of $\states^2\setminus U_{n+1}$.
  Let $(\state, \stateB)\in\states^2\setminus U_{n+1}$.
  If $(\state, \stateB)\in\states^2\setminus U_{n}$, the result follows from the induction hypothesis.

  We assume for the remainder of the argument that $(\state, \stateB)\in U_{n}$.
  In particular, $\state$ and $\stateB$ have the same label; we let $a = \lbl(\state) = \lbl(\stateB)$.
  By \cref{theorem:ule-states}, there exists $(\state', \stateB')\in \supp{\trans(\state)}\times \supp{\trans(\stateB)}$ such that $(\state', \stateB')\in U_{n-1}\setminus U_n$ (letting $U_0 = \states^2$).
  
  Fix $\varepsilon > 0$ and let $\eta = \frac{\varepsilon}{2}$.
  The induction hypothesis yields a transition function $\transB^\eta\in\cons{\trans}$ in the $\eta$-neighbourhood of $\trans$ and a word $\word\in\lblSet^{n}$ such that $\probLV{\state'}{\transB^\eta}(\word)\neq\probLV{\stateB'}{\transB^\eta}(\word)$.
  We show that the word $a\word\in\lblSet^{n+1}$ can be used to witness that $\state$ and $\stateB$ are not $\varepsilon$-UL.

  If $\probLV{\state}{\transB^{\eta}}(a\word) \neq \probLV{\stateB}{\transB^{\eta}}(a\word)$, letting $\transB^\varepsilon = \transB^{\eta}$ is enough to complete the inductive proof.
  We assume therefore that $\probLV{\state}{\transB^{\eta}}(a\word) = \probLV{\stateB}{\transB^{\eta}}(a\word)$.
  It follows from $\probLV{\state'}{\transB^{\eta}}(\word) \neq \probLV{\stateB'}{\transB^{\eta}}(\word)$ that there exists $\state''\in\supp{\transB^{\eta}(\state)}$ such that $\probLV{\state'}{\transB^{\eta}}(\word) \neq \probLV{\state''}{\transB^{\eta}}(\word)$ or there exists $\stateB''\in\supp{\transB^{\eta}(\stateB)}$ such that $\probLV{\stateB'}{\transB^{\eta}}(\word) \neq \probLV{\stateB''}{\transB^{\eta}}(\word)$.
  We assume without loss of generality that the former property holds.
  For all $\delta\in\ccInt{0}{\transB^{\eta}(\state, \state'')}$, we define the transition function $\transB^{\eta}_{\delta}$ by letting, for all $\stateC, \stateC'\in\states$,
  \[
    \transB^{\eta}_\delta(\stateC, \stateC')=
    \begin{cases}
      \transB^{\eta}(\stateC, \stateC') + \delta \qquad & \text{if } (\stateC, \stateC') = (\state, \state')\\
      \transB^{\eta}(\stateC, \stateC') - \delta & \text{if } (\stateC, \stateC') = (\state, \state'')\\
      \transB^{\eta}(\stateC, \stateC') & \text{otherwise}.
    \end{cases}
  \]
  As in the proof of \cref{theorem:ule-states}, we can apply Lemma~\ref{lemma:all-shorter-agree} (if $n\geq 2$) or use a direct argument (if $n=1$) to prove that, for all $\delta\in\ccInt{0}{\transB^{\eta}(\state, \state'')}$, we have
  \[
    \probLV{\state}{\transB^{\eta}_\delta}(a\word) -
    \probLV{\stateB}{\transB^{\eta}_\delta}(a\word) =
    \delta\cdot (\probLV{\state'}{\transB^{\eta}}(\word) - \probLV{\state''}{\transB^{\eta}}(\word)),
  \]
  hence $\probLV{\state}{\transB^{\eta}_\delta}(a\word) \neq \probLV{\stateB}{\transB^{\eta}_\delta}(a\word)$ for any $\delta\in\ocInt{0}{\transB^{\eta}(\state, \state'')}$.
  It suffices to let $\transB^\varepsilon = \transB^{\eta}_{\min\{\eta, \transB^{\eta}(\state, \state'')\}}$ to finish the proof.
\end{proof}

We now prove \cref{proposition:ule:deterministic}.

\propositionDeterministicU*
\begin{proof}
  We first show that there exists a deterministic transition function witnessing the non-equivalence of non-UL states.
  This can be shown by adapting the inductive argument from the proof of \cref{proposition:ule:epsilon}.
  To limit redundancy, we only comment how to adapt the inductive proof.
  For the base case $n=1$, any deterministic transition function witnesses non-equivalence of pairs $(\state, \stateB)\notin U_1$.
  For the inductive case, we assume that the transition function given by the induction hypothesis is deterministic.
  In the case in which we perturb the transition function obtained by induction to construct a witness of non-equivalence, there is a suitable deterministic perturbation.
  Indeed, the interval in which $\delta$ is chosen is $\ccInt{0}{1}$ because the function given by the induction hypothesis is deterministic.
  This ends the first part of the proof.

  Let $(\state, \stateB)\notin U_n$.
  The adapted argument above establishes the existence of a deterministic $\transB\in\cons{\trans}$ and $\word\in\lblSet^n$ such that $\probLV{\state}{\transB}(\word) \neq \probLV{\stateB}{\transB}(\word)$.
  Because $\transB$ is deterministic, we have $\probLV{\state}{\transB}(\word), \probLV{\stateB}{\transB}(\word)\in\{0, 1\}$.
  If $\probLV{\state}{\transB}(\word) = 1$, the proof is complete.
  Assume otherwise.
  Let $\word'\in\lblSet^n$ be the unique word of length $n$ such that $\probLV{\state}{\transB}(\word') = 1$ (the existence of such a word follows from $\transB$ being deterministic).
  Necessarily, we have $\word'\neq\word$.
  We conclude the proof by observing that $\probLV{\state}{\transB}(\word') = 1$ and $\probLV{\stateB}{\transB}(\word') = 0$.
\end{proof}

\subsection{Proof of Lemma~\ref{lemma:ule:state-witnesses}}
\label{appendix:ue:nl-algorithm}
We prove Lemma~\ref{lemma:ule:state-witnesses} using an inductive argument based on the characterisation of the UL relation over states from \cref{theorem:ule-states}.

\lemmaULStateWitnesses*
\begin{proof}
  We call histories $\hist = \state_1\state_2\ldots\state_\iLast$ and $\hist'=\stateB_1\stateB_2\ldots\stateB_\iLast$ with $\state_1 = \state$, $\stateB_1=\stateB$ and $\iLast\leq|\states|$ such that $\lbl(\hist) \neq \lbl(\hist')$ and for all $1\leq\iPos\leq\iLast$, $\state_\iPos\neq\stateB_\iPos$ witnesses (of non-equivalence) for $(\state, \stateB)$ in the following.
  We let $U_0 = S^2$ and let $U_1\subseteq\states\times\states$ denote the label equivalence relation.
  For all $\iSeq\geq 2$, we let $U_\iSeq = \stable(U_{\iSeq-1})$ be the UL relation over words of length $\iSeq$ (cf.~\cref{theorem:ule-states}).

  Assume first that $\state$ and $\stateB$ are not UL.
  We construct witnesses by induction.
  Let $0\leq \iSeq <|\states|$ such that $(\state, \stateB)\in U_\iSeq\setminus U_{\iSeq+1}$.
  If $\iSeq = 0$, then $\lbl(\state)\neq\lbl(\stateB)$ and we let $\hist = \state$ and $\hist' = \stateB$.
  We now assume that $\iSeq\geq 1$ and that witnesses of length $\iSeq$ exist for any $(\state', \stateB')\in U_{\iSeq-1}\setminus U_{\iSeq}$ (by induction).
  Because $\state$ and $\stateB$ are not UL, it must be the case that $\state\neq\stateB$.
  Thus, it follows from the definition of $\stable$ that there exist $\state'\in\supp{\trans(\state)}$ and $\stateB'\in\supp{\trans(\stateB)}$ such that $(\state', \stateB')\notin U_{\iSeq}$.
  We observe that $(\state', \stateB')\in U_{\iSeq-1}$ as $(\state, \stateB)\in U_{\iSeq}$.
  We can thus apply the induction hypothesis to obtain witness histories $\hist$ and $\hist'$ for the non-equivalence of $\state'$ and $\stateB'$ of length $\iSeq$.
  It is easy to check that the histories $\state\hist$ and $\stateB\hist'$ are witnesses of non-equivalence for $\state$ and $\stateB$ of length $\iSeq+1$.

  We now assume that there exist witness histories $\state_1\ldots\state_\iLast$ and $\stateB_1\ldots\stateB_\iLast$ with $\state_1=\state$ and $\stateB_1=\stateB$.
  We assume without loss of generality that for all $1\leq\iPos\leq\iLast-1$, $\lbl(\state_\iPos) = \lbl(\stateB_\iPos)$: if one removes the suffixes after the first position $\iPos$ such that $\lbl(\state_\iPos) \neq \lbl(\stateB_\iPos)$, the resulting histories are still witnesses.
  We use the fixpoint characterisation through $\stable$ to show that $(\state, \stateB)$ are not UL.
  On the one hand, we have $(\state_\iLast, \stateB_\iLast)\notin U_1$ ($U_1$ is the label equality relation).
  On the other hand, by definition of $\stable$, for all $1\leq\iPos\leq\iLast-1$, $(\state_{\iLast-\iPos+1}, \stateB_{\iLast-\iPos+1})\notin U_{\iPos}$ implies $(\state_{\iLast-\iPos}, \stateB_{\iLast-\iPos})\notin U_{\iPos+1}$.
  It follows thus that $(\state, \stateB)=(\state_1, \stateB_1)\notin U_{\iLast}$, hence $\state$ and $\stateB$ are not UL.
\end{proof}

\begin{algorithm}[t]
  \caption{Non-deterministic algorithm to decide if two states are not UL.}
  \label{algorithm:states-nonule}
  \KwData{A labelled Markov chain $\mchain = \lmc$ and states $\state,\stateB\in\states$.}
  \For{$i \gets 1$ \KwTo $|\states|$}{
    \If{$\state = \stateB$}{
      \Return reject\;
    }
    \If{$\ell(\state) \neq \ell(\stateB)$}{
      \Return accept\Comment*[r]{Witness histories for non-UL constructed.}
    }
    \KwGuess $\state \in \supp{\trans(\state)}$\;
    \KwGuess $\stateB \in \supp{\trans(\stateB)}$\;
  }
  \Return reject.
\end{algorithm}

\subsection{\nlogspace-completeness of deciding UL for states}\label{appendix:ue:hardness}

We establish the $\nlogspace$-completeness of the problem of deciding whether two states are UL in this section.
The procedure outlined in the main text, summarised in \cref{algorithm:states-nonule}, yields $\nlogspace$-membership of the problem.
We therefore focus on $\nlogspace$-hardness in the remainder of the section.

We provide a reduction from the complement of directed graph reachability (also known as st-connectivity) to the UL problem.
Directed graph reachability is defined as follows: given a directed graph $G = (V, E)$, where $V$ is a finite set of vertices and $E\subseteq V\times V$ is a set of edges, and source and target vertices $s$ and $t$, decide whether there exists a path from $s$ to $t$.
Directed graph reachability is $\nlogspace$-complete~\cite{Jon75} and therefore, its complement also is as $\nlogspace=\conlogspace$~\cite{DBLP:journals/siamcomp/Immerman88,DBLP:journals/acta/Szelepcsenyi88}.
We formalise a logspace reduction from the complement of directed graph reachability to the UL problem below.

\begin{lemma}\label{lemma:hardness:ule}
  Deciding if two states are UL is $\nlogspace$-hard.
\end{lemma}
\begin{proof}
  Let $G = (V, E)$ be the input graph, $s\in V$ the source vertex and $t\in V$ the target vertex.
  We can assume that in $G$ every vertex has at least one outgoing edge (add a self loop if necessary).
  We derive a labelled Markov chain $\mchain = \lmc$ from $G$ as follows.
  We let $\states = V\cup\{\bot\}$ where $\bot\notin V$ is a fresh state.
  The transition function $\trans$ is such that, for all $v\in V$, $\trans(v)$ is a uniform distribution over the successors of $v$ in $G$ and $\trans(\bot, \bot) = 1$.
  For labels, we let $\lblSet=\{0, 1\}$, let $\lbl(t) = 1$ and $\lbl(v) = 0$ for all $v\in \states\setminus\{t\}$.
  We reduce the instance of graph non-reachability given by $G$, $s$ and $t$ to the problem of determining if $s$ and $\bot$ are UL in $\mchain$.

  We prove the correctness of the reduction by showing that states $s$ and $\bot$ are UL if and only if there exists no path from $s$ to $t$.
  First, assume that there exists a path from $s$ to $t$ in $G$.
  Then there exists a transition function $\transB\in\cons{\trans}$ under which $t$ is reachable from $s$.
  It follows that $s$ and $\bot$ are not language-equivalent under $\transB$ as label $1$ can be reached from $s$ but not from $\bot$.
  Conversely, assume that there exists no path from $s$ to $t$ in $G$.
  In that case, for all $\transB\in\cons{\trans}$, the only word that can occur from $s$ and $\bot$ under $\transB$ in $\mchain$ is $0^\omega$, hence $s$ and $\bot$ are $\transB$-language-equivalent.
  This shows that $s$ and $\bot$ are UL.
\end{proof}

\subsection{Proof of Theorem~\ref{theorem:ule-distributions}}
\label{appendix:ul-distributions}

Let $U$ denote the UL relation restricted to states.
The main result of this section (\cref{lemma:ule:distributions}) implies the following theorem from the main text.
\theoremULDistributions*

Intuitively, \cref{lemma:ule:distributions} can be seen as a variant of \cref{theorem:ule-distributions} for words of bounded length.
We prove this result by induction using the relations $U_n\subseteq\states\times\states$ defined in \cref{section:ule:states}.
For all $n\in\INpos$, we define the counterpart of $U_n$ for distributions as the set
\[
    V_n = \left\{(\measure, \measureB)\in\dist{\states}\times\dist{\states}\mid
  \forall\word\in\lblSet^n,\,\forall\transB\in\cons{\trans},\,
  \probLV{\measure}{\transB}(\word) = \probLV{\measureB}{\transB}(\word)
  \right\}.
\]

Our goal is to prove that for all distributions $\measure, \measureB\in\dist{\states}$, $(\measure, \measureB)\in V_n$ if and only if $\measure$ and $\measureB$ agree over all equivalence classes of $U_n$.
To prove this claim, we first formalise two properties related to equivalence classes of $U_n$.
For any $n\geq 2$, $a\in\lblSet$ and $\eClassB\in\states/U_{n-1}$, we let $\eClassTwo{a}{\eClassB}{n} = \{ s \in \states \mid \lbl(\state) = a$ and $ \supp{\trans(\state)}\subseteq \eClassB\}$.
We first show the following property of these sets.

\begin{lemma}\label{lemma:ule:type-two-classes}
  Let $n\geq 2$.
  For all $\eClassB\in\states/U_{n-1}$ and all $a\in\lblSet$, if $\eClassTwo{a}{\eClassB}{n}$ is nonempty, then $\eClassTwo{a}{\eClassB}{n}\in\states/U_n$.
\end{lemma}
\begin{proof}
  Let $\eClass = \eClassTwo{a}{\eClassB}{n}$ and assume that $\eClass$ is nonempty.
  We must show that for all $\state\in\eClass$ and all $\stateB\in\states$, we have $(\state, \stateB)\in U_{n}$ if and only if $\stateB\in\eClass$.
  If $\stateB\in\eClass$, as $\lbl(\state) = \lbl(\stateB)$ and all successors of both $\state$ and $\stateB$ are in $D\in\states/U_{n-1}$, we directly obtain $(\state, \stateB)\in U_n$.

  We now assume that $\stateB\notin\eClass$.
  In particular, $\state\neq\stateB$.
  If $\lbl(\state)\neq\lbl(\stateB)$, we directly obtain $(\state, \stateB)\notin U_n$.
  Otherwise, we have $\lbl(\state)=\lbl(\stateB)$ and therefore $\supp{\trans(\stateB)}\nsubseteq D$, i.e., there exists $\stateB'\in\supp{\trans(\stateB)}$ such that $(\state', \stateB')\notin U_{n-1}$ for all $\state'\in\supp{\trans(\state)}\subseteq D$.
  It follows from \cref{theorem:ule-states} that $(\state, \stateB)\notin\stable(U_{n-1}) = U_n$.
\end{proof}

We now prove that, for all $n\geq 2$, the equivalence classes of $U_n$ can be classified into two categories, called type-1 or type-2 classes.
Each type is defined as follows.
Let $\eClass\in\states/U_n$.
The class $\eClass$ is a type-1 class if there exists some $\state\in\states$ such that $\eClass = \{\state\}$ and $\supp{\trans(\state)}$ intersects at least two distinct $U_{n-1}$ classes.
The class $\eClass$ is a \textit{type-2 class} if there exist $a \in \lblSet$ and $\eClassB \in \states/U_{n-1}$ such that $\eClass=\eClassTwo{a}{\eClassB}{n}$.
We formally prove that all classes are either type-1 or type-2 classes.
\begin{lemma}\label{lemma:ule:class-types}
  Let $n\geq 2$.
  For all $\eClass\in\states/U_n$, $C$ is either a type-1 or type-2 equivalence class.
\end{lemma}
\begin{proof} 
  Assume first that $\eClass$ is a singleton $\{\state\}$ for some $\state\in\states$.
  If $\eClass$ is not a type-1 class, then all successors of $\state$ are in a single equivalence class of $U_{n-1}$, hence $\eClass$ is a type-2 class.

  It remains to show that all non-singleton classes are type-2 classes.
  Assume that $\eClass$ is not a singleton set.
  All states in $\eClass$ share the same label $a$.
  Furthermore, by \cref{theorem:ule-states} and the definition of $\stable$, the set $\bigcup_{\state\in\eClass}\supp{\trans(\state)}$ is a subset of some $\eClassB\in\states/U_{n-1}$.
  We obtain that $\eClass = \eClassTwo{a}{\eClassB}{n}$.
\end{proof}

We now show that for all $n\in\INpos$ and $\measure$, $\measureB\in\dist{\states}$, $\measure$ and $\measureB$ assign the same probability to each equivalence class of $U_n$ if and only if $\measure$ and $\measureB$ agree over words of length $n$ under all $\transB\in\cons{\trans}$, i.e., $(\measure,\measureB)\in V_n$.

First, we note that for all $n\in\INpos$ and $\measure$, $\measureB\in\dist{\states}$, if $\measure$ and $\measureB$ assign the same probability to each equivalence class of $U_n$, then $\measure$ and $\measureB$ agree over words of length $n$ under all $\transB\in\cons{\trans}$ by the definition of $U_n$.
This shows that $(\measure,\measureB)\in V_n$.

We then proceed by induction to prove that for all $n\in\IN$, all $(\measure, \measureB) \in V_n$ are such that $\measure$ and $\measureB$ agree over $S/U_n$.
The base case is direct: $U_1$ is label equality and $V_1$ is the set of pairs of distributions that assign the same probability to each label.
For the inductive case, we assume that all pairs in $V_n$ agree over $S/U_n$ and show that all pairs in $V_{n+1}$ agree over $S/U_{n+1}$.
We fix $(\measure, \measureB)\in V_{n+1}$ and a class $\eClass\in S/U_{n+1}$.
We consider three cases.
First, if $\eClass\in S/U_{n}$, it follows from $(\measure, \measureB)\in V_{n+1}\subseteq V_n$ and the induction hypothesis that $\measure(\eClass) = \measureB(\eClass)$.
For the two remaining cases, we assume that $\eClass$ is a proper subset of a class in $S/U_n$.

We first show that $\measure$ and $\measureB$ must agree on $\eClass$ if it is a type-1 class of the form $\eClass = \{\state\}$.
We build on the fact that there exists $\stateB\neq\state$ such that $(\state, \stateB)\in U_n\setminus U_{n+1}$ to show that for some $\transB\in\cons{\trans}$ and $\word\in\lblSet^{n+1}$, perturbing the outgoing probabilities from $\state$ changes the probability of $\word$ only from $\state$ and nowhere else (cf.~\cref{lemma:all-shorter-agree}).
We conclude from these perturbed transition functions that $(\measure, \measureB)\in V_{n+1}$ implies that $\measure(\state) = \measureB(\state)$.

It remains to discuss the case where $\eClass$ is a type-2 class, i.e., where there exist $a \in \lblSet$ and $\eClassB\in \states/U_n$ such that $\eClass = \eClassTwo{a}{\eClassB}{n}$.
We may assume that $\measure$ and $\measureB$ only assign positive probability to type-2 classes in light of the previous case, and that $\measure$ and $\measureB$ only assign positive probability to states with label $a$ (e.g., we deal with one label at a time).
The main idea is to derive from $\measure$ and $\measureB$ some pair $(\measure', \measureB')\in V_n$ from which we can infer that $\measure$ and $\measureB$ agree on $S/U_{n+1}$ using the induction hypothesis.
We select $\measure'$ and $\measureB'$ such that, for all $\eClassC\in \states/U_n$, $\measure'(\eClassC) = \measure(\eClassTwo{a}{\eClassC}{n+1})$ and $\measureB'(\eClassC) = \measureB(\eClassTwo{a}{\eClassC}{n+1})$.
We then use the fact that $(\measure,\measureB)\in V_{n+1}$ to obtain that $(\measure', \measureB')\in V_n$.
We conclude from the induction hypothesis that, for all $\eClassC\in \states/U_n$, $\measure(\eClassTwo{a}{\eClassC}{n+1}) = \measure'(\eClassC) = \measureB'(\eClassC) = \measureB(\eClassTwo{a}{\eClassC}{n+1})$.

We formally prove this below.

\begin{lemma}
\label{lemma:ule:distributions}
  For all $n\in\INpos$ and all $\measure, \measureB\in\dist{\states}$,
  $(\measure, \measureB)\in V_n$ if and only if $\measure(\eClass) = \measureB(\eClass)$ for all $\eClass \in \states/U_n$.
\end{lemma}
\begin{proof}
  For all $n\in\INpos$, let $W_n = \{(\measure, \measureB) \in\dist{\states}^2\mid \measure(\eClass) = \measureB(\eClass) \text{ for all $\eClass \in \states/U_n$}\}$.
  We must prove that $V_n = W_n$ for all $n\in\INpos$.

  We first show that for all $n\in\INpos$, $W_n \subseteq V_n$.
  Let $n\in\INpos$ and $(\measure, \measureB)\in W_n$.
  By definition, for all $\eClass \in \states/U_{n}$, $\measure(\eClass) = \measureB(\eClass)$.
  Let $\transB \in \cons{\trans}$ and $\word \in L^{n}$.
  For each equivalence class $C \in \states/U_{n}$, pick an arbitrary representative $s_C \in C$.
  By definition of $U_n$, we have
  \begin{equation*}
    \probLV{\measure}{\transB}(\word) = \sum_{s\in \states} \measure(s)\, \probLV{s}{\transB}(\word) = \sum_{C\in \states/U_{n}} \measure(C)\, \probLV{s_C}{\transB}(\word),
  \end{equation*}
  and similarly, we have
  \begin{equation*}
    \probLV{\measureB}{\transB}(\word) = \sum_{s\in \states} \measureB(s)\, \probLV{s}{\transB}(\word) = \sum_{C\in \states/U_{n}} \measureB(C)\, \probLV{s_C}{\transB}(\word).
  \end{equation*}
  It follows from the above and $(\measure,\measureB) \in W_{n}$ that $\probLV{\measure}{\transB}(\word) = \probLV{\measureB}{\transB}(\word)$.
  This ends the proof that $(\measure,\measureB) \in V_{n}$.

We now prove by induction that for all $n\in\INpos$, $V_n\subseteq W_n$.
We first consider the base case $n = 1$.
Let $(\measure, \measureB) \in V_1$.
Clearly, $(\measure, \measureB) \in V_1$ if and only if $\measure$ and $\measureB$ assign the same mass to every label.  Since $U_1$ is equality of labels, we have that $(\measure, \measureB) \in V_1$ if and only if, for all $\eClass \in \states/U_1$, $\measure(\eClass) = \measureB(\eClass)$, i.e., $(\measure,\measureB)\in W_1$.

In the inductive case, assume that $V_n\subseteq W_n$ and let us show that $V_{n+1}\subseteq W_{n+1}$.
Let $(\measure, \measureB) \in V_{n+1}$.
Since $V_{n+1} \subseteq V_n$, by the induction hypothesis, we have $\measure(\eClassB) = \measureB(\eClassB)$ for all $\eClassB \in \states/U_n$.
We now fix $\eClass \in \states/U_{n+1}$.
If $\eClass \in \states/U_n$, then it follows from the induction hypothesis that $\measure(\eClass) = \measureB(\eClass)$.
We now assume that $\eClass\not\in \states/U_n$ and distinguish two cases depending on the type of $\eClass$.

\subparagraph*{The class $\eClass$ is a type-1 class.}
Assume that there exists $\state\in\states$ such that $\eClass = \{\state\}$ and $\supp{\trans(\state)}$ intersects at least two distinct $U_n$-classes $\eClassB_1\neq\eClassB_2$.
Fix $\state'\in \eClassB_1$ and $\state''\in \eClassB_2$.
Since $(\state',\state'') \notin U_n$, there exist $\transB \in \cons{\trans}$ and $\word \in L^n$ with $\probLV{\state'}{\transB}(\word) \neq \probLV{\state''}{\transB}(\word)$.
We assume that $\state''\in\supp{\transB(\state)}$.
  This is without loss of generality: it follows from $\probLV{\state'}{\transB}(\word) \neq \probLV{\state''}{\transB}(\word)$ that there exists $\stateC\in\supp{\transB(\state)}$ such that $\probLV{\state'}{\transB}(\word) \neq \probLV{\stateC}{\transB}(\word)$ (in which case we replace $\state''$ with $\stateC$ to enforce the assumption) or there exists $\stateC\in\supp{\transB(\state)}$ such that $\probLV{\state''}{\transB}(\word) \neq \probLV{\stateC}{\transB}(\word)$ (in which case we respectively replace $\state'$ and $\state''$ with $\state''$ and $\stateC$ to enforce the assumption).

It follows from $\eClass \not\in \states/U_n$ that there exists $\stateB \in \states\setminus\{\state\}$ such that $(\state, \stateB) \in U_n$.
Fix $t' \in \supp{\trans(t)}$.

Let $a = \lbl(\state)$.
For any $\varepsilon\in\ccInt{0}{\transB(\state, \state'')}$, we define $\transB_\varepsilon \in \cons{\trans}$ by perturbing $\transB$ at $\state$ as in the proof of \cref{theorem:ule-states} by letting, for all $\stateC, \stateC'\in\states$,
\[
\transB_\varepsilon(\stateC, \stateC')=
\begin{cases}
\transB(\stateC, \stateC') + \varepsilon \qquad & \text{if } (\stateC, \stateC') = (\state, \state')\\
\transB(\stateC, \stateC') - \varepsilon & \text{if } (\stateC, \stateC') = (\state, \state'')\\
\transB(\stateC, \stateC') & \text{otherwise}.
\end{cases}
\]

We claim that for all $\stateC\in\states$ and all $\varepsilon\in\ccInt{0}{\transB(\state, \state'')}$, we have $\probLV{\stateC}{\transB_\varepsilon}(\word) = \probLV{\stateC}{\transB}(\word)$.
If $n=1$, then $\word\in\lblSet$ and the claim is direct.
Assume that $n\geq 2$.
By \cref{theorem:ule-states}, all successors of $\state$ and $\stateB$ are equivalent in the sense of $U_{n-1}$.
Lemma~\ref{lemma:all-shorter-agree} implies the required claim.

We obtain from the claim and the definition of the perturbed transition functions $\transB_\varepsilon$ that, for all $\varepsilon\in\ccInt{0}{\transB(\state, \state'')}$, $\probLV{\state}{\transB_\varepsilon}(a\word) = \probLV{\state}{\transB}(a\word) + \varepsilon\cdot \left(\probLV{\state'}{\transB}(\word) - \probLV{\state''}{\transB}(\word)\right)$ and $\probLV{\stateC}{\transB_\varepsilon}(a\word) = \probLV{\stateC}{\transB_\varepsilon}(a\word)$ (see the proof of \cref{theorem:ule-states} for a similar detailed argument).

It follows from these equations that for all $\varepsilon\in\ccInt{0}{\transB(\state, \state'')}$, we have
\[\probLV{\measure}{\transB_\varepsilon}(a\word) - \probLV{\measureB}{\transB_\varepsilon}(a\word) =
  \probLV{\measure}{\transB}(a\word) -
  \probLV{\measureB}{\transB}(a\word) +
  \varepsilon\cdot
  \left(\measure(\state) - \measureB(\state)\right) \cdot
  \left(\probLV{\state'}{\transB}(\word) -
    \probLV{\state''}{\transB}(\word)\right).
\]
It follows from $(\measure,\measureB)\in V_{n+1}$ that $\probLV{\measure}{\transB_\varepsilon}(a\word) - \probLV{\measureB}{\transB_\varepsilon}(a\word) =     \probLV{\measure}{\transB}(a\word) -
\probLV{\measureB}{\transB}(a\word) = 0$.
We obtain that for all $\varepsilon\in\ccInt{0}{\transB(\state, \state'')}$, we have
\[\varepsilon\cdot
  \left(\measure(\state) - \measureB(\state)\right) \cdot
  \left(\probLV{\state'}{\transB}(\word) -
    \probLV{\state''}{\transB}(\word)\right) = 0.\]
This implies that $\measure(\state) = \measureB(\state)$ (as we have $\probLV{\state'}{\transB}(\word) \neq \probLV{\state''}{\transB}(\word)$ and $\transB(\state, \state'') > 0$).
This ends the proof of the first case.

\subparagraph*{The class $\eClass$ is a type-2 class.}
Assume that there exist $a \in \lblSet$ and $\eClassB\in \states/U_n$ such that $\eClass = \eClassTwo{a}{\eClassB}{n} = \{\state \in \states \mid \ell(\state) = a,\; \supp{\trans(\state)} \subseteq \eClassB\}$.
We construct initial distributions $\measure'$ and $\measureB'$ such that $(\measure', \measureB')\in V_n$, and $\measure'$ and $\measureB'$ respectively assign (multiples of) $\measure(\eClass)$ and $\measureB(\eClass)$ to some $U_n$ class, so that we can obtain $\measure(\eClass) = \measureB(\eClass)$ from the induction hypothesis.
In the following, given an equivalence class $\eClassC$ of $U_{n+1}$ or of $U_n$, we let $\state_\eClassC\in\eClassC$ denote a fixed representative of the class.

For all $\transB\in\cons{\trans}$ and $\word\in\lblSet^n$, it follows from $(\measure, \measureB) \in V_{n+1}$ that $\probLV{\measure}{\transB}(a\word) = \probLV{\measureB}{\transB}(a\word)$ holds, and therefore
\begin{equation}
  \label{equation:type-2:full-sum}
\sum_{\eClass' \in \states/U_{n+1}} (\measure(\eClass') - \measureB(\eClass'))\, \probLV{\state_{\eClass'}}{\transB}(a\word) = 0.
\end{equation}
By the first case, the type-1 class terms vanish in this sum.
We can simplify the type-2 terms by observing that, for all $\transB \in \cons{\trans}$, $\word \in L^n$, $\eClassC\in\states/U_{n}$ and $\state \in \eClassTwo{a}{\eClassC}{n}$, we have
\begin{equation}
\label{equation:type-2-term}
\probLV{\state}{\transB}(a\word) =
\sum_{\stateC \in \supp{\trans(\state)}} \transB(\state,\stateC)\cdot \probLV{\stateC}{\transB}(\word) =
\probLV{\state_E}{\transB}(\word);
\end{equation}
the second equality holds because $\supp{\trans(\state)} \subseteq \eClassC$ and $\eClassC$ is an equivalence class of $U_n$ (the probability of $\word\in\lblSet^n$ is equal from all states in $\eClassC$ under $\transB$).

It follows from Equations~\eqref{equation:type-2:full-sum} and~\eqref{equation:type-2-term} that for all $\transB \in \cons{\trans}$ and $\word \in L^n$,
\begin{equation}
\label{equation:key}
\sum_{E \in \states/U_n} (\measure(\eClassTwo{a}{\eClassC}{n}) - \measureB(\eClassTwo{a}{\eClassC}{n}))\, \probLV{s_\eClassC}{\transB}(\word) = 0.
\end{equation}

We now define the required distributions $\measure'$ and $\measureB'$.
Let $x = \sum_{\eClassC \in \states/U_n} \measure(\eClass_{a,E}) = \sum_{\eClassC \in \states/U_n} \measureB(\eClass_{a,E})$ (the equality follows from both $\measure$ and $\measureB$ assigning the same probability to the set of states with label $a$ and agreeing over type-1 classes). Define $\measure', \measureB' \in \dist{\states}$ by letting, for all $\state\in\states$,
\[
\measure'(\state) =
\begin{cases}
\frac{\measure(\eClassTwo{a}{\eClassC}{n})}{x} & \text{if there exists } \eClassC\in \states/U_{n} \text{ such that } \state\in\eClassC \text{ and } \state = \state_\eClassC\\
0 & \text{otherwise}
\end{cases}
\]
and similarly
\[
\measureB'(\state) =
\begin{cases}
\frac{\measureB(\eClassTwo{a}{\eClassC}{n})}{x} & \text{if there exists } \eClassC\in \states/U_{n} \text{ such that } \state\in\eClassC \text{ and } \state = \state_\eClassC\\
0 & \text{otherwise}
\end{cases}
\]
By construction, for all $\eClassC \in \states/U_n$, $\measure'(\eClassC) = \measure'(\state_\eClassC) = \frac{\measure(\eClassTwo{a}{E}{n})}{x}$ and $\measureB'(\eClassC) = \measureB'(\state_\eClassC) = \frac{\measureB(\eClassTwo{a}{\eClassC}{n})}{x}$.

We obtain that $(\measure', \measureB') \in V_n$ by \eqref{equation:key}, as we have, for all $\transB \in \cons{\trans}$ and all $\word \in L^n$ that
\[
  \probLV{\measure'}{\transB}(\word) - \probLV{\measureB'}{\transB}(\word)
  = \sum_{E \in \states/U_n} (\measure'(E) - \measureB'(E))\, \probLV{s_E}{\transB}(\word) = 0.
\]
Hence, by the induction hypothesis, for all $E \in \states/U_n$, we have $\measure'(E) = \measureB'(E)$. It follows that for all $E \in \states/U_n$,
\[
\frac{\measure(C_{a,E})}{x} = \frac{\measureB(C_{a,E})}{x}.
\]
In particular, $\measure(C_{a,D}) = \measureB(C_{a,D})$.
\end{proof}

We now apply \cref{lemma:ule:distributions} to prove \cref{theorem:ule-distributions}.
\begin{proof}[Proof of \cref{theorem:ule-distributions}]
  We recall that for all $\transB\in\cons{\trans}$ and all $\measure, \measureB\in\dist{\states}$, $\measure\langEquiv{\transB}\measureB$ if and only if, for all $\word\in\bigcup_{n\leq|\states|}\lblSet^{n}$, $\probLV{\measure}{\transB}(\word) = \probLV{\measureB}{\transB}(\word)$ (see, e.g.,~\cite{DBLP:journals/siamcomp/Tzeng92}).
  This implies that $V_{|\states|}$ is the UL relation and that $U_{|\states|}$ is its restriction to states.
  The result follows from \cref{lemma:ule:distributions}.
\end{proof}

\subsection{Discussion of Theorem~\ref{theorem:ule:d:complexity}}

We establish the $\nlogspace$-membership of the problem of deciding whether two distributions are UL in this section.

First, we observe that it is not necessary to compute and store the whole UL relation over states which would require polynomial space to be stored.
Instead, we can check whether two states are UL using an $\nlogspace$ oracle for UL.
This does not increase the space complexity of the resulting algorithm, because $\nlogspace^\nlogspace = \nlogspace$~\cite{DBLP:journals/siamcomp/Immerman88}, i.e., a problem that can be solved by a non-deterministic Turing machine using logarithmic space with access to an $\nlogspace$ oracle can be solved directly by a non-deterministic Turing machine that uses logarithmic space.

We outline a variant of the procedure outlined in Algorithm~\ref{algorithm:ule:naive} to check if $\measure$ and $\measureB$ are UL.
For the sake of our analysis, we assume that the UL check in the inner loop is performed by an $\nlogspace$ oracle for state-UL.
Algorithm~\ref{algorithm:ule:naive} checks, for each state, if $\measure$ and $\measureB$ assign the same probability to its UL equivalence class.
Its correctness follows from~\cref{theorem:ule-distributions}.

\begin{algorithm}[t]
  \caption{Algorithm to decide if two distributions are UL.} \label{algorithm:ule:naive}
  \KwData{A labelled Markov chain $\mchain = \lmc$ and $\measure, \measureB\in\dist{\states}$ rational-valued distributions.}
  \For{$\state\in\states$}{
    $x_\measure\leftarrow 0$\;
    $x_\measureB\leftarrow 0$\;
    \For{$\stateB\in\states$}{
      \If{$\state$ and $\stateB$ are UL}{
        $x_\measure\leftarrow x_\measure + \measure(\stateB)$\;
        $x_\measureB\leftarrow x_\measureB + \measureB(\stateB)$\;
      }
    }
    \If{$x_\measure\neq x_\measureB$}{
      \Return reject\;
    }
  }
  \Return accept.
\end{algorithm}

Algorithm~\ref{algorithm:ule:naive} may use more than logarithmic space to store the variables $x_\measure$ and $x_\measureB$.
Therefore, the main hurdle to end our proof of $\nlogspace$-membership of the UL problem is to show that we can check the equality of sums of several rational numbers in logarithmic space.
This is known to hold: determining the equality of two sums over finite sets of rational numbers (i.e., so-called iterated sums of rationals) can be done in the circuit class uniform-$\tczero$, which is contained in $\logspace$~\cite{DBLP:journals/tcs/Jerabek12} (this follows from the fact that iterated sums and products of integers can be computed in uniform-$\tczero$; see~\cite{DBLP:journals/siamcomp/ChandraSV84} for sums and~\cite{DBLP:journals/jcss/HesseAB02} for products).

This shows $\nlogspace$-membership of the UL problem for distributions.

\section{Details of Section~\ref{section:ee} -- \nameref{section:ee}}
\label{appendix:ee}
We fix a labelled Markov chain $\mchain = \mchainTuple$ for the whole section.

\subsection{Proof of \cref{thm:ele-det}}

The goal of this section is to show that deterministic transition functions suffice to witness that states are EL.
We first establish a variation of \cref{thm:ele-det} with respect to words of a fixed length.
The following proof follows the same argument as the proof sketch of \cref{thm:ele-det} in the main text.

\begin{restatable}{proposition}{propositionELbounded}\label{prop:ele-bounded-det}
  Let $n\in\INpos$ and $\state, \stateB \in\states$.
  If there exists $\transB\in\cons{\trans}$ such that the same words of length $n$ occur from both $\state$ and $\stateB$ under $\transB$, then there exists a deterministic such $\transB$.
\end{restatable}

\begin{proof}
  Assume that there exists $\transB\in\cons{\trans}$ such that for all $\word\in\lblSet^{n}$, $\probPV{\state}{\transB}(\word)>0$ if and only if $\probPV{\stateB}{\transB}(\word)>0$.
  Let $\leq$ be a total order on $\lblSet$ and write $\leLex$ for the associated lexicographic ordering of $\lblSet^n$.
  For each $\stateC\in\states$, let $\word_\stateC$ be the minimum for $\leLex$ of the set $\{\word\in\lblSet^n\mid\probPV{\stateC}{\transB}(\word)>0\}$.
  We define the deterministic transition function $\transB'\colon\states\to\states$ by letting, for all $\stateC\in\states$, $\transB'(\stateC)\in\argmin_{\stateC'\in\supp{\trans(\stateC)}}\word_{\stateC'}$, i.e., we fix a successor $\stateC'$ of $\stateC$ with a minimum $\word_{\stateC'}$ with respect to the lexicographic ordering $\leLex$.
  By construction and definition of the lexicographic order, for all $\stateC\in\states$, we have $\word_\stateC = \lbl(\stateC)\word_{\transB'(\stateC)}^{n-1}$ where $\word_{\transB'(\stateC)}^{n-1}$ is the prefix of length $n-1$ of $\word_{\transB'(\stateC)}$.

  We now show that under $\transB'$, for all $\stateC\in\states$, $\word_{\stateC}$ has probability $1$ from $\stateC$.
  The proof is by induction: we show that for all $1\leq k\leq n$ and $\stateC\in\states$, the only word of length $k$ that occurs from $\stateC$ under $\transB'$ is the prefix of $\word_{\stateC}$ of length $k$.
  For the base case, for all $\stateC\in\states$, the prefix of length one of $\word_\stateC$ is $\lbl(\stateC)$ and the claim follows directly.

  We now assume that the claim holds by induction for $1\leq k < n$ and show it for $k+1$.
  Fix $\stateC\in\states$ and let $\stateC' = \transB(\stateC)$.
  Let $\word^k_{\stateC'}$ denote the prefix of length $k$ of $\word_{\stateC'}$.
  By the induction hypothesis, it suffices to show that $\lbl(\stateC)\word^k_{\stateC'}$ is the prefix of $\word_\stateC$ of length $k+1$ to end the proof.
  This holds by construction of $\transB'$ (this property is explained after the definition).
\end{proof}

We now apply \cref{prop:ele-bounded-det} to show \cref{thm:ele-det}.

\theoremELnew*

\begin{proof}
  Fix $\state, \stateB\in\states$.
  Let $\transB\in\cons{\trans}$ such that $\state$ and $\stateB$ generate the same words under $\transB$.
  In particular, $\state$ and $\stateB$ generate the same words of length $n = |\states|$.
  Thus, by \cref{prop:ele-bounded-det}, there exists a deterministic function $\transB^{\mathsf{det}}\in\cons{\trans}$ such that the same word of length $n$ is induced from $\state$ and $\stateB$ under $\transB^{\mathsf{det}}$.
  As language equivalence is determined by words of length at most $n$ (see, e.g.,~\cite{DBLP:journals/siamcomp/Tzeng92}), it follows that $\state$ and $\stateB$ are language equivalent under $\transB^{\mathsf{det}}$.
\end{proof}

\subsection{Proof of \cref{thm:ele-epb}}

\theoremELepb*
\begin{proof}
  If $\state$ and $\stateB$ are EB, then they are EL given that probabilistic bisimilarity implies language-equivalence.
  Conversely, assume that $\state$ and $\stateB$ are EL.
  By \cref{thm:ele-det}, this is witnessed by a deterministic transition function $\transB\in\cons{\trans}$.
  In a Markov chain with a deterministic transition function, the language-equivalence relation is a bisimulation.
  This shows that $\state$ and $\stateB$ are EB.
\end{proof}

\subsection{Proof of \cref{proposition:ele:epsilon-implications}}\label{appendix:ele:epsilon}

The goal of this section is to prove \cref{proposition:ele:epsilon-implications}, which establishes the implications between the $\varepsilon$-EB, $\varepsilon$-EL, EB and EL relations.
Before providing a proof, we highlight the following relationship between the standard language equivalence and bisimilarity relations, and $\varepsilon$-EB and $\varepsilon$-EL.

\begin{restatable}{proposition}{thmELepsilon}\label{theorem:ele:epsilon}
  Two states are language equivalent (resp.~bisimilar) if and only if they are $\varepsilon$-EL (resp.~$\varepsilon$-EB) for all $\varepsilon > 0$.
\end{restatable}

We divide \cref{theorem:ele:epsilon} in two lemmas, one per equivalence relation.
\cref{lemma:ele:epsilon-l} is concerned with language equivalence and \cref{lemma:ele:epsilon-b} with probabilistic bisimilarity.

We first discuss language equivalence.
States that are language equivalent are $\varepsilon$-EL for all $\varepsilon$ by definition.
For the converse, consider states $\state$ and $\stateB$ such that $\state$ and $\stateB$ are $\varepsilon$-EL for all $\varepsilon > 0$.
This implies that there exists a sequence of transition functions $(\transB_n)_{n\in\IN}$ converging to $\trans$ such that $\transB_n\in\cons{\trans}$ and $\state\langEquiv{\transB_n}\stateB$ for all $n\in\IN$.
Fix a word $\word\in\lblSet^*$.
We can show that $(\probLV{\state}{\transB_n}(\word))_{n\in\IN}$ converges to $\probLV{\state}{\trans}(\word)$, and that $(\probLV{\stateB}{\transB_n}(\word))_{n\in\IN}$ converges to $\probLV{\stateB}{\trans}(\word)$.
Because both of these sequences of probabilities are equal, their limits also are equal.
As this holds for all words, we conclude that $\state$ and $\stateB$ are $\trans$-language equivalent.

\begin{lemma}\label{lemma:ele:epsilon-l}
  Two states are language equivalent if and only if they are $\varepsilon$-EL for all $\varepsilon > 0$.
\end{lemma}
\begin{proof}
  Let $\state, \stateB\in\states$.
  If $\state$ and $\stateB$ are $\trans$-language equivalent, then they are $\varepsilon$-EL for all $\varepsilon>0$.

  We now assume that $\state$ and $\stateB$ are $\varepsilon$-EL for all $\varepsilon> 0$.
  For all $n\in\IN$, let $\transB_n\in\cons{\trans}$ be a transition function in the $2^{-n}$-neighbourhood of $\trans$ such that $\state\langEquiv{\transB_n}\stateB$.
  By construction, $\lim_{n\to\infty}\transB_n = \trans$.

  We now show that for all $\word\in\lblSet^*$, $\probLV{\state}{\trans}(\word)=\probLV{\stateB}{\trans}(\word)$, thereby showing that $\state\langEquiv{\trans}\stateB$.
  Let $\word\in\lblSet^*$.
  The functions $\cons{\trans}\to\IR\colon\transB\mapsto\probLV{\state}{\transB}(\word)$ and $\cons{\trans}\to\IR\colon\transB\mapsto\probLV{\stateB}{\transB}(\word)$ are continuous (the probability of a word from a state is a polynomial function of the transition probabilities).
  It follows that $\lim_{n\to\infty}\probLV{\state}{\transB_n}(\word) = \probLV{\state}{\trans}(\word)$ and $\lim_{n\to\infty}\probLV{\stateB}{\transB_n}(\word) = \probLV{\stateB}{\trans}(\word)$.
  This property, together with the fact that $\probLV{\state}{\transB_n}(\word) = \probLV{\stateB}{\transB_n}(\word)$ for all $n\in\IN$, implies that $\probLV{\state}{\trans}(\word) = \probLV{\stateB}{\trans}(\word)$.
\end{proof}

We now consider bisimilarity.
The reasoning is analogous to the above.
First, we note that bisimilar states are $\varepsilon$-EB for all $\varepsilon>0$.
For the converse implication, we fix states $\state$ and $\stateB$ that are $\varepsilon$-EB for all $\varepsilon > 0$.
Once again, we consider a sequence $(\transB_n)_{n\in\IN}\subseteq\cons{\trans}$ converging to $\trans$ such that $\state$ and $\stateB$ are $\transB_n$-bisimilar for all $n\in\IN$.
We select a subsequence of this sequence such that all elements have the same probabilistic bisimilarity relation.
We then show that this relation is a bisimulation with respect to $\trans$, which yields $\state\bisim{\trans}\stateB$.

\begin{lemma}\label{lemma:ele:epsilon-b}
    Two states are $\trans$-bisimilar if and only if they are $\varepsilon$-EB for all $\varepsilon > 0$.
\end{lemma}
\begin{proof}
  Let $\state, \stateB\in\states$.
  If $\state\bisim{\trans}\stateB$, then $\state$ and $\stateB$ are $\varepsilon$-EB for all $\varepsilon > 0$ by definition.

  We now assume that $\state$ and $\stateB$ are $\varepsilon$-EB for all $\varepsilon > 0$.
  We fix a sequence $(\transB_n)\subseteq\cons{\trans}$ such that, for all $n\in\IN$, $\transB_n$ is in the $2^{-n}$-neighbourhood of $\trans$ and $\state\bisim{\transB_n}\stateB$.
  We assume that the bisimilarity relation for all $\transB_n$ is the same relation $R\subseteq\states\times\states$.
  This assumption can be enforced by selecting a subsequence if necessary: as $\states$ is finite, there are only finitely many (bisimilarity) relations on $\states$.
  By construction, $\lim_{n\to\infty}\transB_n = \trans$.

  We show that $R$ is a bisimulation with respect to $\trans$.
  Let $(\stateC, \stateC')\in R$ and $\eClass$ be an equivalence class of $R$.
  For all $n\in\IN$, $R$ is a bisimulation with respect to $\transB_n$, therefore $\transB_n(\stateC)(\eClass) =\transB_n(\stateC')(\eClass)$.
  It follows that
  \[
    \trans(\stateC)(\eClass) =
    \lim_{n\to\infty}\transB_n(\stateC)(\eClass) =
    \lim_{n\to\infty}\transB_n(\stateC')(\eClass) =
    \trans(\stateC')(\eClass).
  \]
  This shows that $R$ is a bisimulation with respect to $\trans$, which in turn implies that $\state\bisim{\trans}\stateB$ as we have $(\state, \stateB)\in R$ by construction.
\end{proof}

We now provide a proof of \cref{proposition:ele:epsilon-implications}.
First, we note that all implications stated in \cref{proposition:ele:epsilon-implications} hold due to the definitions and the fact that bisimilarity implies language equivalence.
Second, \cref{theorem:ele:epsilon} implies that for any two states that are not language equivalent (resp.~bisimilar), there exists some $\varepsilon > 0$ such that the two states are not $\varepsilon$-EL (resp.~$\varepsilon$-EB).
Thus, to show that $\varepsilon$-EL does not imply $\varepsilon$-EB in general, it suffices to construct an example with two states that are language equivalent but not bisimilar.
Similarly, to show that EL does not imply $\varepsilon$-EL in general, it suffices to construct an example with two EL states that are not language equivalent.
For both cases, such examples can be found in the Markov chain depicted in \cref{figure:intro-example}.
We provide a detailed argument with explicit witnesses below.

\propEpsilonELImplications*
\begin{proof}
  Let $\varepsilon > 0$ and fix $\state, \stateB\in\states$.
  First, assume that $\state$ and $\stateB$ are $\varepsilon$-EB.
  Let $\transB\in\cons{\trans}$ be a transition function in the $\varepsilon$-neighbourhood of $\trans$ such that $\state\bisim{\transB}\stateB$.
  Since bisimilarity implies language-equivalence, it follows that $\state\langEquiv{\transB}\stateB$, hence $\state$ and $\stateB$ are $\varepsilon$-EL.
  Now, assume that $\state$ and $\stateB$ are $\varepsilon$-EL. Then $\state$ and $\stateB$ are EL (by definition), and are thus EB by \cref{thm:ele-epb}.

  We now show that the converse implications do not hold in general.
  We consider the Markov chain depicted in \cref{figure:intro-example}.

  We first show that there exist two states and $\varepsilon>0$ such that $\varepsilon$-EL does not imply $\varepsilon$-EB.
  Consider the language equivalent (hence $\varepsilon$-EL for all $\varepsilon>0$) states $q$ and $r$.
  It can be shown that $q$ and $r$ are not $\varepsilon$-EB for any $\varepsilon<\frac{1}{3}$.
  Indeed, for any perturbation of the transitions by less than $\frac{1}{3}$, $r$ has a single successor from which two labels are reachable, whereas the two successors of $q$ can only reach one label each.
  It follows that the successors of $q$ and $r$ cannot be made bisimilar with such perturbations, and therefore $q$ and $r$ are not $\varepsilon$-EB for $\varepsilon < \frac{1}{3}$.

  We now show that there exist two states and $\varepsilon>0$ such that the two states are EL, but not $\varepsilon$-EL.
  This is the case, e.g., of the states $t$ and $u$ for $\varepsilon < \frac{1}{2}$.
  Indeed, if the outgoing probabilities of $t$ are modified by less than $\frac{1}{2}$, then no outgoing transitions from $t$ are disabled, and thus both labels can be reached from $t$ whereas only one can be reached from $u$.
  This shows that $t$ and $u$ are not $\varepsilon$-EL for $\varepsilon < \frac{1}{2}$, despite $t$ and $u$ being EL and EB.
\end{proof}

\subsection{Proof of Theorem~\ref{theorem:complexity:ele-states}}\label{appendix:ele:np-hardness}
The goal of this section is to prove that the EL and EB problems for states are $\np$-hard, to complete the proof of the following theorem.
\theoremComplexityELstates*

Recall that by \cref{thm:ele-epb}, two states are EL if and only if they are EB.
It suffices therefore to show that the EL problem for states is $\np$-hard.
We provide a polynomial reduction from 3-SAT, the restriction of the Boolean satisfiability problem to 3-CNF formulae.
The 3-SAT problem is $\np$-complete~\cite{Co71}.
Since the existence of a consistent transition function ensuring language-equivalence depends only on the graph induced by the original transition function, the transition probabilities of the labelled Markov chain constructed are omitted in the sequel; using uniform distribution over successors for each state would be sufficient.

Let $\varphi = C_1 \wedge \ldots \wedge C_m$ be a 3-CNF formula over the variables $x_1, \ldots, x_n$, where $C_j = \lit{1j} \vee \lit{2j} \vee \lit{3j}$ for all $1 \leq j \leq m$.
We assume without loss of generality that all variables appear in $\varphi$.
Let $\litset = \{\, \lit{kj} \mid 1 \leq k \leq 3,\ 1 \leq j \leq m \,\}$ be a set of (pairwise distinct) symbols for each literal occurring in $\varphi$.
Intuitively, any two occurrences of the same literal in $\varphi$ are seen as distinct elements of $\litset$.
Fix a total order $<$ on $\litset$.

For each $1 \leq i \leq n$, let $\posset_i = \{\, \literal \in \litset \mid \literal \text{ is an occurrence of } x_i \,\}$ be the set of positive occurrences of $x_i$ and $\negset_i = \{\, \literal \in \litset \mid \literal \text{ is an occurrence of } \neglit{x}_i \,\}$ be the set of negative occurrences of $x_i$.

\subparagraph*{Intuition.}
We construct an LMC containing two isolated subgraphs, with respective initial states $\state$ and $\state'$.
The key idea is that paths that share the same sequences of labels from $\state$ and $\state'$ will move synchronously through two phases consisting of variable gadgets and clause gadgets, respectively.
The subgraphs are designed such that %
the path from $\state'$ forces the path from $\state$ to visit the corresponding gadgets in the intended order.  In this sense, the $\state'$-subgraph acts as an enforcing copy: it has fewer choices inside the gadgets, and these choices determine which successors can be chosen in the $\state$-subgraph.

In the first phase, the paths pass through the variable gadgets, where a truth assignment is chosen.
As shown in Figure~\ref{figure:variable-gadget}, each variable gadget has two main branches corresponding to $x_i = \textit{false}$ and $x_i = \textit{true}$, respectively.
We encode the chosen assignment indirectly, by marking literal occurrences as false.
Intuitively, if the first branch is visited, every literal occurrence of $x_i$ is set to \emph{false} or, if the second branch is visited, every literal occurrence of $\neglit{x}_i$ is set to \emph{false}.
Thus, the first phase records all literal occurrences that are false under the chosen assignment.

\begin{figure}[ht]
  \centering
  \begin{tikzpicture}[font=\small,scale=1.3]

    \node[state] at (2,19) (x) {$x_i$};
    \node[state] at (1,18) (vx1) {$V_{\lit{1}}$};
    \node[state] at (0.5,17) (tx1) {$T_{\lit{1}}$};
    \node[state] at (1.5,17) (fx1) {$F_{\lit{1}}$};
    \node[state] at (1,16) (vx2) {$V_{\lit{2}}$};
    \node[state] at (0.5,15) (tx2) {$T_{\lit{2}}$};
    \node[state] at (1.5,15) (fx2) {$F_{\lit{2}}$};
    \node[state] at (3,18) (vnx1) {$V_{\lit{1}'}$};
    \node[state] at (2.5,17) (tnx1) {$T_{\lit{1}'}$};
    \node[state] at (3.5,17) (fnx1) {$F_{\lit{1}'}$};
    \node[state] at (3,16) (vnx2) {$V_{\lit{2}'}$};
    \node[state] at (2.5,15) (tnx2) {$T_{\lit{2}'}$};
    \node[state] at (3.5,15) (fnx2) {$F_{\lit{2}'}$};
    \node[state, draw=none] at (1,14) (d1) {$\vdots$};
    \node[state, draw=none] at (3,14) (dn1) {$\vdots$};
    \node[state,gray] at (2,13) (y) {};

    \path[-stealth] (x) edge (vx1);
    \path[-stealth] (vx1) edge (fx1);
    \path[-stealth] (vx1) edge (tx1);
    \path[-stealth] (fx1) edge (vx2);
    \path[-stealth] (vx2) edge (fx2);
    \path[-stealth] (vx2) edge (tx2);
    \path[-stealth] (fx2) edge (d1);
    \path[-stealth] (x) edge (vnx1);
    \path[-stealth] (vnx1) edge (fnx1);
    \path[-stealth] (vnx1) edge (tnx1);
    \path[-stealth] (fnx1) edge (vnx2);
    \path[-stealth] (vnx2) edge (fnx2);
    \path[-stealth] (vnx2) edge (tnx2);
    \path[-stealth] (fnx2) edge (dn1);
    \path[-stealth] (d1) edge (y);
    \path[-stealth] (dn1) edge (y);
    
    \node[state] at (7.5,19) (xt) {$x_i'$};
    \node[state] at (6.5,18) (vx1f) {$V_{\lit{1}}^F$};
    \node[state] at (6.5,17) (fx1t) {$F_{\lit{1}}'$};
    \node[state] at (6.5,16) (vx2f) {$V_{\lit{2}}^F$};
    \node[state] at (6.5,15) (fx2t) {$F_{\lit{2}}'$};
    \node[state] at (8.5,18) (vnx1f) {$V_{\lit{1}'}^F$};
    \node[state] at (8.5,17) (fnx1t) {$F_{\lit{1}'}'$};
    \node[state] at (8.5,16) (vnx2f) {$V_{\lit{2}'}^F$};
    \node[state] at (8.5,15) (fnx2t) {$F_{\lit{2}'}'$};
    \node[state, draw=none] at (6.5,14) (d2) {$\vdots$};
    \node[state, draw=none] at (8.5,14) (dn2) {$\vdots$};
    \node[state,gray] at (7.5,13) (yt) {};

    \path[-stealth] (xt) edge (vx1f);
    \path[-stealth] (vx1f) edge (fx1t);
    \path[-stealth] (fx1t) edge (vx2f);
    \path[-stealth] (vx2f) edge (fx2t);
    \path[-stealth] (fx2t) edge (d2);
    \path[-stealth] (xt) edge (vnx1f);
    \path[-stealth] (vnx1f) edge (fnx1t);
    \path[-stealth] (fnx1t) edge (vnx2f);
    \path[-stealth] (vnx2f) edge (fnx2t);
    \path[-stealth] (fnx2t) edge (dn2);
    \path[-stealth] (d2) edge (yt);
    \path[-stealth] (dn2) edge (yt);

  \end{tikzpicture}
  \caption{The variable gadgets for $x_i$ for the subgraphs with initial state $s$ (left) and $s'$ (right), assuming that $|\posset_i| \geq 2$ and $|\negset_i| \geq 2$.}
  \label{figure:variable-gadget}
\end{figure}
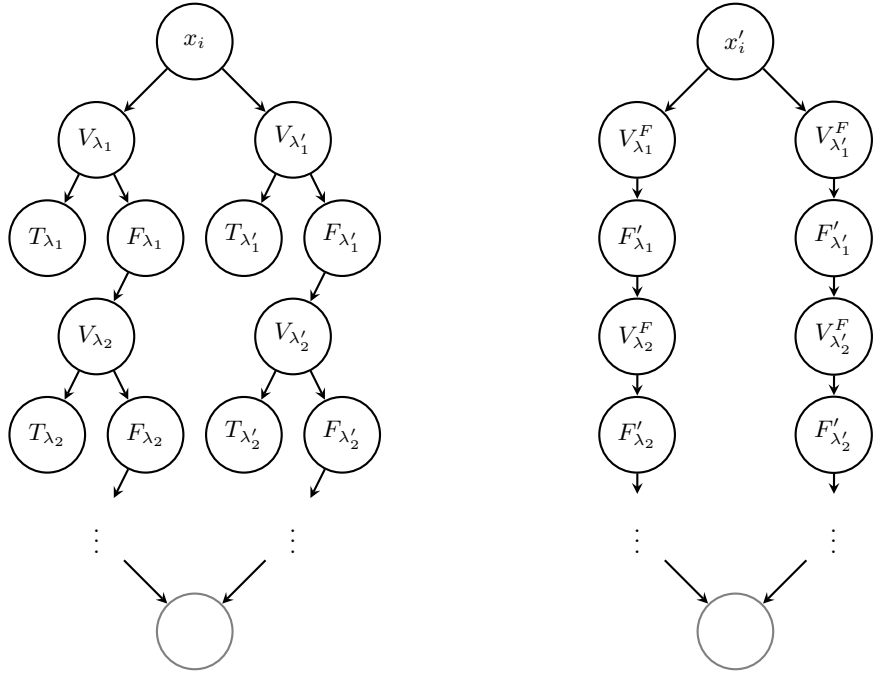

During the second phase, the paths pass through the clause gadgets, where one literal per clause is chosen as a witness that the clause is satisfied.
If the chosen assignment satisfies the formula, then in every clause $C_j$ there exists $1 \leq k \leq 3$ such that the literal $\lit{kj}$ is \emph{true}. The clause gadget attempts to choose such a literal occurrence; see Figure~\ref{figure:clause-gadget}.

\begin{figure}
  \centering
  \begin{tikzpicture}[font=\small,xscale=1.2,yscale=1.4]

    \node[state] at (2,5) (c2) {$C_j$};
    \node[state,gray] at (2,2) (c3) {};
    \node[state,draw=gray,dotted] at (2,4) (vny2s) {$V_{\lit{2j}}$};
    \node[state,draw=gray,dotted] at (1.5,3) (tny2s) {$T_{\lit{2j}}$};
    \node[state,draw=gray,dotted] at (2.5,3) (fny2s) {$F_{\lit{2j}}$};
    \node[state,draw=gray,dotted] at (0,4) (vx2s) {$V_{\lit{1j}}$};
    \node[state,draw=gray,dotted] at (-0.5,3) (tx2s) {$T_{\lit{1j}}$};
    \node[state,draw=gray,dotted] at (0.5,3) (fx2s) {$F_{\lit{1j}}$};
    \node[state,draw=gray,dotted] at (4,4) (vnz2s) {$V_{\lit{3j}}$};
    \node[state,draw=gray,dotted] at (3.5,3) (tnz2s) {$T_{\lit{3j}}$};
    \node[state,draw=gray,dotted] at (4.5,3) (fnz2s) {$F_{\lit{3j}}$};

    \path[-stealth] (c2) edge (vx2s);
    \path[-stealth] (c2) edge (vny2s);
    \path[-stealth] (c2) edge (vnz2s);
    \path[-stealth] (vx2s) edge (tx2s);
    \path[-stealth] (vx2s) edge (fx2s);
    \path[-stealth] (vny2s) edge (tny2s);
    \path[-stealth] (vny2s) edge (fny2s);
    \path[-stealth] (vnz2s) edge (tnz2s);
    \path[-stealth] (vnz2s) edge (fnz2s);
    \path[-stealth] (tx2s) edge[bend right,looseness=0.7,in=197] (c3);
    \path[-stealth] (tny2s) edge (c3);
    \path[-stealth] (tnz2s) edge[bend left] (c3);

    \node[state] at (8,5) (c2t) {$C_j'$};
    \node[state] at (8,4) (vny2t) {$V_{\lit{2j}}^T$};
    \node[state] at (8,3) (tny2t) {$T_{\lit{2j}}'$};
    \node[state,gray] at (8,2) (t2) {};
    \node[state] at (6.5,4) (vx2t) {$V_{\lit{1j}}^T$};
    \node[state] at (6.5,3) (tx2t) {$T_{\lit{1j}}'$};
    \node[state] at (9.5,4) (vnz2t) {$V_{\lit{3j}}^T$};
    \node[state] at (9.5,3) (tnz2t) {$T_{\lit{3j}}'$};

    \path[-stealth] (c2t) edge (vx2t);
    \path[-stealth] (c2t) edge (vny2t);
    \path[-stealth] (c2t) edge (vnz2t);
    \path[-stealth] (vx2t) edge (tx2t);
    \path[-stealth] (vny2t) edge (tny2t);
    \path[-stealth] (vnz2t) edge (tnz2t);
    \path[-stealth] (tx2t) edge (t2);
    \path[-stealth] (tny2t) edge (t2);
    \path[-stealth] (tnz2t) edge (t2);
    
  \end{tikzpicture}
  \caption{The clause gadget for $C_j$ for the subgraphs with initial state $s$ (left) and $s'$ (right).  States with a dotted outline are reused from the variable gadgets in phase one.}
  \label{figure:clause-gadget}
\end{figure}
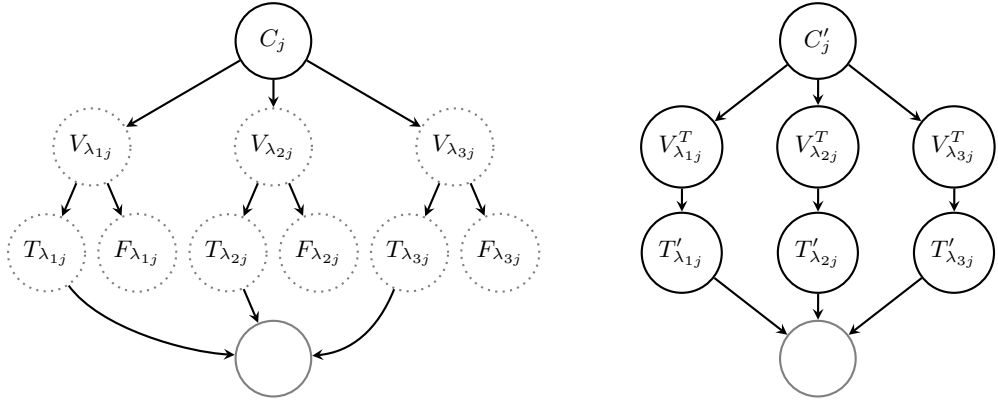

Observe that the state $V_\literal^F$ is used only in the variable gadgets, while the state $V_\literal^T$ is used only in the clause gadgets.  Moreover, these states have single successors.
Since a deterministic transition function chooses a unique successor for $V_\literal$ in the $\state$-subgraph, the choices forced by different visits to $V_\literal$ must be consistent.  Thus, if some part of the construction forces $V_\literal$ to go to $F_\literal$, no later part of the construction can force $V_\literal$ to go to $T_\literal$, as this would conflict with the earlier choice.  This is the mechanism that records whether the literal occurrence $\literal$ has been made false.
Due to this construction, a state visited in phase one cannot be revisited in phase two without falsifying the language equivalence of $\state$ and $\state'$.
Hence, for each clause, we can only choose a literal occurrence that was left unvisited in the variable gadgets, i.e., literal occurrences that are true under the chosen assignment.
In this sense, the $\state$-subgraph guarantees consistency.

The example in Figure~\ref{figure:sat-to-ele-example} illustrates the construction for the formula $(x \vee y) \wedge (x \vee \neglit{y} \vee \neglit{z})$.
The solid edges depict a deterministic transition function corresponding to the satisfying assignment $x = z = \textit{false}$ and $y = \textit{true}$.
In the first phase, all occurrences of $x$, $\neglit{y}$, and $z$ are recorded as false.  During the second phase, $y$ is chosen in the first clause and $\neglit{z}$ in the second clause.  These choices are compatible, because the chosen literals are true under the assignment.

\subparagraph*{Reduction.}
We construct a labelled Markov chain $\mchain_\varphi = \lmc$ as follows.
The state space $\states$ is partitioned into two subsets $\states_{s}$ and $\states_{s'}$ (for each isolated subgraph), where each element of $\states_{s'}$ has a counterpart in $\states_{s}$.
Each state is labelled by an element of $\states_{s}$: the label of each state in $\states_{s}$ is itself and the label of any state in $\states_{s'}$ is its counterpart in $\states_{s}$.
Formally, we define
\[\states_{s} = \lblSet = \{\, \state \,\} \cup \bigcup_{i=1}^n \{\, x_i \,\} \cup \bigcup_{j=1}^m \{\, C_j \,\} \cup \bigcup_{\literal \in \litset} \{\, V_\literal, T_\literal, F_\literal \,\},\]
and define $\states_{s'}$ as the set of so-called \textit{primed counterparts} of elements of $\states_{s}$, given by
\[
  \states_{s'} = 
\{\, \state' \,\} \cup \bigcup_{i=1}^n \{\,x_i' \,\} \cup \bigcup_{j=1}^m \{\, C_j' \,\} \cup \bigcup_{\literal \in \litset} \{\,V_\literal^T, V_\literal^F, T_\literal', F_\literal' \,\}.
\]
We note that states of the form $V_\literal\in\states_{s}$ have two primed counterparts $V_\literal^T$ and $ V_\literal^F$ in $\states_{s'}$, whereas all other elements of $\states_{s}$ have a single primed counterpart (indicated by the notation).
The labelling function $\lbl$ is defined by letting, for all $\stateC\in\states_{s}$, the label of $\stateC$ and its primed counterparts be $\stateC$.

We now describe the transition function $\trans$.
In the following, we abuse notation and write $\supp{\stateC}$ for $\supp{\trans(\stateC)}$ for any $\stateC\in\states$.
The initial states enter their respective gadget for the first variable, i.e., we let $\supp{\state} = \{ x_1 \}$ and $\supp{\state'} = \{ x_1' \}$.
For each literal occurrence $\literal \in \litset$, the transitions from the $V_\literal$-labelled states are to $T_\literal$ and $F_\literal$-labelled states: we let $\supp{V_\literal} = \{\, T_\literal, F_\literal \,\}$, $\supp{V_\literal^T} = \{\, T_\literal' \,\}$ and $\supp{V_\literal^F} = \{\, F_\literal' \,\}$.

We now define the outgoing transitions in variable gadgets.
We fix $1\leq i\leq n$.
Let $\lit{1}, \ldots, \lit{|\posset_i|}$ be the ordered list of literals in $\posset_i$ and $\lit{1}', \ldots, \lit{|\negset_i|}'$ be the ordered list of literals in $\negset_i$.
We formalise the transition structure illustrated in Figure~\ref{figure:variable-gadget}.
Let $1\leq l\leq|P_i|$.
Define the states $y_i$ and $y_i'$ as follows, which represent the initial states of the next gadgets.  If $i < n$, then $y_i = x_{i+1}$ and $y_i' = x_{i+1}'$, i.e., we move on to the next variable.  Otherwise, we have $i = n$, and then $y_i = C_1$ and $y_i' = C_1'$, i.e., we enter the first clause gadget next.
The states $F_{\lit{l}}$ and $F_{\lit{l}}'$ have a single successor, defined by one of two cases.
If $l < |P_i|$, we let $\supp{F_{\lit{l}}} = \{V_{\lit{l+1}}\}$ and $\supp{F_{\lit{l}}'} = \{V_{\lit{l+1}}^F\}$, i.e., if $\lit{l}$ is not the last positive occurrence of $x_i$, we move on to the next literal in the list.
Otherwise, $l = |P_i|$ and, thus, all literals in $P_i$ have been fixed.
Hence, we let $\supp{F_{\lit{l}}} = \{y_i\}$ and $\supp{F_{\lit{l}}'} = \{y_i'\}$.
We define analogous transitions for each $1 \leq l \leq |\negset_i|$.

We now define the outgoing transitions of $x_i$ and $x_i'$.
We distinguish three cases (recall that we assume that $x_i$ appears in $\varphi$, i.e., $P_i$ or $N_i$ is nonempty).
If $P_i$ and $N_i$ are both nonempty, we let $\supp{x_i} = \{V_{\lit{1}}, V_{\lit{1}'}\}$ and $\supp{x_i'} = \{V_{\lit{1}}^F, V_{\lit{1}'}^F\}$, i.e., we visit the first literal of each list.
If only $P_i$ is nonempty, then we let $\supp{x_i} = \{V_{\lit{1}}, y_i\}$ and $\supp{x_i'} = \{V_{\lit{1}}^F, y_i'\}$.
Similarly, if only $N_i$ is nonempty, we let $\supp{x_i} = \{y_i, V_{\lit{1}'}\}$ and $\supp{x_i'} = \{y_i', V_{\lit{1}}^F\}$.

It remains to define outgoing transitions of clause gadgets, generalising the illustration in Figure~\ref{figure:clause-gadget}.
Fix $1\leq j \leq m$.
We define $\supp{C_j} = \{\, V_{\lit{1j}}, V_{\lit{2j}}, V_{\lit{3j}} \,\}$ and $\supp{C_j'} = \{\, V_{\lit{1j}}^T, V_{\lit{2j}}^T, V_{\lit{3j}}^T \,\}$.
Finally, for all $1 \leq k \leq 3$, we let $\supp{T_{\lit{kj}}} = \{\, C_{j+1} \,\}$ if $j < m$ and $\supp{T_{\lit{kj}}} = \{\, s \,\}$ otherwise, and define outgoing transitions of $T_{\lit{kj}}'$ in the same way (with primed states).

This construction is polynomial in the size of $\varphi$.

\begin{figure}
  \centering
  \begin{tikzpicture}[scale=1.07, every node/.style={scale=0.8}]

    \node[state] at (2,20) (s1) {$\state$};
    \node[state] at (2,19) (x) {$x$};
    \node[state] at (1,18) (vx1) {$V_{x_1}$};
    \node[state] at (0.5,17) (tx1) {$T_{x_1}$};
    \node[state] at (1.5,17) (fx1) {$F_{x_1}$};
    \node[state] at (1,16) (vx2) {$V_{x_2}$};
    \node[state] at (0.5,15) (tx2) {$T_{x_2}$};
    \node[state] at (1.5,15) (fx2) {$F_{x_2}$};
    \node[state] at (2,14) (y) {$y$};
    \node[state] at (1,13) (vy1) {$V_{y_1}$};
    \node[state] at (1.5,12) (fy1) {$F_{y_1}$};
    \node[state] at (0.5,12) (ty1) {$T_{y_1}$};
    \node[state] at (3,13) (vny2) {$V_{\neglit{y}_2}$};
    \node[state] at (3.5,12) (fny2) {$F_{\neglit{y}_2}$};
    \node[state] at (2.5,12) (tny2) {$T_{\neglit{y}_2}$};
    \node[state] at (2,11) (z) {$z$};
    \node[state] at (3,10) (vnz2) {$V_{\neglit{z}_2}$};
    \node[state] at (3.5,9) (fnz2) {$F_{\neglit{z}_2}$};
    \node[state] at (2.5,9) (tnz2) {$T_{\neglit{z}_2}$};
    \node[state] at (2,8) (c1) {$C_1$};
    \node[state] at (2,5) (c2) {$C_2$};
    \node[state,draw=gray,dotted] at (2,2) (s2) {$\state$};
    \node[state,draw=gray,dotted] at (1,7) (vx1s) {$V_{x_1}$};
    \node[state,draw=gray,dotted] at (0.5,6) (tx1s) {$T_{x_1}$};
    \node[state,draw=gray,dotted] at (1.5,6) (fx1s) {$F_{x_1}$};
    \node[state,draw=gray,dotted] at (3,7) (vy1s) {$V_{y_1}$};
    \node[state,draw=gray,dotted] at (2.5,6) (ty1s) {$T_{y_1}$};
    \node[state,draw=gray,dotted] at (3.5,6) (fy1s) {$F_{y_1}$};
    \node[state,draw=gray,dotted] at (2,4) (vny2s) {$V_{\neglit{y}_2}$};
    \node[state,draw=gray,dotted] at (1.5,3) (tny2s) {$T_{\neglit{y}_2}$};
    \node[state,draw=gray,dotted] at (2.5,3) (fny2s) {$F_{\neglit{y}_2}$};
    \node[state,draw=gray,dotted] at (0,4) (vx2s) {$V_{x_2}$};
    \node[state,draw=gray,dotted] at (-0.5,3) (tx2s) {$T_{x_2}$};
    \node[state,draw=gray,dotted] at (0.5,3) (fx2s) {$F_{x_2}$};
    \node[state,draw=gray,dotted] at (4,4) (vnz2s) {$V_{\neglit{z}_2}$};
    \node[state,draw=gray,dotted] at (3.5,3) (tnz2s) {$T_{\neglit{z}_2}$};
    \node[state,draw=gray,dotted] at (4.5,3) (fnz2s) {$F_{\neglit{z}_2}$};

    \path[-stealth] (s1) edge (x);
    \path[-stealth] (x) edge (vx1);
    \path[-stealth] (vx1) edge (fx1);
    \path[-stealth,dashed] (vx1) edge (tx1);
    \path[-stealth] (fx1) edge (vx2);
    \path[-stealth] (vx2) edge (fx2);
    \path[-stealth,dashed] (vx2) edge (tx2);
    \path[-stealth] (fx2) edge (y);
    \path[-stealth,dashed] (x) edge[bend left] (y);
    \path[-stealth,dashed] (y) edge (vy1);
    \path[-stealth,dashed] (vy1) edge (fy1);
    \path[-stealth] (vy1) edge (ty1);
    \path[-stealth,dashed] (fy1) edge (z);
    \path[-stealth] (y) edge (vny2);
    \path[-stealth] (vny2) edge (fny2);
    \path[-stealth,dashed] (vny2) edge (tny2);
    \path[-stealth] (fny2) edge (z);
    \path[-stealth,dashed] (z) edge (vnz2);
    \path[-stealth,dashed] (vnz2) edge (fnz2);
    \path[-stealth] (vnz2) edge (tnz2);
    \path[-stealth,dashed] (fnz2) edge (c1);
    \path[-stealth] (z) edge[bend right] (c1);
    \path[-stealth,dashed] (c1) edge (vx1s);
    \path[-stealth] (c1) edge (vy1s);
    \path[-stealth,dashed] (vx1s) edge (tx1s);
    \path[-stealth] (vx1s) edge (fx1s);
    \path[-stealth] (vy1s) edge (ty1s);
    \path[-stealth,dashed] (vy1s) edge (fy1s);
    \path[-stealth,dashed] (tx1s) edge (c2);
    \path[-stealth] (ty1s) edge (c2);
    \path[-stealth,dashed] (c2) edge (vx2s);
    \path[-stealth,dashed] (c2) edge (vny2s);
    \path[-stealth] (c2) edge (vnz2s);
    \path[-stealth,dashed] (vx2s) edge (tx2s);
    \path[-stealth] (vx2s) edge (fx2s);
    \path[-stealth,dashed] (vny2s) edge (tny2s);
    \path[-stealth] (vny2s) edge (fny2s);
    \path[-stealth] (vnz2s) edge (tnz2s);
    \path[-stealth,dashed] (vnz2s) edge (fnz2s);
    \path[-stealth,dashed] (tx2s) edge[bend right=18] (s2);
    \path[-stealth,dashed] (tny2s) edge (s2);
    \path[-stealth] (tnz2s) edge[bend left] (s2);

    \node[state] at (9,20) (t1) {$\state'$};
    \node[state] at (9,19) (xt) {$x'$};
    \node[state] at (8,18) (vx1f) {$V_{x_1}^F$};
    \node[state] at (8,17) (fx1t) {$F_{x_1}'$};
    \node[state] at (8,16) (vx2f) {$V_{x_2}^F$};
    \node[state] at (8,15) (fx2t) {$F_{x_2}'$};
    \node[state] at (9,14) (yt) {$y'$};
    \node[state] at (8,13) (vy1f) {$V_{y_1}^F$};
    \node[state] at (8,12) (fy1t) {$F_{y_1}'$};
    \node[state] at (10,13) (vny2f) {$V_{\neglit{y}_2}^F$};
    \node[state] at (10,12) (fny2t) {$F_{\neglit{y}_2}'$};
    \node[state] at (9,11) (zt) {$z'$};
    \node[state] at (10,10) (vnz2f) {$V_{\neglit{z}_2}^F$};
    \node[state] at (10,9) (fnz2t) {$F_{\neglit{z}_2}'$};
    \node[state] at (9,8) (c1t) {$C_1'$};
    \node[state] at (8,7) (vx1t) {$V_{x_1}^T$};
    \node[state] at (8,6) (tx1t) {$T_{x_1}'$};
    \node[state] at (10,7) (vy1t) {$V_{y_1}^T$};
    \node[state] at (10,6) (ty1t) {$T_{y_1}'$};
    \node[state] at (9,5) (c2t) {$C_2'$};
    \node[state] at (9,4) (vny2t) {$V_{\neglit{y}_2}^T$};
    \node[state] at (9,3) (tny2t) {$T_{\neglit{y}_2}'$};
    \node[state,draw=gray,dotted] at (9,2) (t2) {$\state'$};
    \node[state] at (7.5,4) (vx2t) {$V_{x_2}^T$};
    \node[state] at (7.5,3) (tx2t) {$T_{x_2}'$};
    \node[state] at (10.5,4) (vnz2t) {$V_{\neglit{z}_2}^T$};
    \node[state] at (10.5,3) (tnz2t) {$T_{\neglit{z}_2}'$};

    \path[-stealth] (t1) edge (xt);
    \path[-stealth] (xt) edge (vx1f);
    \path[-stealth] (vx1f) edge (fx1t);
    \path[-stealth] (fx1t) edge (vx2f);
    \path[-stealth] (vx2f) edge (fx2t);
    \path[-stealth] (fx2t) edge (yt);
    \path[-stealth,dashed] (xt) edge[bend left] (yt);
    \path[-stealth,dashed] (yt) edge (vy1f);
    \path[-stealth,dashed] (vy1f) edge (fy1t);
    \path[-stealth,dashed] (fy1t) edge (zt);
    \path[-stealth] (yt) edge (vny2f);
    \path[-stealth] (vny2f) edge (fny2t);
    \path[-stealth] (fny2t) edge (zt);
    \path[-stealth,dashed] (zt) edge (vnz2f);
    \path[-stealth,dashed] (vnz2f) edge (fnz2t);
    \path[-stealth,dashed] (fnz2t) edge (c1t);
    \path[-stealth] (zt) edge[bend right] (c1t);
    \path[-stealth,dashed] (c1t) edge (vx1t);
    \path[-stealth] (c1t) edge (vy1t);
    \path[-stealth,dashed] (vx1t) edge (tx1t);
    \path[-stealth] (vy1t) edge (ty1t);
    \path[-stealth,dashed] (tx1t) edge (c2t);
    \path[-stealth] (ty1t) edge (c2t);
    \path[-stealth,dashed] (c2t) edge (vx2t);
    \path[-stealth,dashed] (c2t) edge (vny2t);
    \path[-stealth] (c2t) edge (vnz2t);
    \path[-stealth,dashed] (vx2t) edge (tx2t);
    \path[-stealth,dashed] (vny2t) edge (tny2t);
    \path[-stealth] (vnz2t) edge (tnz2t);
    \path[-stealth,dashed] (tx2t) edge (t2);
    \path[-stealth,dashed] (tny2t) edge (t2);
    \path[-stealth] (tnz2t) edge (t2);
    
  \end{tikzpicture}
  \caption{The labelled Markov chain constructed for the formula $(x \vee y) \wedge (x \vee \neglit{y} \vee \neglit{z})$.  States with a dotted outline are duplicated for the ease of the reader. A deterministic $\transB \in \cons{\trans}$ such that $\state\langEquiv{\transB}\state'$ is shown with solid transitions and corresponds to the satisfying assignment $x = z = \textit{false}$ and $y = \textit{true}$. }
  \label{figure:sat-to-ele-example}
\end{figure}

\subparagraph*{Correctness.}

Before proving correctness, we spell out the role of the shared states $V_\literal$ under a deterministic transition function (it is sufficient to only consider deterministic transition functions by \cref{thm:ele-det}).  The paths in each isolated subgraph are forced to remain synchronised by label equality.  Thus, whenever the path from $\state'$ visits $V_\literal^F$, the path from $\state$ must visit the state with the same label, namely $V_\literal$.  Since $V_\literal^F$ has the unique successor $F_\literal'$, the transition function must choose $F_\literal$ as the successor of $V_\literal$.

Similarly, in the clause gadgets, picking $V_\literal^T$ as the successor in the path from $\state'$ (on the right-hand side of illustrations) leads us to visit $V_\literal$ in the run from $\state$ (on the left-hand side).  Since $V_\literal^T$ has the unique successor $T_\literal'$, equality of the generated words forces the transition function to choose $T_\literal$ as the successor of $V_\literal$.

Hence a literal occurrence that was marked false in phase one cannot be chosen in phase two to satisfy a clause without violating label-equality of the paths in each subgraph.  Therefore, there exists a deterministic transition function such that $\state$ and $\state'$ generate the same infinite word if and only if every clause contains a literal occurrence that was not marked false in the first phase, which is precisely the condition that the formula is satisfiable.

\begin{proposition}
The formula $\varphi$ is satisfiable if and only if there exists a deterministic $\transB\in\cons{\trans}$ for $\mchain_\varphi$ such that $\state\langEquiv{\transB}\state'$.
\end{proposition}
\begin{proof}
We first assume that $\varphi$ is satisfiable.
Let $f$ be a satisfying assignment of $\varphi$.
We define a deterministic transition function $\transB \in \cons{\trans}$.
To lighten notation, we view $\transB$ as a function $\states\to\states$ below.

For each variable $x_i$ with $1 \leq i \leq n$, and for all $\literal \in \posset_i \cup \negset_i$, if $f(x_i) = \textit{false}$, we let
\begin{equation*}
\sigma(x_i) =
\begin{cases}
 V_{\lit{p}}  \qquad & \text{if } \posset_i \neq \varnothing \text{ and } \lit{p} = \min \posset_i \\
 x_{i+1}  \qquad & \text{if } \posset_i = \varnothing \text{ and } i < n\\
 C_1  \qquad & \text{if } \posset_i = \varnothing \text{ and } i = n
\end{cases}
\end{equation*}
and let $\sigma(V_\literal)=T_\literal$ if $\literal \in \negset_i$ and $\sigma(V_\literal)=F_\literal$ otherwise, and, if $f(x_i) = \textit{true}$, we let
\begin{align*}
\sigma(x_i) &=
\begin{cases}
 V_{\lit{n}}  \qquad & \text{if } \negset_i \neq \varnothing \text{ and } \lit{n} = \min \negset_i \\
 x_{i+1}  \qquad & \text{if } \negset_i = \varnothing \text{ and } i < n\\
\ C_1  \qquad & \text{if } \negset_i = \varnothing \text{ and } i = n
\end{cases}
\end{align*}
and let $\sigma(V_\literal)=T_\literal$ if $\literal \in \posset_i$ and $\sigma(V_\literal)=F_\literal$ otherwise.
We define $\transB(x_i')$ analogously to $\transB(x_i)$.

For each clause $C_j$ with $1 \leq j \leq m$, at least one literal occurrence in $C_k$ is \emph{true} under $f$.  Let $1 \leq k \leq 3$ such that $\lit{kj}$ is \emph{true}.
We define $\transB(C_j) = V_{\lit{kj}}$ and $\transB(C_j') = V_{\lit{kj}}^T$.
Note that the earlier choice for the successor of $V_{\lit{kj}}$ must have been $T_{\lit{kj}}$, since no true literal occurrence was marked false in the first variable phase above.

The remaining states have unique successors.  By construction, the paths given by $\transB$ from $\state$ and $\state'$ visit states with equal labels at every step.  After the final clause gadget, both paths return to $\state$ and $\state'$, respectively, so the same finite pattern repeats forever.  Hence $\state$ and $\state'$ generate the same infinite word.
This shows that $\state$ and $\state'$ are EL.

  We now suppose that there exists a transition function $\transB \in \cons{\trans}$ such that $\state\langEquiv{\transB}\state'$, i.e., the same infinite word is generated from $\state$ and $\state'$ under $\transB$.
  Without loss of generality, we can assume that $\transB$ is deterministic, by \cref{thm:ele-det}.
  We view $\transB$ as a function $\states\to\states$.
  We derive from $\transB$ an assignment $f$ of $x_1, \ldots, x_n$ that satisfies $\varphi$.

For all $1 \leq i \leq n$, the state $x_i'$ has two branches with distinct labels (cf.~Figure~\ref{figure:variable-gadget}). 
If $\transB(x_i') = V_{\lit{n}}^F$ for some $\lit{n} \in \negset_i$, or if $\negset_i = \varnothing$ and $\sigma(x_i')\in\{x_{i+1}', C_1'\}$, then we set $f(x_i) = \textit{true}$.  Otherwise, we set $f(x_i) = \textit{false}$.

For the words generated from $\state$ and $\state'$ to be equal, the choices of $\transB$ from $x_i$ must match those from $x_i'$ until the subsequent gadget is reached.
If a literal $\literal \in \litset$ is \emph{false} under $f$, then by construction of the variable gadgets, we must have $\sigma(V_\literal) = F_\literal$ given that the state $V_\literal^F$ has the unique successor $F_\literal'$.

Now consider a clause $C_j$ with $1 \leq j \leq m$.  In its clause gadget, one successor of state $C_j'$ must be chosen, say $V_{\lit{kj}}^T$ with $1 \leq k \leq m$.  Since $V_{\lit{kj}}^T$ has the unique successor $T_{\lit{kj}}'$, equality of the generated words forces $\sigma(V_{\lit{kj}}) = T_{\lit{kj}}$ (seeing that $C'_j$ is visited in all paths from $\state'$).
This is impossible if $f(\lit{kj}) = \textit{false}$, due to the argument above.  Therefore, $f(\lit{kj}) = \textit{true}$.  It follows that every clause contains a true literal under $f$, thus, $f$ satisfies $\varphi$.
\end{proof}

We have shown that the satisfiability of $\varphi$ is equivalent to states $\state$ and $\state'$ of $\mchain_\varphi$ being EL.
Given that $\mchain_\varphi$ can be constructed from $\varphi$ in time polynomial in the size of $\varphi$, we have established the $\np$-hardness of the EL problem for states.
This completes the proof of \cref{theorem:complexity:ele-states}.

\subsection{Proof of Theorem~\ref{theorem:nmf-ele}}\label{appendix:ele:nmf}
We prove that the EL problem (for distributions) is $\exists\IR$-hard in this section.
We describe a reduction from the non-negative matrix factorisation problem (NMF) to the EL problem.
The NMF problem asks, given a non-negative matrix $M\in\IR^{n\times m}$ and $k\geq 1$, whether there exist non-negative matrices $H\in\IR^{n\times k}$ and $G\in\IR^{k\times m}$ such that $M = H\cdot G$.
The NMF problem is $\exists\IR$-complete~\cite{shitov18NMF,shitov21NMF}.
We adapt a reduction presented in~\cite{DBLP:conf/fsttcs/Kiefer020}. %

\theoremNMFele*
\begin{proof}
  We show that NMF problem reduces to the EL problem for distributions in polynomial time.
  Let $M\in\IR^{n\times m}$ be a non-negative matrix and $k\geq 1$.
  We assume without loss of generality that $M$ is a stochastic matrix~\cite[Section 3]{CR93} (similarly to~\cite{DBLP:journals/lmcs/FijalkowKS20,DBLP:conf/fsttcs/Kiefer020}).
  We consider the set of $n + m + 1$ labels $\lblSet = \{a_{i}\mid 1\leq i\leq n\}\cup \{b_{j}\mid 1\leq j\leq m\}\cup \{c\}$.
  Intuitively, we construct an $\lblSet$-labelled Markov chain consisting of two components, one for the matrix $M$ and another to witness the existence of the matrices $H$ and $G$.
  Each component can be seen as a directed acyclic graph of depth $3$ (with the last layer being absorbing), where states have, in order, an $a_i$ label, the $c$ label and a $b_j$ label.
  For the $M$-component, we design the transition function such that there is only one consistent choice and define one of the initial distributions such that the word $a_{i}cb_{j}$ has probability $\frac{1}{n}M_{ij}$.
  We construct the $GH$-component such that the probability of a word $a_{i}cb_{j}$ from the sole $a_{i}$ labelled state is determined by the product of two matrices, given by the transition probabilities from layer $1$ to $2$ and from layer $2$ to $3$.
  See \cref{figure:nmf-reduction} for an illustration.
  
  For the $M$-component of the Markov chain, define
  \[S_M = \{s^d_{ij} \mid 1\leq i\leq n, 1\leq j\leq m, d\in\{a, b, c\}\},\]
  and for the $GH$-component, define
  \[S_{GH} = \{t^a_{i}\mid 1\leq i\leq n\}\cup\{t^c_{r}\mid 1\leq r\leq k\}\cup\{t^b_{j}\mid 1\leq j\leq m\}.\]
  We let $S = S_M\cup S_{GH}$.
  We define the labelling function $\lbl$ as follows.
  For all $1\leq i\leq n$, $1\leq j\leq m$ and $1\leq r\leq k$, we let $\lbl(s^a_{ij}) = \lbl(t^a_i) = a_i$, $\lbl(s^c_{ij}) = \lbl(t^c_r) = c$ and $\lbl(s^b_{ij}) = \lbl(t^b_j) = b_j$.

  It remains to describe the transition structure.
  We define a relation $\rightarrow$ that indicates enabled (i.e., consistent) transitions in the Markov chain.
  For $S_M$, we let, for all $1\leq i\leq n$ and $1\leq j\leq m$, $s^a_{ij}\rightarrow s^c_{ij}$, $s^c_{ij}\rightarrow s^b_{ij}$ and $s^b_{ij}\rightarrow s^b_{ij}$ be the only enabled transitions.
  For $S_{GH}$, we let, for all $1\leq i\leq n$, $1\leq r\leq k$ and $1\leq j\leq m$, we let $t^a_i\rightarrow t^c_r$, $t^c_r\rightarrow t^b_j$ and $t^b_j\rightarrow t^b_j$ be the only enabled transitions (in particular, all transitions between two layers are enabled).

  Let $\measure\in\dist{S}$ and $\measureB\in\dist{S}$ be such that $\measure(s^a_{ij}) = \frac{1}{n}\cdot M_{ij}$ for all $1\leq i\leq n$ and $1\leq j\leq m$ and $\measure(t^a_i) = \frac{1}{n}$ for all $1\leq i\leq n$.
  We claim that $\measure$ and $\measureB$ are EL if and only if the NMF instance considered above is positive.
  This follows from the fact that for all transition functions $\transB\colon\states\to\dist{\states}$ that are consistent with $\rightarrow$ and all $1\leq i\leq n$, $1\leq j\leq m$, we have $\probLV{\measure}{\transB}(a_icb_j^\omega) = \frac{1}{n} M_{ij}$ and $\probLV{\measureB}{\transB}(a_icb_j^\omega) = \frac{1}{n}\cdot\sum_{r=1}^k\transB(t^a_i, t^c_r)\cdot\transB(t^c_r, t^b_j)$.
  On the one hand, if the EL of $\measure$ and $\measureB$ is witnessed by the transition function $\transB$, a non-negative factorisation of $M$ is given by $G = (\transB(t^a_i, t^c_r))_{1\leq i\leq n, 1 \leq k\leq r}$ and $H = (\transB(t^c_r, t^b_j))_{1\leq r\leq k, 1 \leq j\leq m}$.
  Conversely, assume that there exists a factorisation of $M$ given by non-negative matrices $G\in\IR^{n\times k}$ and $H\in\IR^{k\times m}$.
  We may assume that $G$ and $H$ are (row-)stochastic (see~\cite[Thm.~3.2]{CR93}).
  The existential language equivalence of $\measure$ and $\measureB$ is witnessed by $\transB\in\cons{\trans}$ such that, for all $1\leq i\leq n$, $1\leq r\leq k$ and $1\leq j\leq m$, we have $\transB(t^a_i, t^c_r) = G_{ir}$ and $\transB(t^c_r, t^b_j) = H_{rj}$.
\end{proof}

\begin{figure}[ht]
  \centering
  \begin{tikzpicture}[font=\small]
    \node at (1.75,5.75) () {$\mu$};
    \node at (8.75,5.75) () {$\nu$};
    \node[state,fill=black,scale=0.4,minimum size=1pt] at (1.75,5.5) (mu) {};
    \node[state,fill=black,scale=0.4,minimum size=1pt] at (8.75,5.5) (nu) {};
    \node[state] at (0,4) (s1) {$s_{11}^a$};
    \node[state] at (0,2) (s1p) {$s_{11}^c$};
    \node[state] at (0,0) (p1) {$s_{11}^b$};
    \node[state] at (3.5,4) (sn) {$s_{nm}^a$};
    \node[state] at (3.5,2) (snp) {$s_{nm}^c$};
    \node[state] at (3.5,0) (pm) {$s_{nm}^b$};
    \node at (1.75,4.85) () {$\cdots$};
    \node at (1.75,4) () {$\cdots$};
    \node at (1.75,2) () {$\cdots$};
    \node at (1.75,0) () {$\cdots$};
    \node[state] at (7,4) (t1) {$t_1^a$};
    \node[state] at (7,2) (t1p) {$t_1^c$};
    \node[state] at (7,0) (q1) {$t_1^b$};
    \node[state] at (10.5,4) (tn) {$t_n^a$};
    \node[state] at (10.5,2) (tnp) {$t_k^c$};
    \node[state] at (10.5,0) (qm) {$t_m^b$};
    \node at (8.75,4.85) () {$\cdots$};
    \node at (8.75,4) () {$\cdots$};
    \node at (8.75,2) () {$\cdots$};
    \node at (8.75,0) () {$\cdots$};
    \node at (7.75,3) () {$\cdots$};
    \node at (9.75,3) () {$\cdots$};
    \node at (7.75,1) () {$\cdots$};
    \node at (9.75,1) () {$\cdots$};
    \path[-stealth] (mu) edge node[above left] {$\frac{1}{n} M_{11}$} (s1);
    \path[-stealth] (mu) edge node[above right] {$\frac{1}{n} M_{nm}$} (sn);
    \path[-stealth] (s1) edge node[left] {$1$} (s1p);
    \path[-stealth] (s1p) edge node[left] {$1$} (p1);
    \path[-stealth] (sn) edge node[right] {$1$} (snp);
    \path[-stealth] (snp) edge node[right] {$1$} (pm);
    \path[-stealth] (p1) edge[loop below] node[left,xshift=-5] {$1$} (p1);
    \path[-stealth] (pm) edge[loop below] node[right,xshift=5] {$1$} (pm);
    \path[-stealth] (nu) edge node[above left] {$\frac{1}{n}$} (t1);
    \path[-stealth] (nu) edge node[above right] {$\frac{1}{n}$} (tn);
    \path[-stealth] (t1) edge node[left] {$G_{11}$} (t1p);
    \path[-stealth] (t1p) edge node[left] {$H_{11}$} (q1);
    \path[-stealth] (tn) edge node[right] {$G_{nk}$} (tnp);
    \path[-stealth] (tnp) edge node[right] {$H_{km}$} (qm);
    \path[-stealth] (t1) edge node[near start,above,xshift=5] {$G_{1k}$} (tnp);
    \path[-stealth] (tn) edge node[near start,above,xshift=-5] {$G_{n1}$} (t1p);
    \path[-stealth] (t1p) edge node[near start,above,xshift=5] {$H_{1m}$} (qm);
    \path[-stealth] (tnp) edge node[near start,above,xshift=-5] {$H_{k1}$} (q1);
    \path[-stealth] (q1) edge[loop below] node[left,xshift=-5] {$1$} (q1);
    \path[-stealth] (qm) edge[loop below] node[right,xshift=5] {$1$} (qm);
  \end{tikzpicture}
  \caption{The labelled Markov chain consisting of an $M$-component (left) and a $GH$-component (right) with distributions $\mu$ and $\nu$. %
  The labels of the states are as follows: for all $1\leq i\leq n$, $1\leq j\leq m$, we have $\lbl(s^a_{ij}) = \lbl(t^a_i) = a_i$ and $\lbl(s^b_{ij}) = \lbl(t^b_j) = b_j$, and all remaining states have label $c$.}
  \label{figure:nmf-reduction}
\end{figure}
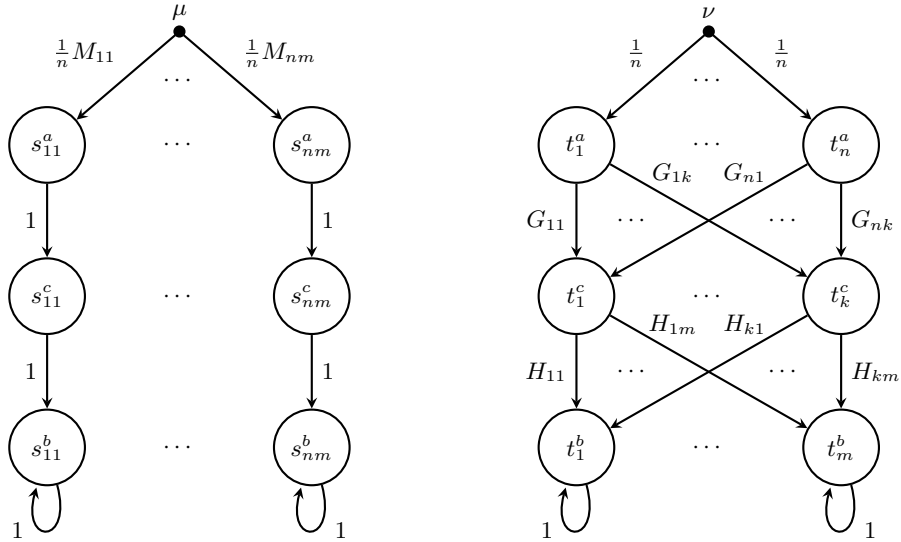

\section{Perturbation equivalence in Markov decision processes}
\label{appendix:mdps}
This section discusses the problems considered in~\cite{DBLP:conf/fsttcs/Kiefer020} that generalise the UL, UB, EL and EB problems considered in the main text to Markov decision processes (MDPs).
For consistency with the main text, we use perturbation equivalence terminology along with the UL, UB, EL and EB nomenclature throughout this section for these MDP-related problems.
\cref{table:complexity-comparison} recapitulates our complexity results for LMCs and provides a comparison with the corresponding problems for MDPs.
The complexity results for MDPs reported in the table combine results from Kiefer and Tang \cite{DBLP:conf/fsttcs/Kiefer020} with several enhancements established in this section.

\begin{table}[b]
  \caption{Summary of our complexity results and comparison with counterpart problems in MDPs.
    We use the subscripts $\mathcal{D}$ and $S$ to denote the variants of the universal and existential language equivalence problems that study the equivalence of distributions and states respectively.
  }\label{table:complexity-comparison}
  \centering
    \begin{tabular}{|c|c|c|c|c|c|c|}
      \hline
      & UL$_\mathcal{D}$ & UL$_S$ &  UB & EL$_\mathcal{D}$ & EL$_S$ & EB \\
      \hline
      Labelled Markov chains
            & \multicolumn{3}{c|}{\nlogspace-c.}
            & $\exists\IR$-c.
            & \multicolumn{2}{c|}{\np-c.}
      \\
      \hline
      Labelled MDPs & \multicolumn{2}{c|}{In $\mathsf{C_{=}L}$} & $\ptime$-c. & \multicolumn{2}{c|}{$\exists\IR$-c.} & $\np$-c. \\
      \hline
  \end{tabular}
\end{table}

This section is structured as follows.
In Section~\ref{section:mdps:background}, we formally define MDPs and other relevant notions.
Section~\ref{section:mdps:problems} defines the universal and existential equivalence problems in MDPs and explain how they generalise the perturbation equivalence problems in Markov chains defined in Section~\ref{section:problems}.
In Section~\ref{section:mdps:upb}, we provide a simple proof of $\ptime$-hardness of the UB problem in MDPs.
In Section~\ref{section:mdps:ule}, we slightly refine the complexity upper bound in~\cite{DBLP:conf/fsttcs/Kiefer020} for the UL problem in MDPs from $\ptime$ to $\mathsf{C_{=}L}$ by giving a reduction from this UL problem to the problem of language-equivalence in Markov chains.
To simplify the presentation of this reduction, we use an alternative formalism for labelled Markov chains and MDPs in Section~\ref{section:mdps:ule}, where labels are placed on transitions.
Section~\ref{section:mdps:perturbation-equivalence-interreducibility} presents an efficient transformation of MDPs from one formalism to the other that preserves perturbation equivalence.

\subsection{Labelled MDPs}\label{section:mdps:background}

\subparagraph*{Markov decision processes.} MDPs can be seen as a generalisation of Markov chains with non-determinism: at each step of the process, an action has to be selected in the current state, and the next state is determined by a distribution that depends only on the current state and chosen action.
Formally, a labelled \textit{Markov decision process} (MDP) is a tuple $\mdp = (\states, \actions, \trans, \lblSet, \lbl)$, where $\states$ is a finite set of states, $\actions$ is a finite set of actions (actions are denoted by $\action$ to distinguish them from labels), $\trans\colon\states\times\actions\to\dist{\states}$ is a (partial) transition function, $\lblSet$ is a finite set of labels and $\lbl\colon\states\to\lblSet$ is a labelling function.

An action $\action\in\actions$ is \textit{enabled} in a state $\state\in\states$ if $\trans(\state, \action)$ is defined.
We write $\actions(\state)$ for the set of actions enabled in a state $\state\in\states$.
We assume that MDPs are deadlock-free, i.e., there is at least one enabled action in all states.
For any states $\state, \state'\in\states$ and $\action\in\actions(\state)$, we write $\trans(\state, \action, \state')$ for the probability $\trans(\state, \action)(\state')$.

Paths and histories of an MDP are sequences of states and actions.
Formally, a \textit{path} of $\mdp$ is a sequence $\state_0\action_0\state_1\action_1\ldots\in (\states\actions)^\omega$ such that, for all $\iPos\in\IN$, $\trans(\state_\iPos, \action_{\iPos}, \state_{\iPos+1}) > 0$.
A \textit{history} is a finite prefix of a path ending in a state (i.e., in $\states(\actions\states)^*$).

\subparagraph*{Strategies.}
A strategy describes how one selects actions throughout a path of the MDP.
In general, strategies are defined as functions that assign to each history a distribution over actions enabled in the last state of the history.
For this section, we are only concerned with memoryless strategies, i.e., strategies whose decision depend only on the current state.
Formally, a \textit{memoryless strategy} is a function $\strat\colon\states\to\dist{\actions}$ such that for all $\state\in\states$, $\supp{\strat(\state)}\subseteq \actions(\state)$.
As for transition functions of Markov chains, we write $\strat(\state, \action)$ instead of $\strat(\state)(\action)$ for all $\state\in\states$ and $\action\in\actions$.

\subparagraph*{Induced Markov chains.}
A memoryless strategy $\strat$ induces a Markov chain $\mdp^\strat = (\states, \trans^\strat, \lblSet, \lbl)$ where the transition function $\trans^\strat$ is defined by, for all $\state, \state'\in\states$,
\[\trans^{\strat}(\state, \state') = \sum_{\action\in\actions(\state)}
  \strat(\state,\action)\cdot\trans(\state, \action, \state').\]
To lighten notation, for any $\measure\in\dist{\states}$, we write $\probPV{\measure}{\strat}$ instead of $\probPV{\measure}{\trans^\strat}$ for the distribution over paths induced by $\mdp^\strat$ from the initial distribution $\measure$ (and similarly for initial states).
We remark that actions are omitted from induced Markov chains.

\subparagraph*{Equivalence notions.}
We define language-equivalence and probabilistic bisimilarity in labelled MDPs with respect to a given memoryless strategy as in Markov chains.
Let $\strat\colon\states\to\dist{\actions}$.
Formally, two distributions $\measure$, $\measureB\in\dist{\states}$ are language-equivalent under $\strat$, denoted $\measure\langEquiv{\strat}\measureB$ if  $\measure$ and $\measureB$ are language-equivalent in the Markov chain $\mdp^\strat$.
Similarly, two states $\state$, $\stateB\in\states$ are bisimilar under $\strat$, denoted by $\state\bisim{\strat}\stateB$, if they are bisimilar in $\mdp^\strat$.

\subsection{Perturbation equivalence problems in MDPs}\label{section:mdps:problems}

Perturbation equivalence problems in MDPs ask if there exist \textit{memoryless strategies} that make two states or distributions equivalent or inequivalent.
These problems, described below, have been studied in~\cite{DBLP:conf/fsttcs/Kiefer020}, and generalise the UL, UB, EL and EB problems in labelled Markov chains.
Intuitively, choosing a memoryless strategy in a deterministic MDP (i.e., an MDP where all distributions given by the transition function are Dirac) can be seen as selecting a consistent transition function in a Markov chain.
We formalise this idea at the end of this section.

\subparagraph*{Language-equivalence.}
We describe first MDP variants of UL and EL.

On the one hand, the \textit{UL problem in labelled MDPs} asks, given a labelled MDP and two initial distributions, whether these distributions are language-equivalent under all memoryless strategies.
The UL problem is the negation of the $\mathsf{TV}>0$ problem studied in~\cite{DBLP:conf/fsttcs/Kiefer020}, which asks whether there exists a memoryless strategy under which two given distributions are not language-equivalent.
The name $\mathsf{TV}>0$ references the \textit{total-variation distance}, which is non-zero between two distributions if and only if these distributions are not language-equivalent.
This problem is known to be solvable in polynomial time~\cite{DBLP:conf/fsttcs/Kiefer020}.
We show in Section~\ref{section:mdps:ule} that this problem is logspace-reducible to the problem of deciding the language-equivalence of two initial distributions in labelled Markov chains, which slightly refines this upper bound.

On the other hand, the \textit{EL problem in labelled MDPs} asks, given two initial distributions, whether there exists a memoryless strategy under which the two distributions are equivalent.
It is referred to as the $\mathsf{TV}=0$ problem in~\cite{DBLP:conf/fsttcs/Kiefer020}, where it is shown to be $\exists\IR$-complete.

\subparagraph*{Probabilistic bisimilarity.}
We now move on to bisimilarity.
The \textit{UB problem in labelled MDPs} asks, given a labelled MDP and two initial states, whether these two states are bisimilar under all memoryless strategies.
Finally, the \textit{EB problem in labelled MDPs} asks, given a labelled MDP and two initial states, whether these two states are bisimilar under some memoryless strategy.

The negation of the UB problem is studied in~\cite{DBLP:conf/fsttcs/Kiefer020} under the name $\mathsf{PB}>0$ (can we make the probabilistic bisimilarity distance positive?), where it is shown to be in $\ptime$.
We show $\ptime$-completeness below: we outline a direct reduction from the problem of checking whether two states of a labelled Markov chain are bisimilar, which is $\ptime$-hard~\cite{DBLP:conf/fossacs/ChenBW12}, to the UB problem in MDPs.
The EB problem, known as $\mathsf{PB}=0$, is shown to be $\np$-complete~\cite{DBLP:conf/fsttcs/Kiefer020}.

\begin{remark}
  The $\np$-hardness proof for the EB problem in MDPs presented in~\cite{DBLP:conf/fsttcs/Kiefer020} does not directly yield $\np$-hardness of the EB problem for LMCs.
  This reduction is from the subset sum problem.
  Given an instance of the subset sum problem, this reduction constructs an MDP in which transition probabilities are proportional to the subset sum inputs, and thus may differ from $0$ and $1$.
  This reduction hinges on these transition probabilities to enforce a quantitative criterion.
  Because an instance of the EB problem on LMCs cannot directly encode these numbers, we gave a direct reduction from SAT in Appendix~\ref{appendix:ele:np-hardness}.
  \hfill$\lhd$
\end{remark}

Unlike their counterparts in Markov chains, the UL (resp.~EL) and UB (resp.~EB) problems are not equivalent, as indicated by the different complexities for each problem (cf.\ Table~\ref{table:complexity}).
Furthermore, for the EL problem for states in MDPs and the EB problem in MDPs, strategies that use randomisation may be necessary to witness the perturbation equivalence of given states (e.g.,~\cite[Figure~5]{DBLP:conf/fsttcs/Kiefer020}), which contrasts with the situation in Markov chains where deterministic transition functions suffice (\cref{thm:ele-epb}).

\subparagraph*{Reducing from Markov chains to MDPs.}
We briefly explain how to derive, given a labelled Markov chain $\mchain = \lmc$, a labelled MDP $\mdp = (\states, \actions, \transB, \lblSet, \lbl)$ with the same state space and labelling such that any instance of the UL or EL (resp.~UB or EB) on $\mchain$ is equivalent to an instance of the corresponding problem on the MDP $\mdp$ with the same input distributions (resp.~states).
Intuitively, we construct $\mdp$ so that any transition in $\mchain$ corresponds to an action with a deterministic outcome.
This way, any memoryless strategy in $\mdp$ can be seen as a consistent transition function with respect to $\trans$ and vice-versa.

Formally, we define the set of actions of $\mdp$ to be $\actions = \states$.
For all $\state\in\states$, the set of actions enabled in $\state$ is $\actions(\state) = \supp{\trans(\state)}$, i.e., the successors of $\state$ in $\mchain$.
For all $\state\in\states$ and $\stateB\in\actions(\state)$, we let $\transB(\state, \stateB, \stateB) = 1$.

We remark that a memoryless strategy of $\mdp$ is a function $\strat\colon\states\to\dist{\states}$ such that, for all $\state\in\states$, $\supp{\strat(\state)}\subseteq A(\state) = \supp{\trans(\state)}$.
In other words, the set of memoryless strategies of $\mdp$ is the set $\cons{\trans}$.
Hence, to prove that the UL, UB, EL and EB problems on $\mchain$ are equivalent to the corresponding problems on $\mdp$ for the same initial states or distributions, it is sufficient to show that, for all memoryless strategies $\strat\in\cons{\trans}$ of $\mdp$, the Markov chain $(\states, \strat, \lblSet, \lbl)$ (derived from $\mdp$ by replacing $\trans$ by $\strat$) is equal to the induced Markov chain $\mdp^\strat$.
We prove this below.

\begin{lemma}\label{lemma:mc-to-mdp:translation:1}
  For all $\strat\in\cons{\trans}$, the Markov chains $(\states, \strat, \lblSet, \lbl)$ and $\mdp^\strat$ are equal.
\end{lemma}
\begin{proof}
  Let $\transB^\strat$ denote the transition function of $\mdp^\strat$.
  We must prove that $\transB^\strat=\strat$.
  Let $\state, \stateB\in\states$.
  If $\stateB\notin\supp{\trans(\state)}$, then by construction of $\mdp$, we have $\transB^\strat(\state, \stateB) = 0 = \strat(\state, \stateB)$ (there is no direct transition from $\state$ to $\stateB$).
  Assume now that $\stateB\in\supp{\trans(\state)} = \actions(\state)$.
  By definition of $\transB$, we have, for all $\stateC\in\actions(\state)$, $\transB(\state, \stateC, \stateC) = 1$.
  It follows that
  \begin{equation*}
    \transB^\strat(\state, \stateB)
     = \sum_{\stateC\in\actions(\state)} \strat(\state, \stateC)\cdot \transB(\state, \stateC, \stateB)
     = \strat(\state, \stateB).
   \end{equation*}
   We have shown that $\transB^\strat = \strat$, which ends the proof.
\end{proof}

Lemma~\ref{lemma:mc-to-mdp:translation:1} implies the following proposition.
\begin{proposition}
  Two distributions over $\states$ are UL (resp.~EL) in $\mchain$ if and only if they are language-equivalent under all memoryless strategies of $\mdp$ (resp.~some memoryless strategy of $\mdp$).
  Two states are UB (resp.~EB) in $\mchain$ if and only if they are bisimilar under all memoryless strategies of $\mdp$ (resp.~some memoryless strategy of $\mdp$).
\end{proposition}

\subsection{The UB problem in MDPs}\label{section:mdps:upb}

We establish the $\ptime$-hardness of the UB problem in MDPs in this section.
This proves that the UB problem is harder in MDPs than in Markov chains, where it is $\nlogspace$-complete (Theorem~\ref{theorem:ule:d:complexity}).

We provide a direct reduction from the problem of deciding if two states are bisimilar in a labelled Markov chain, which is $\ptime$-complete~\cite{DBLP:conf/fossacs/ChenBW12}.
Intuitively, we construct an MDP with a single action such that the probabilistic transitions in the MDP from any state with this action match those of the original Markov chain.
By design, the Markov chain induced by the single (memoryless) strategy of this MDP is identical to the initial Markov chain.
In particular, two states are bisimilar in the initial Markov chain if and only if they are bisimilar under all memoryless strategies of the MDP.
By combining the $\ptime$-hardness of the UB problem that follows from this reduction and the $\ptime$-membership shown in~\cite{DBLP:conf/fsttcs/Kiefer020}, we obtain the following theorem.
\begin{theorem}\label{theorem:upb-mdp:complexity}
  The UB problem in MDPs is $\ptime$-complete.
\end{theorem}
\begin{proof}
  We need only prove $\ptime$-hardness.
  We formalise the reduction described above.
  Let $\mchain = \lmc$ be a labelled Markov chain.
  We consider the MDP $\mdp = (\states, \{\action\}, \transB, \lblSet, \lbl)$ where $\transB$ is defined by, for all $\state, \stateB\in\states$, $\transB(\state, \action, \stateB) = \trans(\state, \stateB)$.
  It is easily checked that for all $\state, \stateB\in\states$, $\state$ and $\stateB$ are bisimilar in $\mchain$ if and only if they are under all (memoryless) strategies of $\mdp$.
\end{proof}

\subsection{The UL problem in MDPs}\label{section:mdps:ule}

We now show that the UL problem in MDPs is logspace-equivalent to the problem of checking language equivalence in labelled Markov chains.
The language equivalence problem in labelled Markov chains can be reduced to the UL problem on MDPs with no non-determinism with the construction used to prove Theorem~\ref{theorem:upb-mdp:complexity}.
In the following, we present a reduction from the UL problem in MDPs to the language-equivalence problem in Markov chains that is analogous to a reduction presented in~\cite{DBLP:journals/lmcs/FijalkowKS20}.
This reduction implies that the UL problem in MDPs is in the class $\ceql$~\cite{DBLP:conf/stacs/CernyS26} (this recent upper bound refines the $\nc$ upper bound of Tzeng~\cite{DBLP:journals/ipl/Tzeng96}); this new upper bound refines the $\ptime$-membership result asserted in~\cite{DBLP:conf/fsttcs/Kiefer020}.

The remainder of the section is structured as follows.
We introduce transition-labelled MDPs in Section~\ref{section:ule-mdp:transition-based} to simplify the presentation of the reduction.
Section~\ref{section:ule-mdp:equiv} presents background on language-equivalence of weighted automata and introduces analogous notions for the UL problem.
We formalise our reduction in Section~\ref{section:ule-mdp:reduction}.

\subsubsection{Transition-labelled systems}\label{section:ule-mdp:transition-based}

We introduce transition-labelled Markov chains and MDPs in this section.
In the definitions of labelled Markov chains and MDPs in earlier sections, labels were assigned to states; henceforth, we refer to these as \textit{state-labelled} Markov chains and MDPs.
In transition-labelled systems, transitions carry labels, and there can be several transitions with different labels from one state to another.

This section is only concerned with relevant definitions for transition-labelled processes.
We discuss the inter-reducibility of perturbation equivalence problems in state-labelled and transition-labelled systems in Section~\ref{section:mdps:perturbation-equivalence-interreducibility}.

\subparagraph*{Markov chains.}
A \textit{transition-labelled Markov chain} is a tuple $\mchain = (\states, \lblSet, (\ltrans{a}{})_{a\in\lblSet})$, where $\states$ is a finite set of states, $\lblSet$ is a finite set of labels, and, for all $a\in\lblSet$, $\ltrans{a}{}\in\IQ_{\geq 0}^{\states\times\states}$, and $\sum_{a\in\lblSet}\ltrans{a}{}$ is a (row) stochastic matrix.
In the following, given $\state\in\states$, $\eClass\subseteq\states$ and $a\in\lblSet$, we write $\ltrans{a}{}(\state, \eClass)$ as shorthand for $\sum_{\stateB\in\eClass}\ltrans{a}{}(\state, \stateB)$.

Paths of $\mchain$ are sequences $\state_0a_0\state_1a_1\ldots\in(\states\lblSet)^\omega$ such that for all $\iPos\in\IN$, $\ltrans{a_\iPos}{}(\state_\iPos,\state_{\iPos+1}) > 0$, and histories of $\mchain$ are prefixes of paths ending in a state.
Given a history $\hist$ of $\mchain$, we let the cylinder of $\hist$, denoted by $\cyl{\hist}$, be the set of paths that extend $\hist$.

Given an initial distribution $\measure\in\dist{\states}$, we define the distribution $\probLV{\measure}{}$ over the set of paths of $\mchain$ from $\measure$ is defined for cylinder sets by letting, for any history $\state_0a_0\state_1\ldots a_{\iLast-1}\state_\iLast$ of $\mchain$
\[
  \probLV{\measure}{}(\cyl{\state_0a_0\state_1\ldots a_{\iLast-1}\state_\iLast}) = \measure(\state_0)\cdot\prod_{\iPos = 0}^{\iLast-1} \ltrans{a_\iPos}{}(\state_\iPos,\state_{\iPos+1}).
\]
As before, for any measurable $\event\subseteq\lblSet^\omega$, we write $\probLV{\measure}{}(\event)$ for the probability of plays whose sequence of labels is in $\event$, and, for all $\word\in\lblSet^*$, we write $\probLV{\measure}{}(\word)$ as shorthand for $\probLV{\measure}{}(\word\lblSet^\omega)$.

For all $\word=a_1\ldots a_\iLast\in\lblSet^*$, we let $\ltrans{\word}{}$ be the identity matrix if $\word$ is the empty word, and otherwise let 
$\ltrans{\word}{} = \ltrans{a_1}{}\cdot\ldots\cdot\ltrans{a_\iLast}{}.$
For all $\word\in\lblSet^*$ and initial distributions $\measure\in\dist{\states}$, viewed below as a row vector, we have $\probLV{\measure}{}(\word) = \measure\cdot \ltrans{\word}{}\cdot\oneVect$, where $\oneVect\in\IQ^\states$ is the column vector whose components are all $1$.

Distributions $\measure, \measureB\in\dist{\states}$ are \textit{language-equivalent} in $\mchain$ if for all measurable $\event\subseteq\lblSet^\omega$, $\probLV{\measure}{}(\event) = \probLV{\measureB}{}(\event)$.
An equivalence relation $R\subseteq \states\times\states$ is a probabilistic bisimulation if, for all $(\state, \stateB)\in R$, all labels $a\in\lblSet$ and all equivalence classes $\eClass\in\states/R$, $\ltrans{a}{}(\state, \eClass) = \ltrans{a}{}(\stateB, \eClass)$.
States $\state, \stateB\in\states$ are bisimilar if there exists a probabilistic bisimulation $R$ such that $(\state, \stateB)\in R$.
The coarsest bisimulation relation is called bisimilarity.

\subparagraph*{MDPs.}
A \textit{transition-labelled MDP} is a tuple $\mdp = (\states, \actions, \lblSet, (\ltrans{a}{\action})_{a\in\lblSet, \action\in\actions})$, where $\states$ is a finite set of states, $\actions$ is a finite set of actions, $\lblSet$ is a finite set of labels and, for all $a\in\lblSet$ and $\action\in\actions$, $\ltrans{a}{\action}\in\IQ_{\geq 0}^{\states\times\states}$ is such that for all $\state\in\states$ and all $\action\in\actions$, $\sum_{a\in\lblSet}\sum_{\stateB\in\states}\ltrans{a}{\action}(\state,\stateB) \in\{0, 1\}$.
An action $\action\in\actions$ is enabled in a state $\state\in\states$ if $\sum_{a\in\lblSet}\sum_{\stateB\in\states}\ltrans{a}{\action}(\state,\stateB) = 1$.
As before, we write $\actions(\state)$ for the set of actions enabled in $\state$ and assume that there is an enabled action in each state.

Memoryless strategies are defined as in state-labelled MDPs.
Let $\strat\colon\states\to\dist{\actions}$ be a memoryless strategy.
For all $a\in\lblSet$ , we let $\ltrans{a}{\strat}$ be the matrix defined, for all $\state,\stateB\in\states$, by 
\[\ltrans{a}{\strat}(\state,\stateB) = \sum_{\action\in\actions}\strat(\state)(\action)\cdot\ltrans{a}{\action}(\state,\stateB).\]
We note that $\sum_{a\in\lblSet}\ltrans{a}{\strat}$ is a (row) stochastic matrix.

The Markov chain induced by $\strat$ on $\mdp$ is the transition-labelled Markov chain $\mdp^\strat= (\states, \lblSet, (\ltrans{a}{\strat})_{a\in\lblSet})$.
For all $\word=a_1\ldots a_\iLast\in\lblSet^*$, we define the matrix $\ltrans{\word}{\strat}$ as for Markov chains.

The UL, UB, EL and EB relations are defined in transition-labelled MDPs in the same way as in state-labelled MDPs.

\subsubsection{Backward spaces}\label{section:ule-mdp:equiv}

In this section, we recall the definition of the backward space of a labelled Markov chain and a characterisation of language-equivalent distributions through this backward space.
We also introduce a variant of the backward space for universal language equivalence.
We refer the reader to~\cite{DBLP:journals/corr/abs-2009-01217} for a more extensive overview on the equivalence of weighted automata (which subsume labelled Markov chains).

\subparagraph*{Backward space of a labelled Markov chain.}
Let $\mchain = (\states, \lblSet, (\ltrans{a}{})_{a\in\lblSet})$ be a transition-labelled Markov chain.
The \textit{backward space} of $\mchain$ is the vector subspace of $\IQ^{\states}$ defined by
\[
  B_\mchain = \spanVect{
    \{\ltrans{\word}{}\cdot \oneVect\mid\word\in\lblSet^*\}
  },\]
(where $\spanVect{\cdot}$ denotes the linear span).
We note that the generators of $B_\mchain$ are vectors of the form $\left(\probLV{\state}{}(\word)\right)_{\state\in\states}$ where $\word\in\lblSet^*$.
We recall the following characterisation of $B_\mchain$ (which is easy to check).
\begin{lemma}\label{lemma:back:chain}
  The backward space $B_\mchain$ of $\mchain$ is the smallest (column) vector subspace $\vectSpace$ of $\IQ^\states$ such that $\oneVect\in\vectSpace$ and $\ltrans{a}{}\vectSpace \subseteq\vectSpace$ for all $a\in\lblSet$.
\end{lemma}

Two distributions $\measure, \measureB\in\dist{\states}$ (seen as row vectors) are language-equivalent in $\mchain$ if and only if $\measure\cdot \ltrans{\word}{}\cdot \oneVect = \measureB\ltrans{\word}{}\cdot \oneVect$ for all $\word\in\lblSet^*$, i.e., if the difference $\measure - \measureB$ is orthogonal to (all generators of) $B_\mchain$.

\subparagraph*{Backward spaces for UL.}

We fix a transition-labelled MDP $\mdp = (\states, \actions, \lblSet, (\ltrans{a}{\action})_{a\in\lblSet, \action\in\actions})$.
For any memoryless strategy $\strat$ of $\mdp$, let $B_\strat$ denote the backward space of the induced Markov chain $\mdp^\strat$, given by
\[B_\strat = \spanVect{
    \{\ltrans{\word}{\strat}\cdot \oneVect\mid\word\in\lblSet^*\}
  }.\]
By the above, $\measure$ and $\measureB$ are language-equivalent in $\mdp^\strat$ if and only if $\measure-\measureB\in B_\strat^\perp$ where $B_\strat^\perp$ denotes the orthogonal complement of $B_\strat$.
It follows that two distributions are UL in $\mdp$ if and only if they are in $B_\strat^\perp$ for all memoryless strategies $\strat$.

\begin{lemma}\label{lemma:back:ule-space}
  Let $B_\mdp = \sum_{\strat} B_\strat =\spanVect{\bigcup_{\strat} B_\strat}$, where $\strat$ ranges over all memoryless strategies of $\mdp$.
  For all $\measure, \measureB\in\dist{\states}$, $\measure$ and $\measureB$ are UL in $\mdp$ if and only if $\measure-\measureB\in B_\mdp^\perp$.
\end{lemma}

To reduce the UL problem to checking language-equivalence in Markov chains, the main idea is to construct, from the MDP $\mdp$, a labelled Markov chain with state space $\states$ whose backward space is $B_\mdp = \sum_{\strat} B_\strat$.
We provide a variation of the result of Lemma~\ref{lemma:back:chain} for $B_\mdp$.
This property is used to prove the correctness of our upcoming reduction.

\begin{lemma}\label{lemma:back:ule-smallest}
  Let $B_\mdp= \sum_{\strat} B_\strat$.
  Then $B_\mdp$ is the smallest (column) vector subspace $\vectSpace$ of $\IQ^\states$ such that $\oneVect\in\vectSpace$ and $\ltrans{a}{\strat}\vectSpace \subseteq\vectSpace$ for all $a\in\lblSet$ and memoryless strategies $\strat$ of $\mdp$.
\end{lemma}

Lemma~\ref{lemma:back:ule-smallest} is a direct consequence of the following property.
\begin{lemma}[Lem.~4 of~\cite{DBLP:conf/fsttcs/Kiefer020}]\label{lemma:back:ule-technical}
  Let $B_\mdp= \sum_{\strat} B_\strat$.
  Then 
  \[B_\mdp = \spanVect{
    \{\ltrans{a_1}{\strat_1}\cdot\ldots\cdot\ltrans{a_\iLast}{\strat_\iLast} \oneVect\mid\iLast\geq 1, a_1, \ldots, a_\iLast\in\lblSet, \strat_1, \ldots, \strat_\iLast\text{ memoryless strategies}\}
  }.\]
\end{lemma}

We provide a brief proof of Lemma~\ref{lemma:back:ule-smallest} in the interest of completeness.
\begin{proof}[Proof of Lemma~\ref{lemma:back:ule-smallest}]
  Say that a column-vector set $D\subseteq\IQ^\states$ is stable if $\ltrans{a}{\strat} D \subseteq D$ for all $a\in\lblSet$ and memoryless strategies $\strat$ of $\mdp$.
  We first show that $B_\mdp$ is stable.
  This follows from the set of generators of $B_\mdp$ given in Lemma~\ref{lemma:back:ule-technical} being stable.
  Second, we must show that if a vector subspace $\vectSpace$ is stable and is such that $\oneVect\in\vectSpace$, then $B_\mdp\subseteq\vectSpace$.
  This follows from all generators of $B_\mdp$ being in any such $\vectSpace$.
\end{proof}

\subsubsection{Reducing UL in MDPs to trace-equivalence in Markov chains}\label{section:ule-mdp:reduction}

Fix a transition-labelled MDP $\mdp = (\states, \actions, \lblSet, (\ltrans{a}{\action})_{a\in\lblSet, \action\in\actions})$.
We explain how to construct a Markov chain $\mchain^\star$ over $\states$ from $\mdp$ whose backward space is $B_\mdp$, i.e., such that for all $\measure, \measureB\in\dist{\states}$, $\measure$ and $\measureB$ are UL in $\mdp$ if and only if they are language equivalent in $\mchain^\star$.

Fix an arbitrary pure memoryless strategy $\strat_0$ of $\mdp$, and define
\[\Sigma = \{\strat_0[\state\leftarrow\action]\mid \state\in\states,\action\in\actions\}\]
where for all $\state\in\states$ and $\action\in\actions$, $\strat_0[\state\leftarrow\action]$ is a pure memoryless strategy that agrees with $\strat_0$ in all states besides $\state$ and such that $\strat_0[\state\leftarrow\action](\state) = \action$.
In particular, $\strat_0\in\Sigma$.
The key observation is that checking that a vector space $\vectSpace$ satisfies $\ltrans{a}{\strat}\vectSpace\subseteq\vectSpace$ for all $a\in\lblSet$ and memoryless strategies $\strat$ can be reduced to checking these inclusions only for the strategies in $\Sigma$.
\begin{lemma}\label{lemma:back:finite-check}
  Let $\vectSpace\subseteq\IQ^\states$ be a vector space.
  Then $\ltrans{a}{\strat'}\vectSpace\subseteq\vectSpace$ for all $a\in\lblSet$ and all $\strat'\in\Sigma$ if and only if $\ltrans{a}{\strat}\vectSpace\subseteq\vectSpace$ for all $a\in\lblSet$ and all memoryless strategies $\strat$.
\end{lemma}
\begin{proof}
  Let $a\in\lblSet$ and $\strat$ be a memoryless strategy.
  We have $\ltrans{a}{\strat}\in\spanVect{\{\ltrans{a}{\strat'}\mid\strat'\in\Sigma\}}$ by observing that
  \begin{equation}\label{eq:ule-mdp-finite}
    \ltrans{a}{\strat} = \ltrans{a}{\strat_0} +
    \sum_{\state\in\states}\left(
      -\ltrans{a}{\strat_0} +
      \sum_{\action\in\actions}\strat(\state)(\action)\cdot\ltrans{a}{\strat_0[\state\leftarrow\action]}
    \right).\end{equation}
  The claim of the lemma follows from the above and linearity.
\end{proof}

We now describe the labelled Markov chain $\mchain^\star= (\states, \lblSet^\star, (\ltrans{a}{\star})_{b\in\lblSet^\star})$.
We define the set of labels of $\mchain^\star$ as $\lblSet^\star = \lblSet\times\Sigma$ and for all $b = (a, \strat)\in\lblSet^\star$, we define the transition matrix $\ltrans{b}{\star}$ as $\frac{1}{|\Sigma|}\cdot \ltrans{a}{\strat}$.
The following is a direct consequence of all previous technical lemmas.
\begin{proposition}
  Let $\measure, \measureB\in\dist{\states}$.
  Then $\measure$ and $\measureB$ are language-equivalent in $\mchain^\star$ if and only if $\measure$ and $\measureB$ are UL in $\mdp$.
\end{proposition}
\begin{proof}
  It suffices to show that the backward space $B_\star$ of $\mchain^\star$ is equal to the span $B_\mdp$ of the union of backward spaces of the induced Markov chains $\mdp^\strat$ where $\strat$ ranges over all memoryless strategies.
  Both inclusions can be shown using the respective characterisations of $B_\star$ and $B_\mdp$ in Lemmas~\ref{lemma:back:chain} and~\ref{lemma:back:ule-smallest}.
  We note that $\oneVect\in B_\star$ and $\oneVect\in B_\mdp$.

  We first consider the inclusion $B_\star\subseteq B_\mdp$.
  By Lemmas~\ref{lemma:back:chain}, it suffices to check that for all $a\in\lblSet$ and all strategies $\strat\in\Sigma$, $\frac{1}{|\Sigma|}\cdot \ltrans{a}{\strat} B_\mdp\subseteq B_\mdp$.
  This follows from Lemma~\ref{lemma:back:ule-smallest}.

  We now consider the inclusion $B_\mdp\subseteq B_\star$.
  In this case, it suffices to check that for all $a\in\lblSet$ and all memoryless strategies $\strat$, $\ltrans{a}{\strat} B_\star\subseteq B_\star$.
  This property holds for strategies in $\Sigma$ by Lemma~\ref{lemma:back:chain} and extends to all memoryless strategies by Lemma~\ref{lemma:back:finite-check}.
\end{proof}

The Markov chain $\mchain^\star$ can be constructed from $\mdp$ using only logarithmic space.
We have therefore shown that the UL problem in MDPs can be reduced in logspace to language-equivalence in labelled Markov chains.
In particular, both problems are in the complexity class $\ceql$~\cite{DBLP:conf/stacs/CernyS26}.

\begin{theorem}\label{theorem:mdp-ule:complexity}
  The UL problem in labelled MDPs is logspace-equivalent to the language-equivalence problem in labelled Markov chains. Both problems are in $\ceql$.
\end{theorem}

\subsection{Perturbation equivalence in state-labelled and transition-labelled MDPs}\label{section:mdps:perturbation-equivalence-interreducibility}

The goal of this section is to show that the perturbation equivalence problems in state-labelled and transition-labelled are inter-reducible in logarithmic space.
This implies that the complexity results discussed in earlier sections apply to both models of MDPs.
It is known that state-labelled and transition-labelled presentations of stochastic systems are equivalent in terms of modelling expressiveness.

Informally, a state-based MDP can be transformed into an equivalent transition-labelled MDP by pushing state labels onto outgoing transitions.
Conversely, a transition-labelled MDP can be converted into an equivalent state-labelled MDP by augmenting states with labels and pushing labels on transitions to the incoming states.

We present constructions that preserve all variants of perturbation equivalence.
The subtler case is going from transition-labelled MDPs to state-labelled MDPs.
In this case, care must be taken when augmenting the state space to ensure that the memoryless strategies of the state-labelled MDP do not allow more behaviours than the strategies of the original MDP (as this would be harmful to universal equivalence or benefit existential equivalence).
The construction outlined below avoids this issue.

In the sequel, we first present a transformation from state-labelled MDPs to a transition-labelled ones, and then how to derive a state-labelled MDP from a transition-labelled MDP.

\subparagraph*{From state to transition labels.}
Fix a state-labelled MDP $\mdp_{\mathsf{st}} = (\states, \actions, \trans, \lblSet, \lbl)$.
We construct a transition-labelled MDP $\mdp_{\mathsf{tr}} = (\states, \actions, \lblSet, (\ltrans{a}{\action})_{a\in\lblSet, \action\in\actions})$ with the same state and action spaces by pushing the labels of states onto their outgoing transitions.
Transition matrices are defined by letting, for all $\state, \stateB\in\states$, $\action\in\actions$ and $a\in\lblSet$,
\[\ltrans{a}{\action}(\state, \stateB) =
  \begin{cases}
    \trans(\state, \action, \stateB) & \text{if } \action\in\actions(\state) \text{ and } a=\lbl(\state) \\
    0 & \text{otherwise.}
  \end{cases}
\]

The sets of memoryless strategies of $\mdp_{\mathsf{st}}$ and $\mdp_{\mathsf{tr}}$ coincide.
We now show that any instance of perturbation equivalence on $\mdp_{\mathsf{st}}$ is equivalent to the corresponding instance on $\mdp_{\mathsf{tr}}$.
To this end, we observe that for all memoryless strategies of $\mdp_{\mathsf{st}}$, two states are bisimilar (resp.~two distributions are language equivalent) in the induced Markov chain $\mdp_{\mathsf{st}}^\strat$ if and only if they are in $\mdp_{\mathsf{tr}}^\strat$.

We first consider bisimilarity.
\begin{lemma}
  Let $\strat$ be a memoryless strategy of $\mdp_{\mathsf{st}}$ and $\mdp_{\mathsf{tr}}$.
  For all $\state, \stateB\in\states$, $\state$ and $\stateB$ are bisimilar in $\mdp_{\mathsf{st}}^\strat$ if and only if they are in $\mdp_{\mathsf{tr}}^\strat$.
\end{lemma}
\begin{proof}
  We first show that the bisimilarity relation $\bisim{\mathsf{st}}$ of $\mdp_{\mathsf{st}}^\strat$ is a bisimulation in $\mdp_{\mathsf{tr}}^\strat$.
  Let $\state, \stateB\in\states$ such that $\state\bisim{\mathsf{st}}\stateB$.
  We note that, in particular, $\lbl(\state) =\lbl(\stateB)$.
  Let $a\in\lblSet$ and $C\in\states/\bisim{\mathsf{st}}$.
  If $a\neq\lbl(\state) = \lbl(\stateB)$, we have $\ltrans{a}{\strat}(\state, C) = 0 = \ltrans{a}{\strat}(\stateB, C)$.
  If $a=\lbl(\state) = \lbl(\stateB)$, we obtain
  \[
    \ltrans{a}{\strat}(\state, C) =
    \sum_{\action\in\actions(\state)}\strat(\state, \action)\cdot\trans(\state, \action, C) =
    \sum_{\action\in\actions(\stateB)}\strat(\stateB, \action)\cdot\trans(\stateB, \action, C) =
    \ltrans{a}{\strat}(\stateB, C),
  \]
  where the first and last equality follow by construction of $\mdp_{\mathsf{tr}}$ and the middle equality follows from $\state\bisim{\mathsf{st}}\stateB$.
  We have shown that $\bisim{\mathsf{st}}$ is a bisimulation in $\mdp_{\mathsf{tr}}^\strat$.

  We now show that the bisimilarity relation $\bisim{\mathsf{tr}}$ of $\mdp_{\mathsf{tr}}^\strat$ is a bisimulation in $\mdp_{\mathsf{st}}^\strat$.
  Let $\state, \stateB\in\states$ such that $\state\bisim{\mathsf{tr}}\stateB$.
  First, we obtain that $\lbl(\state) = \lbl(\stateB)$ from the definition of $\mdp_{\mathsf{tr}}$ and the relation
  \[ 1 = \ltrans{\lbl(\state)}{\strat}(\state, \states) = \ltrans{\lbl(\state)}{\strat}(\stateB, \states),\]
  where the first equality is by definition of $\mdp_{\mathsf{tr}}$ and the second follows from $\state\bisim{\mathsf{tr}}\stateB$ and $\states$ being a union of equivalence classes of $\bisim{\mathsf{tr}}$.
  Second, by a similar reasoning as above, we have that for all $\eClass\in\states/\bisim{\mathsf{tr}}$,
  \[
    \sum_{\action\in\actions(\state)}\strat(\state, \action)\cdot\trans(\state, \action, C) =
    \ltrans{a}{\strat}(\state, C) =
    \ltrans{a}{\strat}(\stateB, C) =
    \sum_{\action\in\actions(\stateB)}\strat(\stateB, \action)\cdot\trans(\stateB, \action, C).
  \]
  This shows that $\bisim{\mathsf{tr}}$ is a bisimulation in $\mdp_{\mathsf{st}}^\strat$.
\end{proof}

We now deal with language equivalence.
\begin{lemma}
  Let $\strat$ be a memoryless strategy of $\mdp_{\mathsf{st}}$ and $\mdp_{\mathsf{tr}}$.
  For all $\measure, \measureB\in\dist{\states}$, $\measure$ and $\measureB$ are language-equivalent in $\mdp_{\mathsf{st}}^\strat$ if and only if they are in $\mdp_{\mathsf{tr}}^\strat$.  
\end{lemma}
\begin{proof}
  Let $\measure\in\dist{\states}$.
  By construction of $\mdp_{\mathsf{tr}}$, for all histories $\hist = \state_0\state_1\state_2\ldots\state_\iLast$ of $\mdp_{\mathsf{st}}^\strat$, the probability of the cylinder $\cyl{\hist}$ in $\mdp_{\mathsf{st}}^\strat$ from the initial distribution $\measure$ is equal to the probability of $\cyl{\state_0\lbl(\state_0)\state_1\lbl(\state_1)\ldots\lbl(\state_{\iLast-1})\state_\iLast}$ in $\mdp_{\mathsf{tr}}^\strat$ from $\mu$.
  This implies that for all measurable $\event\subseteq\lblSet^\omega$, the probability of $\event$ coincides in $\mdp_{\mathsf{st}}^\strat$ from $\mu$ and in $\mdp_{\mathsf{tr}}^\strat$ from $\mu$.
  The claim of the lemma follows.
\end{proof}

\subparagraph*{From transition to state labels.}

Let $\mdp_{\mathsf{tr}} = (\states, \actions, \lblSet, (\ltrans{a}{\action})_{a\in\lblSet, \action\in\actions})$ be a transition-labelled MDP.
We first provide some intuition on how we derive a state-labelled MDP from $\mdp_{\mathsf{tr}}$.

A first idea would be to consider the augmented state space $\states\times\lblSet$ and design transitions such that moving from a pair $(\state, a)$ to a state $(\stateB, b)$ with some action $\action\in\actions(\state)$ is given by $\ltrans{b}{\action}$.
This can be seen as pushing labels from transitions onto outgoing states.
The main issue with this approach is that there is no direct correspondence between memoryless strategies of the new state-labelled MDP and the original MDP: a memoryless strategy of the state-labelled MDP is allowed to make different decision in different copies of a same state.

Instead, we construct a state-labelled MDP $\mdp_{\mathsf{st}}$ with the state space $\states_{\mathsf{st}} = \states\cup(\lblSet\times\states)$, a new additional action $\bot\notin\actions$ and a new label $\varepsilon\notin\lblSet$.
The notation for the new label $\varepsilon$ references the empty word; this label $\varepsilon$ can be seen as a neutral label in the following.
Formally, we define $\mdp_{\mathsf{st}} = (\states_{\mathsf{st}}, \actions\cup\{\bot\}, \trans, \lblSet\cup\{\varepsilon\}, \lbl)$ where the transition function $\trans$ is given by, for all $\state, \stateB\in\states$, $a\in\lblSet$ and $\action\in\actions(\state)$, $\trans(\state, \action, (a, \stateB)) = \ltrans{a}{\action}(\state, \stateB)$ and $\trans((a, \stateB), \bot, \stateB) = 1$, and the labelling function $\lbl$ is given, for all $\state\in\states$ and $a\in\lblSet$, $\lbl(\state) = \varepsilon$ and $\lbl((a, \state)) = a$.

Intuitively, $\mdp_{\mathsf{st}}$ is obtained by splitting transitions of $\mdp_{\mathsf{tr}}$ into two steps.
In the first step, the label and successor state are chosen, but no decision can be made yet.
In the second step, we move to a state with a neutral label, in which non-determinism is available as in $\mdp_{\mathsf{tr}}$.
With this construction, each state of $\mdp_{\mathsf{tr}}$ has a single counterpart in which there is non-deterministic choice in $\mdp_{\mathsf{st}}$ (i.e., itself).
In particular, there is a natural bijection between memoryless strategies of $\mdp_{\mathsf{tr}}$ and those of $\mdp_{\mathsf{st}}$.
For this reason, we identify memoryless strategies of $\mdp_{\mathsf{tr}}$ and of $\mdp_{\mathsf{st}}$ in the sequel.

We now show that for all memoryless strategies $\strat$ and all pairs of states of $\mdp_{\mathsf{tr}}$, these states are bisimilar in $\mdp_{\mathsf{tr}}^\strat$ if and only if they are bisimilar in $\mdp_{\mathsf{st}}^\strat$.
\begin{lemma}
  Let $\strat$ be a memoryless strategy of $\mdp_{\mathsf{tr}}$.
  For all $\state, \stateB\in\states$, $\state$ and $\stateB$ are bisimilar in $\mdp_{\mathsf{tr}}^\strat$ if and only if they are in $\mdp_{\mathsf{st}}^\strat$.
\end{lemma}
\begin{proof}
  We first show that any two bisimilar states in $\mdp_{\mathsf{tr}}^\strat$ are bisimilar in $\mdp_{\mathsf{st}}^\strat$.
  Let $\bisim{\mathsf{tr}}$ denote the bisimilarity relation of $\mdp_{\mathsf{tr}}^\strat$.
  We consider the relation $R\subseteq \states_{\mathsf{st}}\times\states_{\mathsf{st}}$ defined by
  \[R = \mathord{\bisim{\mathsf{st}}} \cup \{(a, \state), (a, \stateB)\mid a\in\lblSet,\, \state\bisim{\mathsf{st}}\stateB\}.\]
  We show that $R$ is a bisimulation in $\mdp_{\mathsf{st}}^\strat$.
  By construction of $\mdp_{\mathsf{st}}$, any pair in $R$ shares the same label.
  In particular, any pair in $R$ is either in $\states\times\states$ or in $(\lblSet\times\states)\times(\lblSet\times\states)$.

  We now check that for all pairs in $R$, the outgoing probability to each equivalence class is equal from each state.
  First, let $(\state, \stateB)\in R\cap (\states\times\states)$.
  It follows from $\state\bisim{\mathsf{tr}} \stateB$ that for all $a\in\lblSet$ and all equivalence classes $\eClass\in\states/\bisim{\mathsf{tr}}$, we have
  \[
    \sum_{\action\in\actions(\state)}\strat(\state, \action)\cdot \trans(\state, \action, \{a\}\times\eClass) =
    \ltrans{a}{\strat}(\state, \eClass) =
    \ltrans{a}{\strat}(\stateB, \eClass) =
    \sum_{\action\in\actions(\stateB)}\strat(\stateB, \action)\cdot \trans(\stateB, \action, \{a\}\times\eClass),
  \]
  where the second equality follows from the bisimilarity assumption in $\mdp_{\mathsf{tr}}^\strat$.
  For the second case, let $((a, \state), (a, \stateB))\in R\cap ((\lblSet\times\states)\times(\lblSet\times\states))$.
  By definition of $R$, we have $\state\bisim{\mathsf{tr}} \stateB$.
  Letting $\eClass\in\states/R$ denote the equivalence class of $\state$ and $\stateB$ (which is also their equivalence class with respect to $\mathord{\bisim{\mathsf{tr}}}$), it follows from the definition of $\mdp_{\mathsf{st}}$ that
  \begin{align*}
    \sum_{\action\in\actions(\state)}\strat((a, \state), \action)\cdot \trans((a, \state), \action, \eClass)
    & =
      \trans((a, \state),\bot,  \state) \\
    & =
      1\\
    & =
      \trans((a, \stateB),\bot,  \stateB)\\
    &=
    \sum_{\action\in\actions(\stateB)}\strat((a, \stateB), \action)\cdot \trans((a, \stateB), \action, \eClass).
  \end{align*}

  We have shown that $R$ is a bisimulation with respect to $\mdp_{\mathsf{st}}^\strat$.
  It follows that any two bisimilar states in $\mdp_{\mathsf{tr}}^\strat$ are bisimilar in $\mdp_{\mathsf{st}}^\strat$.
  
  We now show that for all $\state, \stateB\in\states$, if $\state$ and $\stateB$ are bisimilar in $\mdp_{\mathsf{st}}^\strat$, then they are bisimilar in $\mdp_{\mathsf{tr}}^\strat$.
  Let $\mathord{\bisim{\mathsf{st}}}$ denote the bisimilarity relation of $\mdp_{\mathsf{st}}^\strat$.
  We first show a technical claim: for all $\eClass\in\states_{\mathsf{st}}/\bisim{\mathsf{st}}$ such that $\eClass\subseteq\states$ (i.e., $\eClass$ is an equivalence class of $\bisim{\mathsf{st}}$ with only $\varepsilon$-labelled states of $\mdp_{\mathsf{st}}$) and all $a\in\lblSet$, it holds that $\{a\}\times\eClass\in\states_{\mathsf{st}}/\bisim{\mathsf{st}}$.
  To this end, we fix $\state, \stateB\in\eClass$ and let $R_{s, t}$ be the equivalence relation obtained from $\mathord{\bisim{\mathsf{st}}}$ by merging the equivalence classes of $(a, \state)$ and $(a, \stateB)$ into a single equivalence class (by taking their union).
  We show that $R_{s, t}$ is a bisimulation (and therefore $R_{s, t} = \mathord{\bisim{\mathsf{st}}}$).

  Any two states in relation by $R_{s, t}$ share the same label.
  For all $(\stateC_1, \stateC_2)\in R_{s, t}$ such that $\stateC_1\bisim{\mathsf{st}}\stateC_2$, the outgoing probability from $\stateC_1$ and $\stateC_2$ to each equivalence class of $R_{s, t}$ is identical in $\mdp_{\mathsf{st}}^\strat$ (given that $R_{s, t}$ is coarser than $\mathord{\bisim{\mathsf{st}}}$).
  To establish this same property for the pairs $(\stateC_1, \stateC_2)\in R_{s, t}\setminus\mathord{\bisim{\mathsf{st}}}$ (if any exist), it is sufficient (by a transitivity argument) to show that the outgoing probability from $(a, \state)$ and $(a, \stateB)$ to each equivalence class of $R_{s, t}$ is identical in $\mdp_{\mathsf{st}}^\strat$.
  This follows from there being only one outgoing transition from $(a, \state)$ and $(a, \stateB)$ to $\state$ and $\stateB$ respectively, and the fact that $\state, \stateB$ are both in $\eClass$ which is an equivalence class of $R_{s, t}$.
  This concludes the proof of the claim.
  
  We now use the claim to establish that $R = \mathord{\bisim{\mathsf{st}}}\cap (\states\times\states)$ is a bisimulation of $\mdp_{\mathsf{tr}}^\strat$.
  Let $a\in\lblSet$ and $C\in\states/R$.
  We must check that $\ltrans{a}{\strat}(\state, C) = \ltrans{a}{\strat}(\stateB, C)$.
  On the one hand, by construction of $\mdp_{\mathsf{st}}$, we have for $\stateC\in\{\state, \stateB\}$,
  \[
    \ltrans{a}{\strat}(\stateC, C) =
    \sum_{\action\in\actions(\stateC)}\strat(\stateC, \action)\cdot \trans(\stateC, \action, \{a\}\times \eClass).
  \]
  On the other hand, it follows from the previous claim and $\state\bisim{\mathsf{st}}\stateB$ that
  \[
    \sum_{\action\in\actions(\state)}\strat(\state, \action)\cdot
    \trans(\state, \action, \{a\}\times \eClass) =
    \sum_{\action\in\actions(\stateB)}\strat(\stateB, \action)\cdot
    \trans(\stateB, \action, \{a\}\times \eClass)
  \]
  This shows that $R$ is a bisimulation in $\mdp_{\mathsf{tr}}^\strat$.
  It follows that any elements of $\states$ that are bisimilar in $\mdp_{\mathsf{st}}^\strat$ are also bisimilar in $\mdp_{\mathsf{tr}}^\strat$.
\end{proof}

We now consider language equivalence.
\begin{lemma}
  Let $\strat$ be a memoryless strategy of $\mdp_{\mathsf{tr}}$.
  For all $\measure, \measureB\in\dist{\states}$, $\measure$ and $\measureB$ are language-equivalent in $\mdp_{\mathsf{st}}^\strat$ if and only if they are in $\mdp_{\mathsf{tr}}^\strat$.  
\end{lemma}
\begin{proof}
  Let $\measure\in\dist{\states}$.
  Let $\probLV{\measure, \mathsf{tr}}{\strat}$ denote the distribution over plays of $\mdp_{\mathsf{tr}}^\strat$ from $\measure$ and let $\probLV{\measure, \mathsf{st}}{\strat}$ denote the distribution over plays of $\mdp_{\mathsf{st}}^\strat$.
  For all histories $\hist = \state_0a_1\state_1a_1\ldots\state_\iLast$ of $\mdp_{\mathsf{tr}}^\strat$, let $\mathsf{st}(\hist) = \state_0(a_1, \state_1)\state_1(a_1, \state_2)\ldots(a_{\iLast-1}, \state_\iLast)\state_\iLast$ denote the corresponding history of $\mdp_{\mathsf{st}}^\strat$.

  We observe that for all histories $\hist$ of $\mdp_{\mathsf{tr}}^\strat$, we have $\probLV{\measure, \mathsf{tr}}{\strat}(\cyl{\hist}) = \probLV{\measure, \mathsf{st}}{\strat}(\cyl{\mathsf{st}(\hist)})$.
  This follows from the fact that for all $a\in\lblSet$, and all $\state, \stateB\in\states$, we have (by definition)
  \begin{align*}
    \ltrans{a}{\strat}(\state, \stateB)
    & =
    \sum_{\action\in\actions(\state)}\strat(\state, \action)\cdot
      \ltrans{a}{\action}(\state, \stateB) \\
    & =
    \sum_{\action\in\actions(\state)}\strat(\state, \action)\cdot
      \trans(\state, \action, (a, \stateB)) \\
    & = \left(\sum_{\action\in\actions(\state)}\strat(\state, \action)\cdot
      \trans(\state, \action, (a, \stateB))\right)\cdot
      \trans((a, \stateB), \bot, \stateB).
  \end{align*}

  It follows from above that for all $\word = a_1a_2\ldots a_\iLast\in\lblSet^*$, $\probLV{\measure, \mathsf{tr}}{\strat}(\word) = \probLV{\measure, \mathsf{st}}{\strat}(\varepsilon a_1\varepsilon a_2\ldots\varepsilon a_\iLast\varepsilon)$.
  The lemma follows from this property.
\end{proof}

\section{Complete experimental results}
\label{appendix:experiments}
This section contains the complete experimental evaluation. 
We compare the performance of our algorithm to compute UB to that of PRISM's probabilistic bisimilarity implementation, which we denote B.
In order to allow a clear comparison, we implemented \cref{algorithm:stable} in Java as part of PRISM's explicit engine.
Our experiments were run on a MacBook with an M1 chip and 16GB memory, and with the Java virtual machine limited to 8GB.

The complete benchmarking results can be found in \cref{table:complete-results}, while the aggregated results are reported in \cref{table:aggregate}. The synthetic benchmark \emph{haddad-monmege} (which includes $3$ instances) is excluded from the aggregated statistics, as neither algorithm reduces the model and UB completes in less than one millisecond for all instances. For each benchmark family, the table presents the average percentage reduction in the size of the state space achieved by B and UB, together with the average speed-up factor of UB relative to B.

UB yields the same minimised model as B for $74$ benchmark instances ($50\%$ of all evaluated instances).  The largest differences in reduction are observed for the benchmark families \emph{brp} (properties \emph{p1} and \emph{p2}), \emph{crowds}, and \emph{herman}, where B achieves substantially greater reductions.  These models contain many non-trivial probabilistic branches.  Despite the fact that UB always produces an equal or coarser minimised model, the overall average percentage reduction in state space size obtained by UB remains relatively close to that of B.

In terms of runtime, UB is generally faster than, or comparable to, B. The main exception is the \emph{herman} benchmark family, where UB is slower on average, with a speed-up of $0.74$, as the models contain many non-trivial probabilistic branches. On the \emph{nand} benchmark family, which contains the largest and slowest models in the benchmark suite, UB achieves a speed-up of $1.70$.

\begin{table}[htb]
\caption{Summary of the experimental results. \emph{\#} denotes the number of instances per benchmark, \emph{\% Reduction} denotes the percentage reduction of the state space (average $\pm$ standard deviation), and \emph{Speed-up} denotes the speed-up in computation time for UB compared to B.}
\rowcolors{1}{gray!10}{white}
\centering
\begin{tabular}{ m{0.22\textwidth} R{0.03\textwidth} R{0.12\textwidth} R{0.16\textwidth} R{0.16\textwidth} R{0.13\textwidth} }
  \toprule
  \rowcolor{white}
  \multicolumn{3}{c}{Benchmark} & \multicolumn{2}{c}{\% Reduction} & \multicolumn{1}{c}{\; Speed-up} \\
  \midrule
  Name (property) & \# & Max states & \multicolumn{1}{c}{\; B} & \multicolumn{1}{c}{\; UB} & \multicolumn{1}{c}{\; $\frac{\text{B}}{\text{UB}}$} \\
  \midrule
  brp (p1) & 12 & 5192 & $50.41 \pm \phantom{0} 1.15$ & $23.28 \pm \phantom{0} 3.86$ & $3.55 \pm 1.70$ \\
  brp (p2) & 12 & 5192 & $50.41 \pm \phantom{0} 1.15$ & $23.28 \pm \phantom{0} 3.86$ & $0.87 \pm 0.10$ \\
  brp (p4) & 12 & 5192 & $99.22 \pm \phantom{0} 0.44$ & $99.22 \pm \phantom{0} 0.44$ & $0.96 \pm 0.08$ \\
  crowds (positive) & 16 & 10633591 & $99.50 \pm \phantom{0} 0.91$ & $84.36 \pm \phantom{0} 6.25$ & $1.60 \pm 0.71$ \\
  egl (messagesA) & 4 & 156670 & $99.35 \pm \phantom{0} 0.04$ & $99.26 \pm \phantom{0} 0.11$ & $0.87 \pm 0.18$ \\
  egl (messagesB) & 4 & 156670 & $99.34 \pm \phantom{0} 0.05$ & $99.30 \pm \phantom{0} 0.08$ & $0.87 \pm 0.17$ \\
  egl (unfairA) & 4 & 156670 & $99.37 \pm \phantom{0} 0.03$ & $99.29 \pm \phantom{0} 0.09$ & $1.07 \pm 0.19$ \\
  egl (unfairB) & 4 & 156670 & $99.26 \pm \phantom{0} 0.04$ & $99.20 \pm \phantom{0} 0.08$ & $1.11 \pm 0.32$ \\
  herman (steps) & 7 & 32768 & $91.96 \pm \phantom{0} 8.33$ & $15.07 \pm 23.17$ & $0.74 \pm 0.32$ \\
  leader-sync (elected) & 24 & 1312334 & $96.25 \pm \phantom{0} 7.10$ & $96.25 \pm \phantom{0} 7.10$ & $1.39 \pm 0.50$ \\
  leader-sync (time) & 24 & 1312334 & $96.25 \pm \phantom{0} 7.10$ & $96.25 \pm \phantom{0} 7.10$ & $1.28 \pm 0.36$ \\
  nand (reliable) & 11 & 14123252 & $30.93 \pm 10.16$ & $20.44 \pm 11.38$ & $1.70 \pm 0.65$ \\
  oscillators (power) & 7 & 24311 & $0.00 \pm \phantom{0} 0.00$ & $0.00 \pm \phantom{0} 0.00$ & $2.94 \pm 1.17$ \\
  oscillators (time) & 7 & 24311 & $1.20 \pm \phantom{0} 2.57$ & $1.20 \pm \phantom{0} 2.57$ & $1.56 \pm 0.27$ \\
  \midrule
  \rowcolor{white}
  \textbf{Total} & \multicolumn{2}{l}{148} & $75.63 \pm 33.70$ & $65.17 \pm 39.78$ & $1.52 \pm 1.00$ \\
  \bottomrule
\end{tabular}
\label{table:aggregate}
\end{table}

\rowcolors{6}{gray!10}{white}
\begin{landscape}
\begin{center}
\begin{longtable}{ m{0.18\textwidth} m{0.08\textwidth} m{0.09\textwidth} R{0.11\textwidth} R{0.15\textwidth} R{0.09\textwidth} R{0.1\textwidth} R{0.15\textwidth} R{0.09\textwidth} R{0.1\textwidth} R{0.11\textwidth} }
\caption{The complete set of experimental results. \emph{Min} denotes the number of states in the minimised model, \emph{\% Red.} denotes the percentage reduction of the state space, \emph{Time} denotes the amount of time taken (in seconds) to compute the quotient, and \emph{Speed-up} denotes the speed-up in computation time for UB compared to B.}\label{table:complete-results}\\
  \toprule
  \multicolumn{4}{c}{Benchmark} & \multicolumn{3}{c}{\qquad B} & \multicolumn{4}{c}{\qquad UB} \\
  \midrule
  Name (property) & \multicolumn{2}{c}{Parameters} & States & Min & \% Red. & Time (s) & Min & \% Red. & Time (s) & Speed-up \\
  \midrule
  brp (p1) & N=16 & MAX=2 & 677 & 326 & 51.85 & 0.020 & 480 & 29.10 & 0.011 & 1.82 \\
   &  & MAX=3 & 886 & 439 & 50.45 & 0.020 & 670 & 24.38 & 0.011 & 1.82 \\
   &  & MAX=4 & 1095 & 552 & 49.59 & 0.023 & 860 & 21.46 & 0.013 & 1.77 \\
   &  & MAX=5 & 1304 & 665 & 49.00 & 0.021 & 1050 & 19.48 & 0.011 & 1.91 \\
   & N=32 & MAX=2 & 1349 & 646 & 52.11 & 0.028 & 960 & 28.84 & 0.009 & 3.11 \\
   &  & MAX=3 & 1766 & 871 & 50.68 & 0.034 & 1342 & 24.01 & 0.011 & 3.09 \\
   &  & MAX=4 & 2183 & 1096 & 49.79 & 0.041 & 1724 & 21.03 & 0.014 & 2.93 \\
   &  & MAX=5 & 2600 & 1321 & 49.19 & 0.051 & 2106 & 19.00 & 0.015 & 3.40 \\
   & N=64 & MAX=2 & 2693 & 1286 & 52.25 & 0.066 & 1920 & 28.70 & 0.013 & 5.08 \\
   &  & MAX=3 & 3526 & 1735 & 50.79 & 0.090 & 2686 & 23.82 & 0.016 & 5.63 \\
   &  & MAX=4 & 4359 & 2184 & 49.90 & 0.125 & 3452 & 20.81 & 0.020 & 6.25 \\
   &  & MAX=5 & 5192 & 2633 & 49.29 & 0.128 & 4218 & 18.76 & 0.022 & 5.82 \\
  brp (p2) & N=16 & MAX=2 & 677 & 326 & 51.85 & 0.020 & 480 & 29.10 & 0.020 & 1.00 \\
   &  & MAX=3 & 886 & 439 & 50.45 & 0.021 & 670 & 24.38 & 0.022 & 0.95 \\
   &  & MAX=4 & 1095 & 552 & 49.59 & 0.025 & 860 & 21.46 & 0.026 & 0.96 \\
   &  & MAX=5 & 1304 & 665 & 49.00 & 0.021 & 1050 & 19.48 & 0.026 & 0.81 \\
   & N=32 & MAX=2 & 1349 & 646 & 52.11 & 0.027 & 960 & 28.84 & 0.034 & 0.79 \\
   &  & MAX=3 & 1766 & 871 & 50.68 & 0.033 & 1342 & 24.01 & 0.039 & 0.85 \\
   &  & MAX=4 & 2183 & 1096 & 49.79 & 0.039 & 1724 & 21.03 & 0.047 & 0.83 \\
   &  & MAX=5 & 2600 & 1321 & 49.19 & 0.045 & 2106 & 19.00 & 0.054 & 0.83 \\
   & N=64 & MAX=2 & 2693 & 1286 & 52.25 & 0.060 & 1920 & 28.70 & 0.060 & 1.00 \\
   &  & MAX=3 & 3526 & 1735 & 50.79 & 0.068 & 2686 & 23.82 & 0.086 & 0.79 \\
   &  & MAX=4 & 4359 & 2184 & 49.90 & 0.089 & 3452 & 20.81 & 0.094 & 0.95 \\
   &  & MAX=5 & 5192 & 2633 & 49.29 & 0.086 & 4218 & 18.76 & 0.124 & 0.69 \\
  brp (p4) & N=16 & MAX=2 & 677 & 10 & 98.52 & 0.005 & 10 & 98.52 & 0.005 & 1.00 \\
   &  & MAX=3 & 886 & 12 & 98.65 & 0.010 & 12 & 98.65 & 0.009 & 1.11 \\
   &  & MAX=4 & 1095 & 14 & 98.72 & 0.009 & 14 & 98.72 & 0.010 & 0.90 \\
   &  & MAX=5 & 1304 & 16 & 98.77 & 0.007 & 16 & 98.77 & 0.007 & 1.00 \\
   & N=32 & MAX=2 & 1349 & 10 & 99.26 & 0.005 & 10 & 99.26 & 0.006 & 0.83 \\
   &  & MAX=3 & 1766 & 12 & 99.32 & 0.007 & 12 & 99.32 & 0.007 & 1.00 \\
   &  & MAX=4 & 2183 & 14 & 99.36 & 0.010 & 14 & 99.36 & 0.011 & 0.91 \\
   &  & MAX=5 & 2600 & 16 & 99.38 & 0.010 & 16 & 99.38 & 0.012 & 0.83 \\
   & N=64 & MAX=2 & 2693 & 10 & 99.63 & 0.008 & 10 & 99.63 & 0.008 & 1.00 \\
   &  & MAX=3 & 3526 & 12 & 99.66 & 0.011 & 12 & 99.66 & 0.011 & 1.00 \\
   &  & MAX=4 & 4359 & 14 & 99.68 & 0.011 & 14 & 99.68 & 0.012 & 0.92 \\
   &  & MAX=5 & 5192 & 16 & 99.69 & 0.014 & 16 & 99.69 & 0.014 & 1.00 \\
  \midrule
  crowds & CS=5 & TR=3 & 1198 & 41 & 96.58 & 0.010 & 240 & 79.97 & 0.009 & 1.11 \\
  (positive) &  & TR=4 & 3515 & 61 & 98.26 & 0.017 & 839 & 76.13 & 0.014 & 1.21 \\
   &  & TR=5 & 8653 & 81 & 99.06 & 0.031 & 2238 & 74.14 & 0.029 & 1.07 \\
   &  & TR=6 & 18817 & 101 & 99.46 & 0.057 & 5037 & 73.23 & 0.058 & 0.98 \\
   & CS=10 & TR=3 & 6563 & 41 & 99.38 & 0.020 & 770 & 88.27 & 0.020 & 1.00 \\
   &  & TR=4 & 30070 & 61 & 99.80 & 0.061 & 4619 & 84.64 & 0.056 & 1.09 \\
   &  & TR=5 & 111294 & 81 & 99.93 & 0.250 & 20018 & 82.01 & 0.137 & 1.82 \\
   &  & TR=6 & 352535 & 101 & 99.97 & 0.685 & 70067 & 80.12 & 0.313 & 2.19 \\
   & CS=15 & TR=3 & 19228 & 41 & 99.79 & 0.041 & 1600 & 91.68 & 0.043 & 0.95 \\
   &  & TR=4 & 119800 & 61 & 99.95 & 0.201 & 13599 & 88.65 & 0.151 & 1.33 \\
   &  & TR=5 & 592060 & 81 & 99.99 & 0.861 & 81598 & 86.22 & 0.455 & 1.89 \\
   &  & TR=6 & 2464168 & 101 & 100.00 & 4.014 & 387597 & 84.27 & 1.250 & 3.21 \\
   & CS=20 & TR=3 & 42318 & 41 & 99.90 & 0.081 & 2730 & 93.55 & 0.085 & 0.95 \\
   &  & TR=4 & 333455 & 61 & 99.98 & 0.438 & 30029 & 90.99 & 0.278 & 1.58 \\
   &  & TR=5 & 2061951 & 81 & 100.00 & 2.705 & 230228 & 88.83 & 1.071 & 2.53 \\
   &  & TR=6 & 10633591 & 101 & 100.00 & 18.360 & 1381377 & 87.01 & 7.062 & 2.60 \\
  \midrule
  egl (messagesA) & N=5 & L=2 & 33790 & 239 & 99.29 & 0.099 & 306 & 99.09 & 0.134 & 0.74 \\
   &  & L=4 & 74750 & 479 & 99.36 & 0.298 & 546 & 99.27 & 0.267 & 1.12 \\
   &  & L=6 & 115710 & 719 & 99.38 & 0.463 & 786 & 99.32 & 0.520 & 0.89 \\
   &  & L=8 & 156670 & 959 & 99.39 & 0.597 & 1026 & 99.35 & 0.821 & 0.73 \\
  egl (messagesB) & N=5 & L=2 & 33790 & 246 & 99.27 & 0.096 & 277 & 99.18 & 0.132 & 0.73 \\
   &  & L=4 & 74750 & 486 & 99.35 & 0.301 & 517 & 99.31 & 0.278 & 1.08 \\
   &  & L=6 & 115710 & 726 & 99.37 & 0.465 & 757 & 99.35 & 0.505 & 0.92 \\
   &  & L=8 & 156670 & 966 & 99.38 & 0.667 & 997 & 99.36 & 0.908 & 0.73 \\
  egl (unfairA) & N=5 & L=2 & 33790 & 229 & 99.32 & 0.091 & 287 & 99.15 & 0.099 & 0.92 \\
   &  & L=4 & 74750 & 469 & 99.37 & 0.215 & 527 & 99.29 & 0.194 & 1.11 \\
   &  & L=6 & 115710 & 709 & 99.39 & 0.341 & 767 & 99.34 & 0.370 & 0.92 \\
   &  & L=8 & 156670 & 949 & 99.39 & 0.715 & 1007 & 99.36 & 0.541 & 1.32 \\
  egl (unfairB) & N=5 & L=2 & 33790 & 270 & 99.20 & 0.069 & 312 & 99.08 & 0.104 & 0.66 \\
   &  & L=4 & 74750 & 550 & 99.26 & 0.281 & 592 & 99.21 & 0.202 & 1.39 \\
   &  & L=6 & 115710 & 830 & 99.28 & 0.424 & 872 & 99.25 & 0.379 & 1.12 \\
   &  & L=8 & 156670 & 1110 & 99.29 & 0.705 & 1152 & 99.26 & 0.548 & 1.29 \\
  \midrule
  haddad- & p=0.7 & N=20 & 41 & 41 & 0.00 & 0.001 & 41 & 0.00 & 0.000 & - \\
  monmege &  & N=100 & 201 & 201 & 0.00 & 0.016 & 201 & 0.00 & 0.001 & 16.00 \\
  (target) &  & N=300 & 601 & 601 & 0.00 & 0.034 & 601 & 0.00 & 0.001 & 34.00 \\
  \midrule
  herman (steps) &  & N=3 & 8 & 2 & 75.00 & 0.002 & 3 & 62.50 & 0.002 & 1.00 \\
   &  & N=5 & 32 & 4 & 87.50 & 0.002 & 23 & 28.13 & 0.002 & 1.00 \\
   &  & N=7 & 128 & 9 & 92.97 & 0.005 & 115 & 10.16 & 0.005 & 1.00 \\
   &  & N=9 & 512 & 23 & 95.51 & 0.017 & 495 & 3.32 & 0.019 & 0.89 \\
   &  & N=11 & 2048 & 63 & 96.92 & 0.047 & 2027 & 1.03 & 0.077 & 0.61 \\
   &  & N=13 & 8192 & 190 & 97.68 & 0.242 & 8167 & 0.31 & 0.512 & 0.47 \\
   &  & N=15 & 32768 & 612 & 98.13 & 1.819 & 32739 & 0.09 & 8.663 & 0.21 \\
  \midrule
  leader-sync & N=3 & K=2 & 26 & 8 & 69.23 & 0.002 & 8 & 69.23 & 0.001 & 2.00 \\
  (elected) &  & K=3 & 69 & 8 & 88.41 & 0.002 & 8 & 88.41 & 0.001 & 2.00 \\
   &  & K=4 & 147 & 8 & 94.56 & 0.002 & 8 & 94.56 & 0.002 & 1.00 \\
   &  & K=5 & 273 & 8 & 97.07 & 0.003 & 8 & 97.07 & 0.001 & 3.00 \\
   &  & K=6 & 459 & 8 & 98.26 & 0.004 & 8 & 98.26 & 0.003 & 1.33 \\
   &  & K=8 & 1059 & 8 & 99.24 & 0.006 & 8 & 99.24 & 0.005 & 1.20 \\
   & N=4 & K=2 & 61 & 10 & 83.61 & 0.001 & 10 & 83.61 & 0.001 & 1.00 \\
   &  & K=3 & 274 & 10 & 96.35 & 0.003 & 10 & 96.35 & 0.002 & 1.50 \\
   &  & K=4 & 812 & 10 & 98.77 & 0.004 & 10 & 98.77 & 0.003 & 1.33 \\
   &  & K=5 & 1933 & 10 & 99.48 & 0.006 & 10 & 99.48 & 0.006 & 1.00 \\
   &  & K=6 & 3962 & 10 & 99.75 & 0.010 & 10 & 99.75 & 0.009 & 1.11 \\
   &  & K=8 & 12400 & 10 & 99.92 & 0.017 & 10 & 99.92 & 0.015 & 1.13 \\
   & N=5 & K=2 & 141 & 12 & 91.49 & 0.001 & 12 & 91.49 & 0.002 & 0.50 \\
   &  & K=3 & 1050 & 12 & 98.86 & 0.007 & 12 & 98.86 & 0.005 & 1.40 \\
   &  & K=4 & 4244 & 12 & 99.72 & 0.015 & 12 & 99.72 & 0.010 & 1.50 \\
   &  & K=5 & 12709 & 12 & 99.91 & 0.018 & 12 & 99.91 & 0.017 & 1.06 \\
   &  & K=6 & 31383 & 12 & 99.96 & 0.034 & 12 & 99.96 & 0.025 & 1.36 \\
   &  & K=8 & 131521 & 12 & 99.99 & 0.113 & 12 & 99.99 & 0.074 & 1.53 \\
   & N=6 & K=2 & 335 & 14 & 95.82 & 0.004 & 14 & 95.82 & 0.003 & 1.33 \\
   &  & K=3 & 3759 & 14 & 99.63 & 0.015 & 14 & 99.63 & 0.009 & 1.67 \\
   &  & K=4 & 20884 & 14 & 99.93 & 0.025 & 14 & 99.93 & 0.021 & 1.19 \\
   &  & K=5 & 78784 & 14 & 99.98 & 0.112 & 14 & 99.98 & 0.052 & 2.15 \\
   &  & K=6 & 234210 & 14 & 99.99 & 0.126 & 14 & 99.99 & 0.135 & 0.93 \\
   &  & K=8 & 1312334 & 14 & 100.00 & 0.427 & 14 & 100.00 & 0.347 & 1.23 \\
  leader-sync & N=3 & K=2 & 26 & 8 & 69.23 & 0.002 & 8 & 69.23 & 0.002 & 1.00 \\
  (time) &  & K=3 & 69 & 8 & 88.41 & 0.002 & 8 & 88.41 & 0.001 & 2.00 \\
   &  & K=4 & 147 & 8 & 94.56 & 0.002 & 8 & 94.56 & 0.003 & 0.67 \\
   &  & K=5 & 273 & 8 & 97.07 & 0.004 & 8 & 97.07 & 0.003 & 1.33 \\
   &  & K=6 & 459 & 8 & 98.26 & 0.006 & 8 & 98.26 & 0.003 & 2.00 \\
   &  & K=8 & 1059 & 8 & 99.24 & 0.005 & 8 & 99.24 & 0.005 & 1.00 \\
   & N=4 & K=2 & 61 & 10 & 83.61 & 0.002 & 10 & 83.61 & 0.002 & 1.00 \\
   &  & K=3 & 274 & 10 & 96.35 & 0.004 & 10 & 96.35 & 0.003 & 1.33 \\
   &  & K=4 & 812 & 10 & 98.77 & 0.005 & 10 & 98.77 & 0.005 & 1.00 \\
   &  & K=5 & 1933 & 10 & 99.48 & 0.007 & 10 & 99.48 & 0.007 & 1.00 \\
   &  & K=6 & 3962 & 10 & 99.75 & 0.015 & 10 & 99.75 & 0.010 & 1.50 \\
   &  & K=8 & 12400 & 10 & 99.92 & 0.022 & 10 & 99.92 & 0.018 & 1.22 \\
   & N=5 & K=2 & 141 & 12 & 91.49 & 0.003 & 12 & 91.49 & 0.002 & 1.50 \\
   &  & K=3 & 1050 & 12 & 98.86 & 0.006 & 12 & 98.86 & 0.005 & 1.20 \\
   &  & K=4 & 4244 & 12 & 99.72 & 0.015 & 12 & 99.72 & 0.015 & 1.00 \\
   &  & K=5 & 12709 & 12 & 99.91 & 0.024 & 12 & 99.91 & 0.020 & 1.20 \\
   &  & K=6 & 31383 & 12 & 99.96 & 0.042 & 12 & 99.96 & 0.032 & 1.31 \\
   &  & K=8 & 131521 & 12 & 99.99 & 0.158 & 12 & 99.99 & 0.085 & 1.86 \\
   & N=6 & K=2 & 335 & 14 & 95.82 & 0.006 & 14 & 95.82 & 0.004 & 1.50 \\
   &  & K=3 & 3759 & 14 & 99.63 & 0.015 & 14 & 99.63 & 0.014 & 1.07 \\
   &  & K=4 & 20884 & 14 & 99.93 & 0.029 & 14 & 99.93 & 0.028 & 1.04 \\
   &  & K=5 & 78784 & 14 & 99.98 & 0.113 & 14 & 99.98 & 0.060 & 1.88 \\
   &  & K=6 & 234210 & 14 & 99.99 & 0.137 & 14 & 99.99 & 0.150 & 0.91 \\
   &  & K=8 & 1312334 & 14 & 100.00 & 0.444 & 14 & 100.00 & 0.375 & 1.18 \\
  \midrule
  nand (reliable) & N=20 & K=1 & 78332 & 39982 & 48.96 & 0.704 & 49047 & 37.39 & 0.682 & 1.03 \\
   &  & K=2 & 154942 & 102012 & 34.16 & 1.998 & 125657 & 18.90 & 1.552 & 1.29 \\
   &  & K=3 & 231552 & 164042 & 29.16 & 3.874 & 202267 & 12.65 & 3.255 & 1.19 \\
   &  & K=4 & 308162 & 226072 & 26.64 & 6.227 & 278877 & 9.50 & 5.792 & 1.08 \\
   & N=40 & K=1 & 1004862 & 559699 & 44.30 & 30.777 & 629473 & 37.36 & 18.398 & 1.67 \\
   &  & K=2 & 2003082 & 1443559 & 27.93 & 114.795 & 1627693 & 18.74 & 65.469 & 1.75 \\
   &  & K=3 & 3001302 & 2327419 & 22.45 & 202.328 & 2625913 & 12.51 & 139.294 & 1.45 \\
   &  & K=4 & 3999522 & 3211279 & 19.71 & 347.300 & 3624133 & 9.39 & 245.758 & 1.41 \\
   & N=60 & K=1 & 4717592 & 2726958 & 42.20 & 335.940 & 2959373 & 37.27 & 110.240 & 3.05 \\
   &  & K=2 & 9420422 & 7046448 & 25.20 & 996.621 & 7662203 & 18.66 & 459.294 & 2.17 \\
   &  & K=3 & 14123252 & 11365938 & 19.52 & 1999.478 & 12365033 & 12.45 & 771.213 & 2.59 \\
  \midrule
  oscillators & T=6 & N=3 & 57 & 57 & 0.00 & 0.001 & 57 & 0.00 & 0.001 & 1.00 \\
  (power) &  & N=6 & 463 & 463 & 0.00 & 0.004 & 463 & 0.00 & 0.002 & 2.00 \\
   & T=8 & N=6 & 1717 & 1717 & 0.00 & 0.010 & 1717 & 0.00 & 0.004 & 2.50 \\
   &  & N=8 & 6436 & 6436 & 0.00 & 0.034 & 6436 & 0.00 & 0.010 & 3.40 \\
   & T=10 & N=6 & 5006 & 5006 & 0.00 & 0.030 & 5006 & 0.00 & 0.008 & 3.75 \\
   &  & N=7 & 11441 & 11441 & 0.00 & 0.045 & 11441 & 0.00 & 0.013 & 3.46 \\
   &  & N=8 & 24311 & 24311 & 0.00 & 0.089 & 24311 & 0.00 & 0.020 & 4.45 \\
  oscillators & T=6 & N=3 & 57 & 53 & 7.02 & 0.002 & 53 & 7.02 & 0.002 & 1.00 \\
  (time) &  & N=6 & 463 & 460 & 0.65 & 0.009 & 460 & 0.65 & 0.005 & 1.80 \\
   & T=8 & N=6 & 1717 & 1710 & 0.41 & 0.035 & 1710 & 0.41 & 0.021 & 1.67 \\
   &  & N=8 & 6436 & 6429 & 0.11 & 0.090 & 6429 & 0.11 & 0.058 & 1.55 \\
   & T=10 & N=6 & 5006 & 4999 & 0.14 & 0.105 & 4999 & 0.14 & 0.059 & 1.78 \\
   &  & N=7 & 11441 & 11434 & 0.06 & 0.148 & 11434 & 0.06 & 0.095 & 1.56 \\
   &  & N=8 & 24311 & 24304 & 0.03 & 0.347 & 24304 & 0.03 & 0.219 & 1.58 \\
  \bottomrule
\end{longtable}
\end{center}
\end{landscape}

\end{document}